\documentclass[noinfoline]{imsart}

\RequirePackage[OT1]{fontenc}
\RequirePackage{amsthm,amsmath}
\RequirePackage[numbers]{natbib}
\RequirePackage[colorlinks,citecolor=blue,urlcolor=blue]{hyperref}

\usepackage{verbatim}

\allowdisplaybreaks

\usepackage{fullpage}

\usepackage{algorithm,algorithmic}

\DeclareFontEncoding{FMS}{}{}
\DeclareFontSubstitution{FMS}{futm}{m}{n}
\DeclareFontEncoding{FMX}{}{}
\DeclareFontSubstitution{FMX}{futm}{m}{n}
\DeclareSymbolFont{fouriersymbols}{FMS}{futm}{m}{n}
\DeclareSymbolFont{fourierlargesymbols}{FMX}{futm}{m}{n}
\DeclareMathDelimiter{\hsnorm}{\mathord}{fouriersymbols}{152}{fourierlargesymbols}{147}

\newcommand{\hs}{\big\hsnorm}

\RequirePackage[OT1]{fontenc}
\RequirePackage{amsthm,amsmath}

\RequirePackage{natbib}
\RequirePackage[colorlinks,citecolor=blue,urlcolor=blue]{hyperref}

\usepackage{verbatim}

\allowdisplaybreaks

\newcommand {\tr}{\mbox{tr}}

\usepackage{subfigure}
\usepackage{layout}

\usepackage{dsfont}

\usepackage[mathscr]{eucal}
\usepackage[toc,page]{appendix}
\usepackage{mathrsfs}
\usepackage{color}
\usepackage{pifont}
\usepackage{bm}
\usepackage{latexsym}
\usepackage{amsmath}
\usepackage{amsthm}
\usepackage{amsfonts}
\usepackage{amssymb}
\usepackage{epsfig}
\usepackage{graphicx}
\usepackage{hyperref}
\usepackage{multirow}

\newtheorem{theorem}{Theorem}

\newtheorem{corollary}{Corollary}
\newtheorem{lemma}{Lemma}
\newtheorem{definition}{Definition}
\newtheorem{remark}{Remark}

\input xy
\xyoption{all}

\startlocaldefs
\numberwithin{equation}{section}
\theoremstyle{plain}

\endlocaldefs

\begin{document}

\begin{frontmatter}

\title{\large{Hybrid Regularisation of Functional Linear Models}} 
\runtitle{Hybrid Regularisation}

\begin{aug}
\author{\fnms{Anirvan} \snm{Chakraborty}\ead[label=e1]{anirvan.chakraborty@epfl.ch}}
\and
\author{\fnms{Victor~M.} \snm{Panaretos}\ead[label=e2]{victor.panaretos@epfl.ch}}

\runauthor{A. Chakraborty and V.~M. Panaretos}

\affiliation{\'Ecole Polytechnique F\'ed\'erale de Lausanne}

\address{Institut de Math\'ematiques,\\
\'Ecole Polytechnique F\'ed\'erale de Lausanne\\
\printead{e1,e2}}
\end{aug}

\begin{abstract}
\noindent We consider the problem of estimating the slope function in a functional regression with a scalar response and a functional covariate. This central problem of functional data analysis is well known to be ill-posed, thus requiring a regularised estimation procedure. The two most commonly used approaches are based on spectral truncation or Tikhonov regularisation of the empirical covariance operator. In principle, Tikhonov regularisation is the more canonical choice. Compared to spectral truncation, it is robust to eigenvalue ties, while it attains the optimal minimax rate of convergence in the mean squared sense, and not just in a concentration probability sense. In this paper, we show that, surprisingly, one can strictly improve upon the performance of the Tikhonov estimator in finite samples by means of a linear estimator, while retaining its stability and asymptotic properties by combining it with a form of spectral truncation. Specifically, we construct an estimator that additively decomposes the functional covariate by projecting it onto two orthogonal subspaces defined via functional PCA; it then applies Tikhonov regularisation to the one component, while leaving the other component unregularised. We prove that when the covariate is Gaussian, this hybrid estimator uniformly improves upon the MSE of the Tikhonov estimator in a non-asymptotic sense, effectively rendering it inadmissible. This domination is shown to also persist under discrete observation of the covariate function. The hybrid estimator is linear, straightforward to construct in practice, and with no computational overhead relative to the standard regularisation methods. By means of simulation, it is shown to furnish sizeable gains even for modest sample sizes. 
\end{abstract}

\begin{keyword}[class=MSC]
\kwd[Primary ]{62M}\kwd{62G}
\kwd[; secondary ]{62J07}
\kwd{62M15}\kwd{15A29}
\end{keyword}

\begin{keyword}
\kwd{admissibility}\kwd{condition index}\kwd{functional data analysis}\kwd{ill-posed problem}\kwd{mean integrated squared error}\kwd{principal component analysis}\kwd{rate of convergence}\kwd{ridge regression}\kwd{spectral truncation}\kwd{Tikhonov regularisation}
\end{keyword}

\end{frontmatter}

\tableofcontents

\newpage
\section{Introduction}
\label{1}

\subsection{Functional Linear Models and their regularisation}
 For a real-valued response $y$, and a random functional covariate $X$ taking values in a separable Hilbert space ${\cal H}$ with inner product $\langle\cdot,\cdot\rangle$, the functional linear regression model with scalar response is given by 
\begin{equation}\label{scalar_FLM}
y = \alpha + \langle X,\beta\rangle +\epsilon,
\end{equation} 
where $\epsilon$ is a scalar random measurement error term that is assumed to be zero mean and independent of the covariate $X$ (\citet{RS05}, \citet{HK12}, \citet{HE15}). The so-called slope parameter $\beta\in\mathcal{H}$ is typically the object of primary importance. 
The statistical task is then to estimate $\beta$ on the basis of an i.i.d. sample of pairs $\{(y_{i},X_{i})\}_{i=1}^n$ generated according to the model \eqref{scalar_FLM}. The classical least squares approach of estimating $\beta$ results in the normal equation 
\begin{equation}\label{normal_equation}
\hat{\mathscr{K}}\,\beta = \hat{C},
\end{equation}
where $\hat{\mathscr{K}}$ is the empirical covariance operator of the $\{X_i\}$ and $\hat{C}$ is the empirical cross-covariance of the $\{y_{i}\}$ and the $\{X_{i}\}$. Since the population operator $\mathscr{K}$ is a trace-class operator, its empirical version $\hat{\mathscr{K}}$ is so too, for all $n$. Its failure to be boundedly invertible gives rise to an ill-posed inverse problem, which is usually solved by regularising the inverse of $\hat{\mathscr{K}}$. The regularisation strategies employed in the functional data analysis literature can be broadly categorised into two classes\footnote{Though there exist even more general descriptions that include the two categories as special cases, see \citet{CMS07}.}: \emph{sieve methods} and \emph{penalised methods}.

In the method of sieves (\citet{grenander1981}), one selects an orthonormal basis $\{\varphi_k\}$ of $\mathcal{H}$, and projects the covariates $\{X_i\}$ and the slope function $\beta$ onto the subspace spanned by the first $r$ basis elements. If the basis $\{\varphi_k\}$ and truncation level $r$ are selected judiciously, one obtains a stable multivariate regression problem, with a small amount of bias. In terms of asymptotics, one must let $K\rightarrow\infty$ but regulate its growth as a function of $n$,  in order to guarantee that the regression problem remain stable for each $n$. The challenge here is to determine a ``good" basis $\{\varphi_k\}$ whose first $r$ elements provide a parsimonious representation simultaneously for $X$ \emph{and} $\beta$ -- but of course $\beta$ is unknown, and worse still, does not have any intrinsic relationship to $X$ that might prove the existence of such a basis. The typical choice is to rely on the Karhunen-Lo\`eve expansion of $X$ and to choose $\{\varphi_k\}$ to be the basis of eigenfunctions of $\mathscr{K}$, and has evolved in the most popular choice in practice. This approach, known as \emph{spectral truncation} or \emph{PCA regression}, uses a sieve that is optimally adapted to $X$,  but makes no reference to $\beta$, thus potentially not providing a good approximation of $\beta$. The thought is, however, that those characteristics of $\beta$ that do not correlate well with $X$ (and thus are not well expressed in the Karhunen-Lo\`eve basis) are worth sacrificing, as they cannot be well-recovered through the model \eqref{scalar_FLM} anyway. In practice, the Karhunen-Lo\`eve basis is estimated from the data, using a functional principal component analysis (\citet[Chapter 10]{RS05}; \citet{FerratyReg}; \citet{CFF02}; \citet{HP06}; \citet{yao2005}). 

Penalised methods, on the other hand, regularise the problem by placing restrictions directly on $\beta$, most often by penalising the degree of roughness of $\beta$ by means of a suitable norm. They lead to constrained least squares problems, instead of the unconstrained problem \eqref{normal_equation}, that are well posed. Estimation procedures following this paradigm have been studied, for instance, by \citet{ramsay1991}, \citet{marx1999}, \citet{CFS03}, \citet{LH07} and \citet{CKS09} to name only a few (see also \citet{RS05}). Depending on the nature of the penalty, the estimator can be represented in a finite a basis $\{\varphi_k\}$, typically a spline basis corresponding to a curvature penalty, and the functional regression translates to a multivariate ridge regression problem. The approach can be elegantly formulated within a reproducing kernel Hilbert space framework, which directly translates the infinite dimensional and ill-posed problem into a finite dimensional and well-posed one (\citet{YC10}). A general description of penalised methods can be viewed as instances of \emph{Tikhonov regularisation} (\citet{tikhonov1977}), where the sum of squares objective function $SS(\beta)=\sum_{i=1}^{n} (y_{i} - \bar{y} - \langle X_{i} - \bar{X},\beta\rangle)^{2}$ is penalised by the addition of a multiple of some norm of $\beta$.

\subsection{Tikhonov vs Spectral regularisation and Our Contributions}
For a regularisation parameter $\rho > 0$, the Tikhonov regularised estimator is defined as 
$$\hat{\beta}_{TR} = \hat{\mathscr{K}}_{\rho}^{-1}\hat{C}$$  
where $\hat{\mathscr{K}}_{\rho} = \hat{\mathscr{K}} + \rho\mathscr{I}$ with $\mathscr{I}$ the identity operator on ${\cal H}$. This estimator is the direct analogue of the ridge estimator (\citet{HK70}) in classical multivariate linear regression with correlated regressors and is the minimiser of the penalized least squares problem
\begin{equation}\label{tikhonov_objective}
\min_{\beta \in {\cal H}}\left\{ \frac{1}{n} \sum_{i=1}^{n} (y_{i} - \bar{y} - \langle X_{i} - \bar{X},\beta\rangle)^{2} + \rho \|\beta\|^{2}\right\}.
\end{equation} 
On the other hand, given $r\in\mathbb{N}$, the spectral truncation estimator is defined as
$$\hat{\beta}_{ST} = \left(\sum_{j=1}^{r}\hat\lambda_j^{-1}\hat\phi_j\otimes\hat\phi_j\right)\hat{C},$$  
where $\{\hat\lambda_j,\hat\phi_j\}$ is the spectrum of the empirical covariance $\hat{\mathscr{K}}$.
In one of the landmark papers on functional regression, \citet{HH07} established general error properties of both spectral truncation and of Tikhonov regularisation. Their results indicate the latter approach to be a more canonical avenue mainly due to two reasons, namely its minimaxity and its stability. More specifically, \citet{HH07} established minimax optimal rates of convergence for the Tikhonov estimator in the mean squared error (MSE) sense, but only obtained minimax rates for spectral truncation estimator in the weaker concentration probability sense\footnote{In an earlier paper, \citet{HHN06}  proved that a \textit{modified, non-linear (thresholded)} version of the spectral truncation estimator \emph{can} attain the minimax MSE rate.  The modification is done to ensure that the resulting estimator does not take very large values (see Theorem 5 in Appendix A.2 in \cite{HHN06}, pp. 116--117 in that paper, and the discussion after Theorem 1 in \citet{HH07}). Unfortunately, this modified estimator depends on arbitrary constants whose choices are subjective and so the estimator is not practically feasible. In fact, it remains unknown whether the original spectral truncation estimator attains the minimax rate of convergence in the mean squared sense at all. The non-linear modification appears to be necessary for the proof techniques of \citet{HH07} to work.} (see also Remark \ref{rem3}). Secondly, \citet{HH07} demonstrated that the spectral truncation approach can suffer from instabilities when the eigenvalues of $\mathscr{K}$ are not well-spaced, while Tikhonov regularisation is immune to such effects (see also \citet{YC10} for similar arguments). In the well-spaced regime, neither approach is seen to dominate the other, except in the rather special circumstance where: 
\begin{itemize}
\item[(i)] the leading eigenvalues of $\mathscr{K}$ are well-conditioned,
\item[] \emph{and} 
\item[(ii)] $\beta$ mostly contained in the span of the leading eigenfunctions of $\mathscr{K}$. 
\end{itemize}
When (i) and (ii) occur simultaneously, spectral truncation prevails, as the problem essentially reduces to a well-conditioned multivariate regression, in no need of regularisation.

The question this paper considers is the following: is it possible to leverage this last observation in order to improve upon the more canonical Tikhonov approach, by combining it in part with the projection rationale of spectral truncation? The answer is an unequivocal yes, and surprisingly the improvement is realised by a lienar estimator: a simple combination of the two approaches yields a hybrid estimator that remains linear, straightforward to compute, and provably strictly improves upon the Tikhonov estimator in a non-asymptotic sense (i.e. not a rate but an exact MSE sense). We note in passing that while adaptive estimators of the slope function have been considered (see e.g. \cite{CJ10}, \cite{CJ12}), they typically introduce a thresholding of the spectral estimator, thus becoming non-linear (and, there has not been any theoretical comparison of the non-asymptotic MSEs of these estimators to those of the Tikhonov or spectral truncation estimator).

The hybrid estimator we introduce (defined rigorously in Section \ref{2}) projects onto a finite dimensional subspace $\mathcal{H}_r$ (as would a sieve estimator), but rather than discard the residual component (i.e. the projection onto the orthogonal complement $\mathcal{H}_r^{\perp}$), it retains it, and applies a ridge regularisation to that and only that. The dimension $r<\infty$ of $\mathcal{H}_r$ does not grow with respect to $n$, and only the ridge parameter $\rho$ is sample-size dependent.  We demonstrate in Section \ref{2} that choosing $\mathcal{H}_r$ to be \emph{any} eigenspace of $\mathscr{K}$ of dimension $r<n$ yields an estimator that  attains the minimax MSE rate but in fact strictly improves upon the Tikhonov approach for large enough samples, uniformly over $\beta$  (Theorem \ref{thm-orcl1} and Corollary \ref{cor-orcl1}). Section \ref{2.2} exploits this observation in order to construct a practically feasible hybrid estimator, by empirical construction of $\mathcal{H}_r$. Section \ref{6} establishes that the empirically constructed estimator also yields the same strict improvement and rates. The practicalities of constructing the estimator are discussed in detail Subsection \ref{2.2.1}, where recommendations are given on how to best choose the dimension $r$ of $\mathcal{H}_r$ as well as the ridge parameter $\rho$. The key message is that one ought to select $r$ so that the first $r$ eigenvalues yield a mild condition index ($\lambda_1/\lambda_r$).  Section \ref{dis} then treats the case where the covariates are functions observed discretely on a grid, and shows that even in this case, the hybrid estimator still enjoys the same properties and improvement in mean squared error over the Tikhonov estimator in this setup. In Section \ref{3}, we conduct a simulation study that illustrates that one can make considerable performance gains in practice, even for moderate sample sizes.  The proofs of our formal results are collected in Section \ref{5}. First, though, Section \ref{notation} introduces some notation that will be employed throughout of the paper.

\section{Preliminaries}\label{notation} 

As mentioned in the introduction, $\mathcal{H}$ will be a real separable Hilbert space, assumed infinite dimensional, with inner product $\langle\cdot,\cdot\rangle:\mathcal{H}\times\mathcal{H}\rightarrow\mathbb{R}$, and induced norm $\|\cdot\|:\mathcal{H}\to[0,\infty)$. Given a linear operator $\mathscr{A}:\mathcal{H}\rightarrow\mathcal{H}$, we will denote its adjoint operator by $\mathscr{A}^*$, its Moore-Penrose generalised inverse by $\mathscr{A}^{-}$, and its inverse by $\mathscr{A}^{-1}$, provided the latter is well-defined. The operator,  Hilbert-Schmidt, and nuclear norms will respectively be
$$\hs \mathscr{A}\hs_{\infty}=\sup_{\|h\|=1}\|\mathscr{A}h\|,\quad \hs \mathscr{A} \hs =\sqrt{\mbox{trace}\left(\mathscr{A}^*\mathscr{A}\right)},\quad  \hs \mathscr{A} \hs_1=\mbox{trace}\left(\sqrt{\mathscr{A}^*\mathscr{A}}\right).$$ 
It is well-known that 
$$ \hs \mathscr{A}\hs_{\infty} \leq \hs \mathscr{A}\hs \leq \hs \mathscr{A}\hs_{1}$$
for any bounded linear operator satisfying $\hs \mathscr{A}\hs_{1} < \infty$. The identity operator on $\mathcal{H}$ will be denoted by $\mathscr{I}$. For a pair of elements $f,g\in \mathcal{H}$, the tensor product $f\otimes g:\mathcal{H}\to\mathcal{H}$ will be defined as the linear operator
$$(f\otimes g)u=\langle g,u\rangle f,\qquad u\in \mathcal{H}.$$
The same notation will be used to denote the tensor product between two operators, so that for operators $\mathscr{A}$, $\mathscr{B}$, and $\mathscr{G}$, one has
$$(\mathscr{A}\otimes \mathscr{B})\,\mathscr{G}=\mbox{trace}\left(\mathscr{B}^*\mathscr{G}\right)\mathscr{A}.$$
Given an estimator $\delta$ of the slope parameter $\beta\in\mathcal{H}$, we define the \emph{Mean Square Error (MSE)} in order to probe the performance of $\delta$,
$$\mathrm{MSE}(\delta) = \mathbb{E}\|\delta - \beta\|^2=\hs \mathbb{E}\{ (\delta - \beta) \otimes (\delta - \beta)\}\hs_1.$$ 
In the usual setting of $\mathcal{H}=L_2[0,1]$, this risk function reduces to the so-called \emph{Mean Integrated Squared Error},
$$\mathbb{E}\left\{\int_0^1(\delta(x)-\beta(x))^2\mathrm{d}x\right\},$$
but of course our results will be valid for any separable Hilbert space $\mathcal{H}$. We also note that all our results hold verbatim if instead of the MSE, we consider the (weaker) Hilbert-Schmidt norm of $\mathbb{E}\{ (\delta - \beta) \otimes (\delta - \beta)\}$ as the risk function.

\section{Motivation: Multivariate Plus Functional Regressors}\label{2} To motivate our hybrid estimator, let $X = Y + Z$, where:
\begin{enumerate}
\item The covariance operator $\mathscr{K}_{1}$ of $Y$ is of finite rank $r$.

\item The random elements $Y$ and $Z$ are uncorrelated.

\item The eigenspaces of $\mathscr{K}_1$ are orthogonal to those of the covariance $\mathscr{K}_2$ of $Z$.
\end{enumerate}
Observe that such a decomposition always exists by the Karhunen-Lo\`eve theorem. The heuristic now is that if $Y$ and $Z$ were observable, we would have a model with two orthogonal and uncorrelated regressors, one multivariate, and one functional,
$$y= \langle X,\beta\rangle+\epsilon=\langle Y,\beta_1\rangle + \langle Z,\beta_2\rangle +\epsilon,$$
with $\beta_1$ and $\beta_2$ being the projections of $\beta$ on the (orthogonal) ranges of $\mathscr{K}_1$ and $\mathscr{K}_2$. So, if $Y$ has a well-conditioned covariance $\mathscr{K}_{1}$, then instead of regularising the entire spectrum of the covariance operator $\mathscr{K}$ of $X$, one should carry out two separate regressions: a multivariate one, without regularisation, corresponding to the well-conditioned $\mathscr{K}_{1}$; and a functional one, with Tikhonov regularisation, corresponding to the ill-conditioned $\mathscr{K}_{2}$. The point here is that functional regression is not ill-conditioned as a result of poor design (as in the multivariate case when covariates may be correlated); it is ill-conditioned by the mere fact that it is infinite dimensional. But, in general, we should be able to extract a subspace on which it is well-conditioned.

We now turn to transforming our heuristic to a concrete result. Write the spectra of the two covariance operators $\mathscr{K}_{1}$ and $\mathscr{K}_{2}$ as 
$$\mathscr{K}_1=\sum_{j=1}^{r}\lambda_{j1}\phi_{j1}\otimes\phi_{j1}\qquad \&\qquad \mathscr{K}_2=\sum_{j\geq 1}\lambda_{j2}\phi_{j2}\otimes\phi_{j2}.$$
Define 
$$\beta_{1} = \left(\sum_{j=1}^{r}\phi_{j1}\otimes\phi_{j1}\right)\beta= \mathscr{P}_{1}\beta \qquad\&\qquad \beta_{2} =\left(\sum_{j\geq 1}\phi_{j2}\otimes\phi_{j2}\right)\beta = \mathscr{P}_{2}\beta$$
 to be the projections of $\beta$ into the eigenspaces of $\mathscr{K}_{1}$ and $\mathscr{K}_{2}$. Note that we must have $\beta = \beta_{1} + \beta_{2}$ for identifiablity so we henceforth assume that $\mbox{range}(K)=\mathcal{H}$. Thus, $\mathscr{P}_{1} + \mathscr{P}_{2} = \mathscr{I}$, where $\mathscr{I}$ is the identity operator on ${\cal H}$, and indeed $\langle{X},\beta\rangle = \langle{Y},\beta_{1}\rangle + \langle{Z},\beta_{2}\rangle$. Now consider the following modification of the population version of the Tikhonov penalised least squares problem
\begin{equation}\label{hybrid_objective}
\min_{\beta_{1}, \beta_{2} \in {\cal H}} \left\{ \mathbb{E}\big[y - \mathbb{E}[y] - \langle Y - \mathbb{E}[Y],\beta_{1}\rangle - \langle Z - \mathbb{E}[Z],\beta_{2}\rangle\big]^{2} + \rho \|\beta_{2}\|^{2}\right\}, 
\end{equation}
where we only penalize the part of the norm of $\beta$ that corresponds to $\beta_{2}$. Direct calculation in the above minimisation problem yields the unique minimiser 
$$\beta_{min} = \mathscr{K}_{1}^{-}C_{1} + \mathscr{K}_{\rho,2}^{-}C_{2},$$ 
where
$$C_{1} = \mathbb{E}[yY] - \mathbb{E}[y]\mathbb{E}[Y],\quad C_{2} = \mathbb{E}[yZ] - \mathbb{E}[y]\mathbb{E}[Z],\quad\mathscr{K}_{\rho,2} = \mathscr{K}_{2} + \rho\mathscr{P}_{2}.$$

 The form of the minimiser $\beta_{min}$ motivates the following definition of a hybrid regularised estimator of $\beta$ in the oracle case. Assume that a sample $(y_{i}, X_{i})$ is available, and the oracle reveals the decompositions $X_{i} = Y_{i} + Z_{i}$ into uncorrelated orthogonal components, as well as their respective covariances ($\mathscr{K}_1,\mathscr{K}_2$). Define a hybrid estimator as
\begin{eqnarray}
\widetilde{\beta}_{HR} = \mathscr{K}_{1}^{-}\widetilde{C}_{1} + \mathscr{K}_{\rho,2}^{-}\widetilde{C}_{2},   \label{HR-orcl1}
\end{eqnarray}
where 
\begin{eqnarray*}
\widetilde{C}_{1} &=& n^{-1} \sum_{i=1}^{n} (y_{i} - \overline{y})(Y_{i} - \overline{Y}),\qquad\mbox{with }\overline{Y} = n^{-1} \sum_{i=1}^{n} Y_{i}\\
\widetilde{C}_{2} &=& n^{-1} \sum_{i=1}^{n} (y_{i} - \overline{y})(Z_{i} - \overline{Z}),\qquad\mbox{with }\overline{Z} = n^{-1} \sum_{i=1}^{n} Z_{i}.
\end{eqnarray*}
On the other hand, the oracle version of the Tikhonov estimator is 
\begin{eqnarray}
\widetilde{\beta}_{TR} = \mathscr{K}_{\rho}^{-1}\hat{C},   \label{HR-orcl2}
\end{eqnarray}
where $\mathscr{K}_{\rho} = \mathscr{K} + \rho\mathscr{I}$ and $\hat{C} = n^{-1} \sum_{i=1}^{n} (y_{i} - \overline{y})(X_{i} - \overline{X})$. Our first theorem shows that, since the hybrid estimator makes explicit use of the additional information (the decomposition $(Z_i,Y_i)$ instead of just $X_i$), it \emph{improves} upon the Tikhonov estimator.

\begin{theorem} \label{thm-orcl1}
Let $X=Y+Z$, where $Y$ and $Z$ are uncorrelated random elements with $\mathbb{E}(\|Y\|^{4})<\infty$ and $\mathbb{E}(\|Z\|^{4})<\infty$. Assume that the eigenspaces of the respective covariances $\mathscr{K}_{1}$ and $\mathscr{K}_{2}$ of $Y$ and $Z$ are orthogonal. Further, assume that the $\langle X,\phi_{j}\rangle$'s are independent, where $\phi_{j}$'s are the eigenfunctions of $\mathscr{K}$. Then, 
\begin{itemize}
\item[(a)] For any fixed $\rho > 0$, $\mathrm{MSE}(\widetilde{\beta}_{TR}) > \mathrm{MSE}(\widetilde{\beta}_{HR})$ for all sufficiently large $n$. 
\item[(b)] If we choose $\rho = \rho(n) \sim cn^{-\gamma}$ for some $\gamma \in (0,1/2]$ and a constant $c > 0$, we have 
\begin{eqnarray*}
n^{2\gamma} \{ \mathrm{MSE}(\widetilde{\beta}_{TR}) - \mathrm{MSE}(\widetilde{\beta}_{HR}) \} > B(n) + o(1).
\end{eqnarray*}
Here, $B(n)$ converges to a positive constant if at least one of $\langle \beta,\phi_{j1}\rangle$, $j=1,2,\ldots,r$ is non-zero, else it converges to zero as $n \rightarrow \infty$.
\end{itemize}
\end{theorem}

\noindent The independence assumption in the above theorem obviously holds for Gaussian processes, and for any process whose Karhunen-L\`oeve expansion has independent coefficients. It can be relaxed to requiring that $\mathbb{E}(\prod_{u=1}^{4}\langle X,\phi_{j_{u}}\rangle^{l_{u}}) = \prod_{u=1}^{4} \mathbb{E}(\langle X,\phi_{j_{u}}\rangle^{l_{u}})$ for $l_{u}$'s satisfying $1 \leq l_{u} \leq 4$ and $\sum_{u=1}^{4} l_{u} \leq 4$. This can be viewed as a ``pseudo-independence'' condition, and similar assumptions have been considered for analysis of high-dimensional data (see, e.g., Sec. 3 in \cite{CQ10}, Sec. 4 in \cite{BS96}). As a direct consequence of part (b) of the above theorem, we have the corollary:

\begin{corollary}  \label{cor-orcl1}
Under the conditions of Theorem \ref{thm-orcl1} and in the setup of part (b) of that theorem, if at least one of $\langle \beta,\phi_{j1}\rangle$, $j=1,2,\ldots,r$ is non-zero, there exists a constant $c_{0} > 0$ such that
$$\mathrm{MSE}(\widetilde{\beta}_{TR}) - \mathrm{MSE}(\widetilde{\beta}_{HR}) > c_{0}n^{-2\gamma}$$
for all sufficiently large $n$. If $\langle \beta,\phi_{j1}\rangle$ is uniformly zero for $1\leq j\leq r$, the two MSE norms are asymptotically equal.
\end{corollary}

Thus, in the oracle case, as long as $\beta$ is at least partially expressed by the principal components of $Y$, then the hybrid regularisation estimator will improve on the Tikhonov estimator -- whether $\rho$ is held fixed, or allowed to decay polynomially in $n$, as one usually does. The next section deals with carrying over this improvement to an empirically feasible estimator.

\section{The Hybrid Estimator}\label{hybrid_section}

 In practice, the components $Y$ and $Z$ are unobservable, and their covariance operators $\mathscr{K}_{1}$ and $\mathscr{K}_{2}$ are unknown. Still, we can replace them by their empirical versions, and consider whether we can still improve upon the Tikhonov estimator by the hybrid approach when doing so. We will focus on the case where $Y$ is the projection of $X$ onto its first $r$ principal components, and $Z=X-Y$, since this case admits straightforward empirical versions of all the quantities involved.  We first define the empirical version of the hybrid estimator (Subsection \ref{2.2}); next we establish its superiority to Tikhonov regularisation (Subsection \ref{6}); and then, we discuss its (straightforward) practical implementation (Subsection \ref{2.2.1}).

\subsection{Definition}\label{2.2}

\noindent Given an i.i.d. sample $X_{1},\ldots, X_{n}$ ditributed as $X$, denote their empirical covariance and its spectrum as
$$\hat{\mathscr{K}}=\frac{1}{n}\sum_{i=1}^{n}X_i\otimes X_i=\sum_{i=1}^{n}\hat{\lambda}_j\hat\phi_j\otimes\hat\phi_j.$$
Now define
$$\hat{Y}_{i} = \sum_{j=1}^{r} \langle X_{i},\hat{\phi}_{j}\rangle \hat{\phi}_{j}=\hat{\mathscr{P}}_{1}X_{i}\quad\mbox{and}\quad\hat{Z}_{i} = X_{i} - \hat{Y}_{i}=\hat{\mathscr{P}}_{2}X_{i}, \qquad i = 1,\ldots,n,$$ 
with $\hat{\mathscr{P}}_{1}=\sum_{i=1}^{r}\hat{\phi}_i\otimes\hat{\phi}_i$ the projection onto $\mathrm{span}\{\hat\phi_1,\ldots,\hat\phi_r\}$ and $\hat{\mathscr{P}}_{2} = \mathscr{I} - \hat{\mathscr{P}}_{1}$.   Let us denote the sample covariance operators of the $\hat{Y}_{i}$'s and the $\hat{Z}_{i}$'s as 
$$\hat{\mathscr{K}}_{1}=\frac{1}{n}\sum_{i=1}^{n}\hat{Y}_i\otimes \hat{Y}_i\quad\&\quad\hat{\mathscr{K}}_{2}=\frac{1}{n}\sum_{i=1}^{n}\hat{Z}_i\otimes \hat{Z}_i,$$
respectively.  Finally, let $\hat{C}_{1}$ and $\hat{C}_2$ be the empirical covariances between the proxy regressors and the responses,
$$\widehat{C}_{1} = \frac{1}{n} \sum_{i=1}^{n} (y_{i} - \overline{y})\left(Y_{i} - \frac{1}{n}\sum_{i=1}^{n}\hat{Y}_i\right),\quad\widehat{C}_{2} = \frac{1}{n} \sum_{i=1}^{n} (y_{i} - \overline{y})\left(Z_{i} - \frac{1}{n}\sum_{i=1}^{n}\hat{Z}_i\right).$$
From \eqref{HR-orcl1}, it is clear that a natural definition of the hybrid regularisation estimator of $\beta$ is:

\begin{definition}[Hybrid Regularisation Estimator] The hybrid regularisation estimator $\hat\beta_{HR}$ is defined as the solution to the penalised least squares problem 
\begin{eqnarray*}
\min_{\beta_{1}, \beta_{2} \in {\cal H}} n^{-1} \sum_{i=1}^{n} (y_{i} - \bar{y} - \langle \hat{Y}_{i} - \bar{Y},\beta_{1}\rangle - \langle \hat{Z}_{i} - \bar{Z},\beta_{2}\rangle)^{2} + \rho \|\beta_{2}\|^{2},
\end{eqnarray*}
where $\beta = \beta_{1} + \beta_{2}$ with $\beta_{1} = \hat{\mathscr{P}}_{1}\beta$, $\beta_{2} = \hat{\mathscr{P}}_{2}\beta$, $\bar{Y} = n^{-1}\sum_{i=1}^{n} \hat{Y}_{i}$ and $\bar{Z} = n^{-1}\sum_{i=1}^{n} \hat{Z}_{i}$. It is explicitly given by
\begin{equation}
\hat{\beta}_{HR} = \hat{\mathscr{K}}_{1}^{-}\hat{C}_{1} + \hat{\mathscr{K}}_{\rho,2}^{-}\hat{C}_{2},   \label{HR-data1}
\end{equation}
where $\hat{\mathscr{K}}_{\rho,2} = \hat{\mathscr{K}}_{2} + \rho\hat{\mathscr{P}}_{2}$.

\end{definition}

\begin{remark} \label{rem4}
A priori, there is no reason why one should choose $\hat Y$ to be the projection of $X$ onto the \textbf{first} $r$ eigenfunctions: \textbf{any} collection of $r$ eigenfunctions could be considered. In principle, we choose $r$ eigenfunctions of $\hat{K}$ that: (1) yield a component $\hat Y$ with a well-conditioned covariance operator $\hat{K}_1$; (2) and capture a large part of the norm of $\beta$. Since $\beta$ is unknown in practice, (2) is impossible to control. Still, the whole point of fitting a functional linear model is the understanding that $\beta$ correlates well with the signal rather than the noise in $X$, and thus this correlation is expected to be carried by the leading principal components of $X$, explaining our choice of selecting the first $r$ components, subject to a well-conditioning restriction. 
\end{remark}

\subsection{Theoretical Properties}
\label{6}

We now turn to prove that both the gain in efficiency and the minimaxity observed in the oracle setup also carry over to the practically feasible hybrid estimator. We will make use of the following assumptions.

\begin{itemize}
\item[\textbf{(A1)}] $X$ is a centered Gaussian process. 
\item[\textbf{(A2)}] The eigenvalues $\lambda_{1} > \lambda_{2} > \ldots$ of $\mathscr{K}$ are all positive. Also, for constants $\alpha > 1$, $0 < c < C$ and $j_{0} \geq 1$, we have $cj^{-\alpha} \leq \lambda_{j} \leq Cj^{-\alpha}$ for all $j \geq j_{0}$. 
\item[\textbf{(A3)}] For constants $d > 0$, $\eta > 1/2$ and $j_{0} \geq 1$, we have $|\langle \beta,\phi_{j}\rangle| \leq dj^{-\eta}$ for all $j \geq j_{0}$. 
\end{itemize}

Condition (A1) can be relaxed to accommodate other distributions. In that case, we would need to assume that $\mathbb{E}(\|X\|^{16}) < \infty$, that $\mathbb{E}(\langle X,\phi_{j}\rangle^{4}) \leq C\lambda_{j}^{2}$ for all $j \geq 1$ and a constant $C > 0$, and the pseudo-independence condition similar to that mentioned earlier, i.e. that $\mathbb{E}(\prod_{u=1}^{4}\langle X,\phi_{j_{u}}\rangle^{l_{u}}) = \prod_{u=1}^{4} \mathbb{E}(\langle X,\phi_{j_{u}}\rangle^{l_{u}})$ for  $l_{u}$'s satisfying $1 \leq l_{u} \leq 4$ and $\sum_{u=1}^{4} l_{u} \leq 16$. These in particular hold if $X$ has the representation $X = \sum_{j=1}^{\infty} \lambda_{j}^{1/2}V_{j}\phi_{j}$, where the $V_{j}$'s are i.i.d. zero mean random variables with finite moments (cf. the discussion before Corollary \ref{cor-orcl1}). Conditions (A2) and (A3) have been used by \citet{HH07} to obtain the rate of convergence of the Tikhonov regularisation estimator in terms of its integrated mean squared error. The interplay between $\alpha$ and $\eta$ determines the degree of difficulty of estimating $\beta$. Clearly, the larger the value of $\eta$, the easier is the estimation problem. If $\alpha$ is large, then the distribution of $X$ becomes almost finite dimensional. In that case, if $\eta$ is small, then the estimation problem is difficult if there are important components of $\beta$, namely, $\langle\beta,\phi_{j}\rangle$ corresponding to small values of $\lambda_{j}$. This is because there is very little information on $X$ in those directions, and thus those components of $\beta$ will be difficult to estimate. We will later see exactly how $\alpha$ and $\eta$ determine the precision in estimating $\beta$. Condition (A2) is sufficient to ensure that $\mathbb{E}(\|X\|^{2}) = \sum_{j=1}^{\infty} \lambda_{j} < \infty$, which in turn implies that $X$ is a (tight) random element in ${\cal H}$. Condition (A3) ensures that $\|\beta\|^{2} < \infty$.

Our first result compares the efficiency of the oracle version of the hybrid and Tikhonov estimator to that of their empirical version for a fixed value of the ridge parameter $\rho$.
\begin{theorem}  \label{thm-data1}
Suppose that conditions (A1)--(A3) hold, and $\alpha < 2\eta$. Then,
\begin{eqnarray}  \label{HR-data2}
&& |\mathrm{MSE}(\hat{\beta}_{HR}) - \mathrm{MSE}(\widetilde{\beta}_{HR})| \\
&& = O(1)\left[\left\{\frac{1}{n\rho^{1+\frac{1}{\alpha}}} + \rho^{m}\right\}^{1/2}\left(\frac{1}{n\rho^{1+\frac{1}{\alpha}}}\right)^{1/2} + \frac{1}{n\rho^{1+\frac{1}{\alpha}}}\right] \nonumber
\end{eqnarray}
for any sequence $\rho \rightarrow 0$ satisfying $n\rho^{2} \rightarrow \infty$ as $n \rightarrow \infty$. Further,  
$$\mathrm{MSE}(\hat{\beta}_{HR}) = O(1)\left\{\frac{1}{n\rho^{1+\frac{1}{\alpha}}} + \rho^{m}\right\}$$
as $n \rightarrow \infty$.  Here $m = (2\eta-1)/\alpha$ or $m = 2$ according as $\alpha > \eta - 1/2$ or $\alpha < \eta - 1/2$. Moreover, analogous rates of convergence also hold for $|\mathrm{MSE}(\hat{\beta}_{TR}) - \mathrm{MSE}(\widetilde{\beta}_{TR})|$ and $\mathrm{MSE}(\hat{\beta}_{TR})$.
\end{theorem}

 The terms $n^{-1}\rho^{-1-1/\alpha}$ and $\rho^{m}$ in the expression of $\mathrm{MSE}(\hat{\beta}_{HR})$ given in the above theorem clearly show the effects of the variance and the bias terms, respectively, in the estimation of $\beta$. It also reveals that only the bias is affected by the rate of decay of the $\langle\beta,\phi_{j}\rangle$'s but not the variance. This is expected because the variability in the estimation of $\beta$ should purely depend on the fluctuations in $X$, which depends on the rate of decay of the eigenvalues of the covariance operator $\mathscr{K}$ of $X$. 

As a corollary, we can obtain a similar comparison theorem when the ridge parameter $\rho$ decays with $n$:
\begin{corollary}  \label{cor-data1}
Consider the setup of Theorem \ref{thm-data1}. Let $c > 0$ be a fixed constant. Then, 

\begin{eqnarray}  \label{HR-data3}
&& \mathrm{MSE}(\hat{\beta}_{HR}) \nonumber \\
&& = 
\begin{cases}
O(n^{-(2\eta-1)/(\alpha+2\eta)}) &\text{if $\eta - 1/2 < \alpha < 2\eta$ and $\rho \sim cn^{-\frac{\alpha}{\alpha+2\eta}}$} \\
O(n^{-2\alpha/(3\alpha+1)}) &\text{if $\alpha < \eta - 1/2$ and $\rho \sim cn^{-\frac{\alpha}{3\alpha+1}}$}
\end{cases}
\end{eqnarray}
as $n \rightarrow \infty$. Further, the same rates of convergence also hold for $\mathrm{MSE}(\hat{\beta}_{TR})$.
\end{corollary}

 The above corollary gives the rates of convergences of the hybrid regularisation estimator in terms of its mean squared error under different regimes determined by $\alpha$ and $\eta$. These regimes correspond to the different degrees of difficulty of the estimation problem in the functional linear regression setting. The rates of decay of $\rho$ to zero are chosen so as to optimize the rates of convergence of the MSEs. 

\begin{remark}\label{rem1} Note that the asymptotic rate of convergence of $\mathrm{MSE}(\hat{\beta}_{TR})$ was proved in Theorem 2 in \citet{HH07} under the restriction that $\alpha < \eta + 1/2$ and $\rho \sim cn^{-\alpha/(\alpha+2\eta)}$. The above corollary reveals that this upper bound on the values of the decay rate of the eigenvalues of $X$ can be relaxed. Further, the same rate of convergence is in fact true for a wider class of values of $\alpha$ and $\eta$ so long as $\eta - 1/2 < \alpha < 2\eta$. Note that \citet{HH07} did not require $\alpha > \eta - 1/2$. 
\end{remark}

\begin{remark} \label{rem2} \citet{HH07} showed that the rate of convergence of $\mathrm{MSE}(\hat{\beta}_{TR})$ is optimal in a minimax sense under the conditions of Theorem \ref{thm-data1} when $1 < \alpha < \eta + 1/2$ and $\rho \sim cn^{-\alpha/(\alpha+2\eta)}$. From Corollary \ref{cor-data1}, Remark \ref{rem1} and the proof of  equation (3.11) in \citet{HH07}, it follows that the hybrid regularisation estimator also enjoys the same minimax optimal rate of convergence for the same choice of regularisation parameter in the regime $\max(1,\eta - 1/2) < \alpha < 2\eta$. 
\end{remark} 

\begin{remark} \label{rem3} The spectral truncation estimator $\widehat{\beta}_{ST}$ 
studied by \citet{HH07} is known to satisfy $\mathrm{MSE}(\widehat{\beta}_{ST}) > \delta n^{-(2\eta-1)/(\alpha+2\eta)}$ for some $\delta > 0$ and sufficiently large $n$ and that $\mathrm{ISE}(\widehat{\beta}_{ST}) = O_{p}(n^{-(2\eta-1)/(\alpha+2\eta)})$ under appropriate conditions including $1 < \alpha < 2\eta - 2$ (see Theorem 1 in \citet{HH07}). This rate is also the minimax rate of convergence in a concentration probability sense. Now, it follows from Corollary \ref{cor-data1} that when $\eta - 1/2 < \alpha < 2\eta$, we have $\mathrm{ISE}(\widehat{\beta}_{HR}) = O_{p}(n^{-(2\eta-1)/(\alpha+2\eta)})$. In particular, when $\lambda_{j} \sim cj^{-\alpha}$ for all large $j$ (so that both condition (A2) in our paper and condition (3.2) in \citet{HH07} are satisfied) and when $\max(1,\eta - 1/2) < \alpha < 2\eta-2$, it follows that both of these two estimators have the same minimax rate of convergence in the concentration probability sense. Note that it is unknown whether the spectral truncation estimator will attain the minimax rate of convergence in the MSE sense like the hybrid and the Tikhonov estimators  discussed in Remark \ref{rem2}.
\end{remark}

Theorem \ref{thm-data1} and Corollary \ref{cor-orcl1} set the stage for our main result, showing that the hybrid estimator can improve upon the Tikhonov estimator in a non-asymptotic sense, even in the empirical case:

\begin{theorem}  \label{thm-data2}
Suppose that the conditions of Theorem \ref{thm-data1} hold. Let $c > 0$ be a fixed constant and $\rho \sim cn^{-\varepsilon}$ for some $\varepsilon > 0$. Also assume that at least one of $\langle \beta,\phi_{j}\rangle$, $j=1,2,\ldots,r$, is non-zero. Then, there exists a constant $\kappa_{0} > 0$ such that 
$$\mathrm{MSE}(\hat{\beta}_{TR}) - \mathrm{MSE}(\hat{\beta}_{HR}) > \kappa_{0}n^{-2\varepsilon}$$
for all sufficiently large $n$ if $\varepsilon < \alpha/(5\alpha-2\eta+2)$ in case $\eta - 1/2 < \alpha < 2\eta$ or if $\varepsilon < \alpha/(3\alpha+1)$ in case $\alpha < \eta - 1/2$.
\end{theorem}

Although the hybrid estimator and the Tikhonov estimator enjoy the same rate of convergence by Theorem \ref{thm-data1}, the latter is effectively rendered \textit{inadmissible} by the hybrid estimator for a broad range of choices of $\rho$, including choices arbitrarily close to the optimal one (as in Corollary \ref{cor-data1}) -- and this is true for all sample sizes above a threshold. It is illustrated in the simulations study in Section \ref{3}, that this improvement can be sizeable, even for modest sample sizes. Moreover, it is interesting to note that we can attain this improvement regardless of the choice of $r$ may be, even for $r=1$ (provided, of course, that $\langle\beta,\phi_{1}\rangle \neq 0$ as the theorem requires).

The proof of the Theorem reveals that the determining factor in the inadmissibility of the Tikhonov estimator is the larger bias component compared to the hybrid estimator (see also equation (1.4) in the proof of the oracle case provided in the Appendix). An important requirement is that the choice of $Y$ to be such that $\beta$ is at least partially expressed by the eigenfunctions of $Y$. Of course, how large the sample size has to be will depend on $r$ and also depend on the condition number of the covariance operator of $Y$. The latter is determined by the relative magnitudes of the eigenvalues of $Y$, equivalently, the relative importance of the associated eigenfunctions in explaining the variation in $X$.

\subsection{Computational Aspects}
\label{2.2.1}
Algorithm \ref{algo} provides a the step-by-step construction of the hybrid estimator. In summary, our recommendation is to fix an $r$ by the condition index approach discussed in in Remark \ref{rem4}, and to then choose $\rho$ by cross-validation (as in standard Tikhonov regularisation, see e.g. \citet{YC10}). Going through the steps in Algorithm \ref{algo}, one case see that there is no computational overhead or added complexity relative to the construction of a spectral truncation or Tikhonov estimator. An alternate, slightly more complex albeit fully automated procedure would be to use a double cross-validation for choosing both $r$ and $\rho$.

\begin{algorithm}[t]
\caption{Construction of the Hybrid Estimator}
\label{algo}
\begin{enumerate}

\item[(Step 0)] Determine the eigenvalues $\widehat{\lambda}_{j}$ and eigenfunctions $\widehat{\phi}_{j}$ of $\widehat{\mathscr{K}}$.

\item[(Step 1)] Fix a condition number $L$ and choose $r$  using the eigenvalues of $\widehat{\mathscr{K}}$ as 
$$r=\sup\left\{j \geq 1 : \left({\hat{\lambda}_{1}}/{\hat{\lambda}_{j}}\right)^{1/2} \leq L\right\}.$$

\item[(Step 2)] Set $\hat{Y}_{i} = \sum_{j=1}^{r} \langle X_{i},\hat{\phi}_{j}\rangle \hat{\phi}_{j}$ and $\widehat{Z}_{i}=X-\hat{Y}_i$, and compute
$$\widehat{C}_{1} = \frac{1}{n} \sum_{i=1}^{n} (y_{i} - \overline{y})\left(Y_{i} - \frac{1}{n}\sum_{i=1}^{n}\hat{Y}_i\right)\quad\&\quad\widehat{C}_{2} = \frac{1}{n} \sum_{i=1}^{n} (y_{i} - \overline{y})\left(Z_{i} - \frac{1}{n}\sum_{i=1}^{n}\hat{Z}_i\right).$$

\item[(Step 3)]  For the chosen $r$, choose $\rho$ by generalized cross-validation.

\item[(Step 4)]  Use this value of $\rho$ and the value of $r$ obtained in Step 1 to compute $\widehat{\beta}_{HR}$ using \eqref{HR-data1}, 
$$ \widehat{\beta}_{HR} = \sum_{j\leq r} \widehat{\lambda}_{j}^{-1} \langle\widehat{C}_{1},\widehat{\phi}_{j}\rangle\widehat{\phi}_{j} + \sum_{j > r} (\widehat{\lambda}_{j} + \rho)^{-1} \langle\widehat{C}_{2},\widehat{\phi}_{j}\rangle\widehat{\phi}_{j}. $$

\end{enumerate}
\end{algorithm}

It is worth remarking that one of the widely-used methods for choosing finitely many principal components of $X$ is to select the number required to capture the bulk of the trace of the empirical covariance operator -- a typical choice of threshold is that of $85\%$ of the total variation in $X$ (see \citet{Joll02} and \citet{RS05}). We shall later see in the simulation studies in Section \ref{3} that this choice is far from optimal, as it makes no reference to the condition number of the resulting multivariate regression. 

Our counterproposal on choosing $r$ guarantees that the covariance operator of $\hat{Y}$ is \emph{well-conditioned} -- the whole point of the hybrid estimator is to extract a component of the regression that does not need regularisation, after all, and such components are in no way connected with the cumulative variance explained. Condition indices and their maximum, which is called the condition number, are well-known in the classical multivariate regression setup as indicators of the degree of collinearity among the covariates, and more generally in numerical analysis as a measure of the instability of a linear problem. A rule-of-thumb is that a condition number $\leq 30$ indicates well-posedness (see, e.g., \citet{Hock03}). An alternative way to choose $L$ could be to consider a plot of the empirical condition indices and look for the ``elbow''. With a pre-fixed $L$, it is obvious why the choice of $r$ should be large when the eigenvalues decay slowly, and why it should be more conservative when they decay fast. Furthermore, since $\sup_{j \geq 1} |\hat{\lambda}_{j} - \lambda_{j}| \rightarrow 0$ in probability as $n \rightarrow \infty$, it follows that $\sup\{j \geq 1 : [\hat{\lambda}_{1}/\hat{\lambda}_{j}]^{1/2} \leq L\} \rightarrow \sup\{j \geq 1 : [\lambda_{1}/\lambda_{j}]^{1/2} \leq L\}$ in probability as $n \rightarrow \infty$, i.e., $r$ is chosen consistently by this procedure.

\section{The Case of Discretely Observed Functions}
\label{dis}

For data in a function space, say, $L_{2}[0,1]$, it may happen that instead of observing the entire curve $X$, one can only observe it on a  grid, say, 
$$0 \leq t_{1} < t_{2} < \ldots < t_{m} \leq 1.$$
Thus, the regressor at hand is an $m$-dimensional vector 
$$X^{(m)} = (X(t_{1}),X(t_{2}),\ldots,X(t_{m}))'.$$ 
In this setup, an approximation of the functional linear model considered in \eqref{scalar_FLM} is
\begin{eqnarray} \label{discrete_FLM}
y = \alpha + m^{-1} \sum_{p=1}^{m} X(t_{p})\beta(t_{p}) + \epsilon.
\end{eqnarray}
We define $\beta^{(m)} = (\beta(t_{1}),\beta(t_{2}),\ldots,\beta(t_{m}))'$. This setup of discretely observed data is closely related to the \textit{time-sampling model} considered by \citet{AW12} or the \textit{common design model} considered by \citet{CY11} with the difference that we do not consider measurement errors in the discrete observations of the $X_i$. 

In the discretely sampled setup considered above, the oracle and the empirical hybrid regularisation estimators of $\beta^{(m)}$ are defined analogously and are denoted by $\widetilde{\beta}_{HR}^{(m)}$ and $\widehat{\beta}_{HR}^{(m)}$, respectively. Similarly, the oracle and the empirical Tikhonov estimators are denoted by $\widetilde{\beta}_{TR}^{(m)}$ and $\widehat{\beta}_{TR}^{(m)}$, respectively. 

In order to state results analogous to Theorems \ref{thm-data1} and \ref{thm-data2}, we need to assume the following modifications of assumptions (A2) and (A3). We denote the eigenvalue-eigenvector pairs of the covariance matrix of $X^{(m)}/\sqrt{m}$ by $(\lambda_{j}^{(m)},\phi_{j}^{(m)})$ for $j = 1,2,\ldots,m$. 
\begin{itemize}
\item[\textbf{(A2')}] Suppose that $\lambda_{1}^{(m)} > \ldots > \lambda_{m}^{(m)} > 0$. Also, for  constants $\alpha > 1$, $0 < c' < C'$ and $j_{0}' \geq 1$, we have $c'j^{-\alpha} \leq \lambda_{j}^{(m)} \leq C'j^{-\alpha}$ for all $j_{0}' \leq j \leq m$ when $m$ is sufficiently large. \\
\item[\textbf{(A3')}] For constants $d' > 0$, $\eta' > 1/2$ and $j_{0}' \geq 1$, we have $m^{-1/2}|\langle \beta^{(m)},\phi_{j}^{(m)}\rangle| \leq d'\{j^{-\eta'} + m^{-1}\}$ for all $j_{0}' \leq j \leq m$ when $m$ is sufficiently large. \
\end{itemize}
In assumption (A3'), the parameter $\eta'$ is some function of the parameters $\alpha$ and $\eta$ that appear in assumptions (A2) and (A3) earlier. The two components in the inequality in assumption (A3') may be respectively interpreted as the contribution at the functional level and the error due to discretization. For instance, when $X$ is a standard Brownian motion, and if $\beta$ lies in its RKHS and satisfies assumption (A3), then $\eta' = \alpha$ for $\alpha \leq \eta$ and $\eta' = \eta$ for $\eta < \alpha < 2\eta -1$ (see the Appendix for a proof of this fact). Note that the condition $\alpha < 2\eta -1$ is needed to ensure that $\beta$ lies in the RKHS of the standard Brownian motion. Also, using the arguments in \citet{AW12}, it can be shown that condition (A2') holds in this case (see Appendix A in the Supplementary Material of \citet{AW12}). 

Theorem \ref{discrete-thm-data1} now shows that, even when we have discretely observed data, the hybrid and the Tikhonov estimators enjoy the same properties as their fully functional counterparts provided that the grid size grows to infinity sufficiently fast.
\begin{theorem} \label{discrete-thm-data1}
Suppose that conditions (A1), (A2') and (A3') hold, and $\alpha < 2\eta'$. Also assume that $m > \rho^{-2}$. Then,
\begin{eqnarray}  \label{discrete-HR-data2}
&& m^{-1}|\mathrm{MSE}(\hat{\beta}_{HR}^{(m)}) - \mathrm{MSE}(\widetilde{\beta}_{HR}^{(m)})| \\
&& = O(1)\left[\left\{\frac{1}{n\rho^{1+\frac{1}{\alpha}}} + \rho^{M}\right\}^{1/2}\left(\frac{1}{n\rho^{1+\frac{1}{\alpha}}}\right)^{1/2} + \frac{1}{n\rho^{1+\frac{1}{\alpha}}}\right] \nonumber
\end{eqnarray}
for any sequence $\rho \rightarrow 0$ satisfying $n\rho^{2} \rightarrow \infty$ as $n \rightarrow \infty$. Further,  
$$m^{-1}\mathrm{MSE}(\hat{\beta}_{HR}^{(m)}) = O(1)\left\{\frac{1}{n\rho^{1+\frac{1}{\alpha}}} + \rho^{M}\right\}$$
as $n \rightarrow \infty$.  Here $M = (2\eta'-1)/\alpha$ or $M = 2$ according as $\alpha > \eta' - 1/2$ or $\alpha < \eta' - 1/2$. Moreover, analogous rates of convergence also hold for $m^{-1}|\mathrm{MSE}(\hat{\beta}_{TR}^{(m)}) - \mathrm{MSE}(\widetilde{\beta}_{TR}^{(m)})|$ and $m^{-1}\mathrm{MSE}(\hat{\beta}_{TR}^{(m)})$. Thus,
\begin{eqnarray*}  \label{discrete-HR-data3}
&& m^{-1}\mathrm{MSE}(\hat{\beta}_{HR}^{(m)}) \nonumber \\
&& = 
\begin{cases}
O(n^{-(2\eta'-1)/(\alpha+2\eta')}) &\text{if $\eta' - 1/2 < \alpha < 2\eta'$ and $\rho \sim cn^{-\frac{\alpha}{\alpha+2\eta'}}$} \\
O(n^{-2\alpha/(3\alpha+1)}) &\text{if $\alpha < \eta' - 1/2$ and $\rho \sim cn^{-\frac{\alpha}{3\alpha+1}}$}
\end{cases}
\end{eqnarray*}
as $n \rightarrow \infty$. Further, the same rates of convergence also hold for $m^{-1}\mathrm{MSE}(\hat{\beta}_{TR}^{(m)})$.
\end{theorem}

Finally, our last result shows that, similar to the case of perfect functional observations,  the hybrid estimator outperforms the Tikhonov estimator for sufficiently large sample sizes and suitably chosen regularisation even when observations are discrete.

\begin{theorem}  \label{discrete-thm-data2}
Suppose that the conditions of Theorem \ref{discrete-thm-data1} hold. Let $c > 0$ be a fixed constant and $\rho \sim cn^{-\varepsilon}$ for some $\varepsilon > 0$. Also assume that at least one of $\langle \beta,\phi_{j}\rangle$, $j=1,2,\ldots,r$, is non-zero. Then, there exists a constant $\theta_{0} > 0$ such that 
$$m^{-1}\{\mathrm{MSE}(\hat{\beta}_{TR}^{(m)}) - \mathrm{MSE}(\hat{\beta}_{HR}^{(m)})\} > \theta_{0}n^{-2\varepsilon}$$
for all sufficiently large $n$ if $\varepsilon < \alpha/(5\alpha-2\eta'+2)$ in case $\eta' - 1/2 < \alpha < 2\eta'$ or if $\varepsilon < \alpha/(3\alpha+1)$ in case $\alpha < \eta' - 1/2$.
\end{theorem}

 The proof of Theorem \ref{discrete-thm-data2} can be developed in the same way as that of the proof of Theorem \ref{thm-data2} and is thus omitted.

\section{Simulation Study}
\label{3}

 We now turn to the assessment of the practical performance of the hybrid regularisation estimator relative to the Tikhonov estimator by means of a simulation study. To this aim, we shall  consider the same simulation framework considered in \citet{HH07} and \citet{YC10}. Take ${\cal H} = L_{2}[0,1]$, the space of square-integrable real functions on the interval $[0,1]$, with the usual inner product. Let $X$ be defined via its Karhunen-Lo\`eve expansion as 
$$X = \sum_{j=1}^{50} \gamma_{j}Z_{j}\phi_{j},$$ 
with the $Z_{j}$'s being i.i.d. uniform random variables on $[-3^{1/2},3^{1/2}]$, $\phi_{1}(t) = 1$ and $\phi_{j}(t) = 2^{1/2}\cos(j{\pi}t)$ for $t \in [0,1]$. Further, $\gamma_{j} = (-1)^{j+1}j^{-\alpha/2}$ for $j \geq 1$, and we choose  $\alpha$  to either be equal to $1.1$ or $2$. These two values of $\alpha$ correspond to slow and fast decays of the eigenvalues of $X$. Let $b_{1} = 1$ and $b_{j} = 4(-1)^{j+1}j^{-2}$ for $j =2,3,\ldots,50$. We have chosen three different kinds of slope function: (a) $\beta = \beta_{1} = \sum_{j=1}^{50} b_{j}\phi_{j}$, (b) $\beta = \beta_{2} = \sum_{j=1}^{5}b_{j}\phi_{j}$, and (c) $\beta = \beta_{3} = \sum_{j=6}^{50}b_{j}\phi_{j}$. Note that in cases (b) and (c) above, $\beta$ is expressed by two mutually orthogonal subcollections of eigenfunctions of $X$. We have considered these two choices of $\beta$ to study how the parsimony of $\beta$ in fewer or more eigenfunctions of $X$ influences the performance of the hybrid estimator. 
The sample size chosen is $n = 100$. The distribution of the error variable $\epsilon$ in the functional regression model is standard Gaussian. The $X$'s are evaluated at $50$ equispaced grid points in $[0,1]$. All the estimated mean squared errors are averaged over $1000$ Monte-Carlo replications. 
\begin{figure}
\vspace{-0.2in}
\begin{center}
{
\includegraphics[scale=0.5]{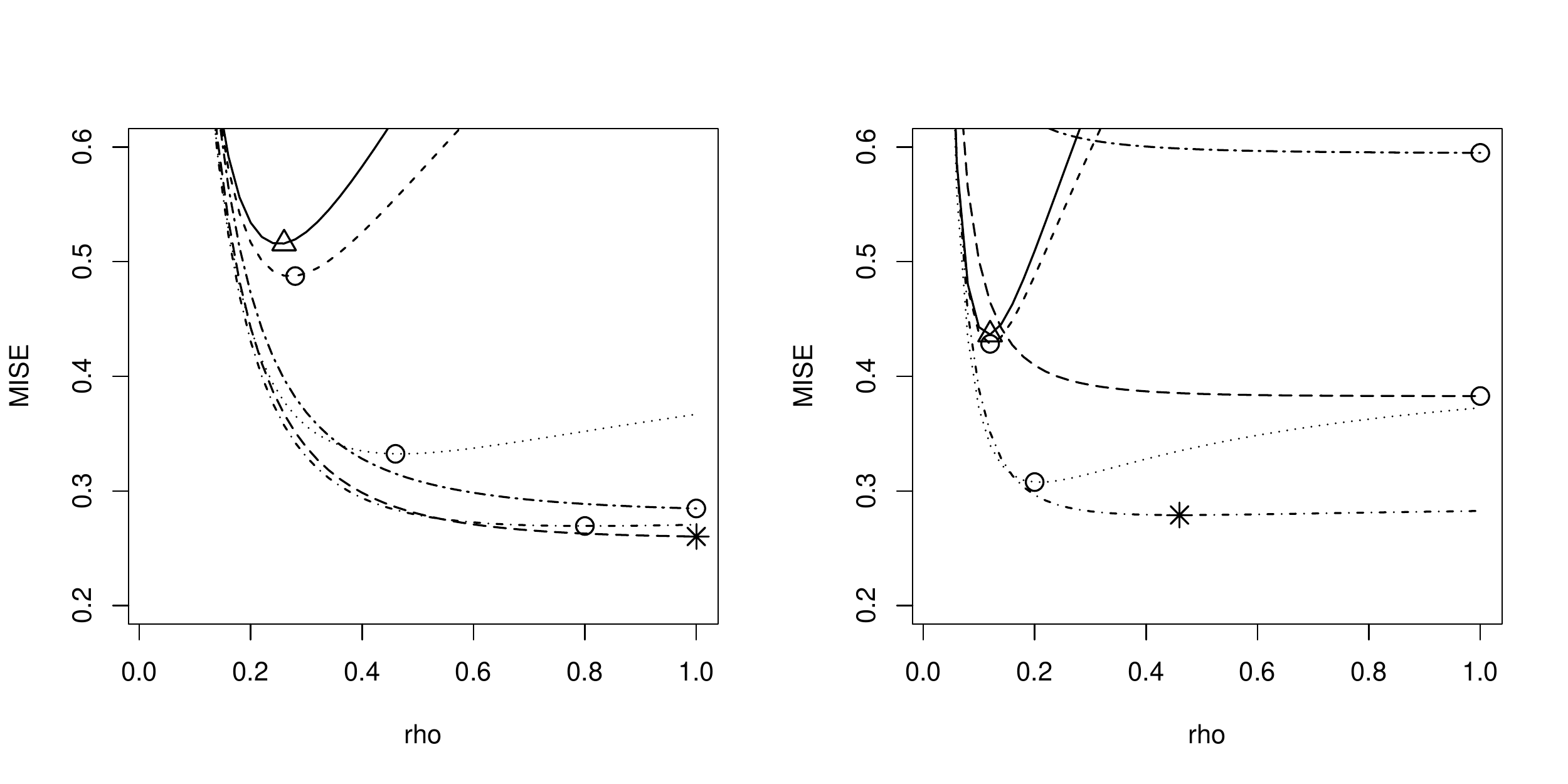}
\vspace{-0.4in}
\includegraphics[scale=0.5]{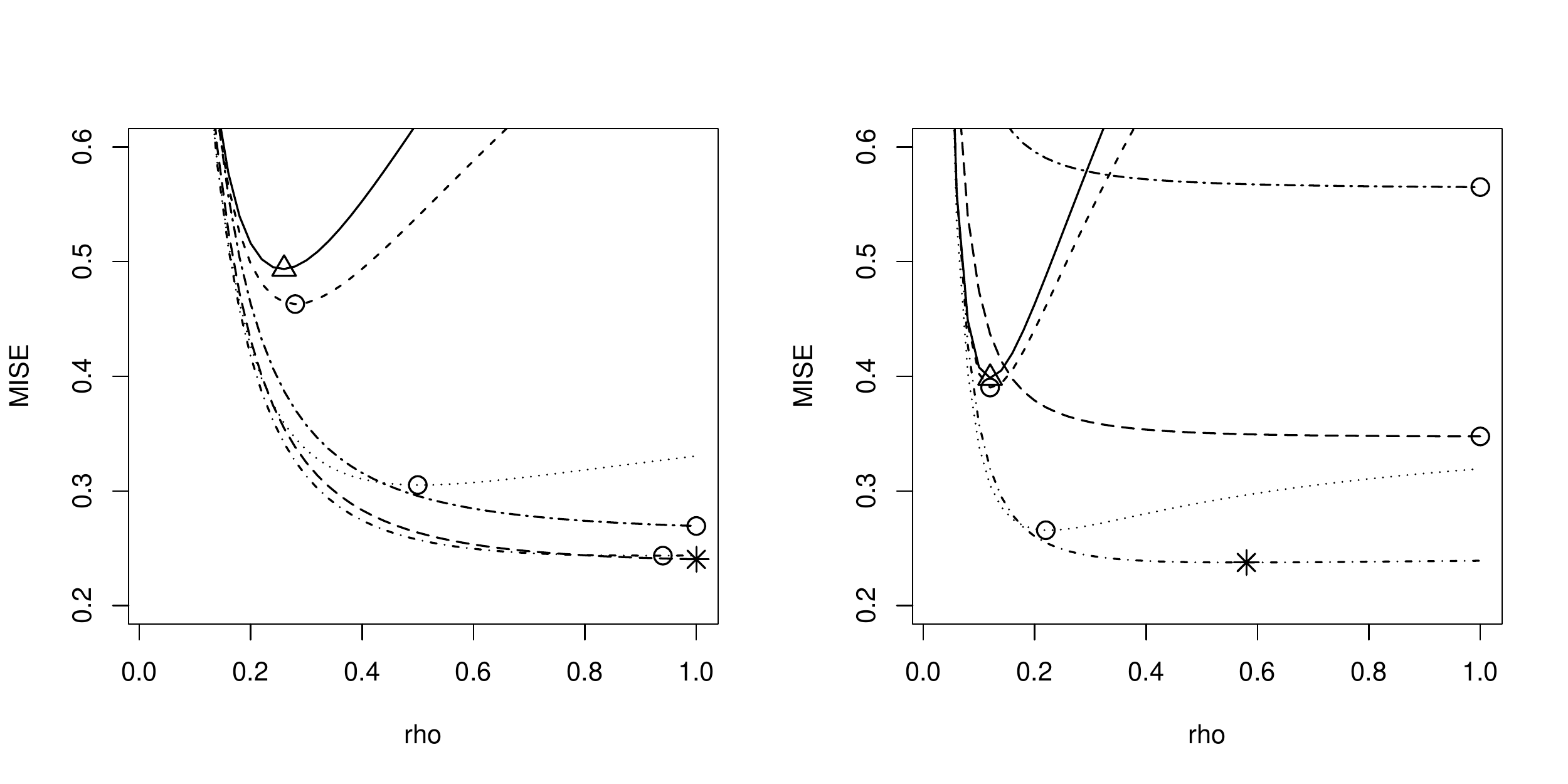}
\vspace{-0.2in}
\includegraphics[scale=0.5]{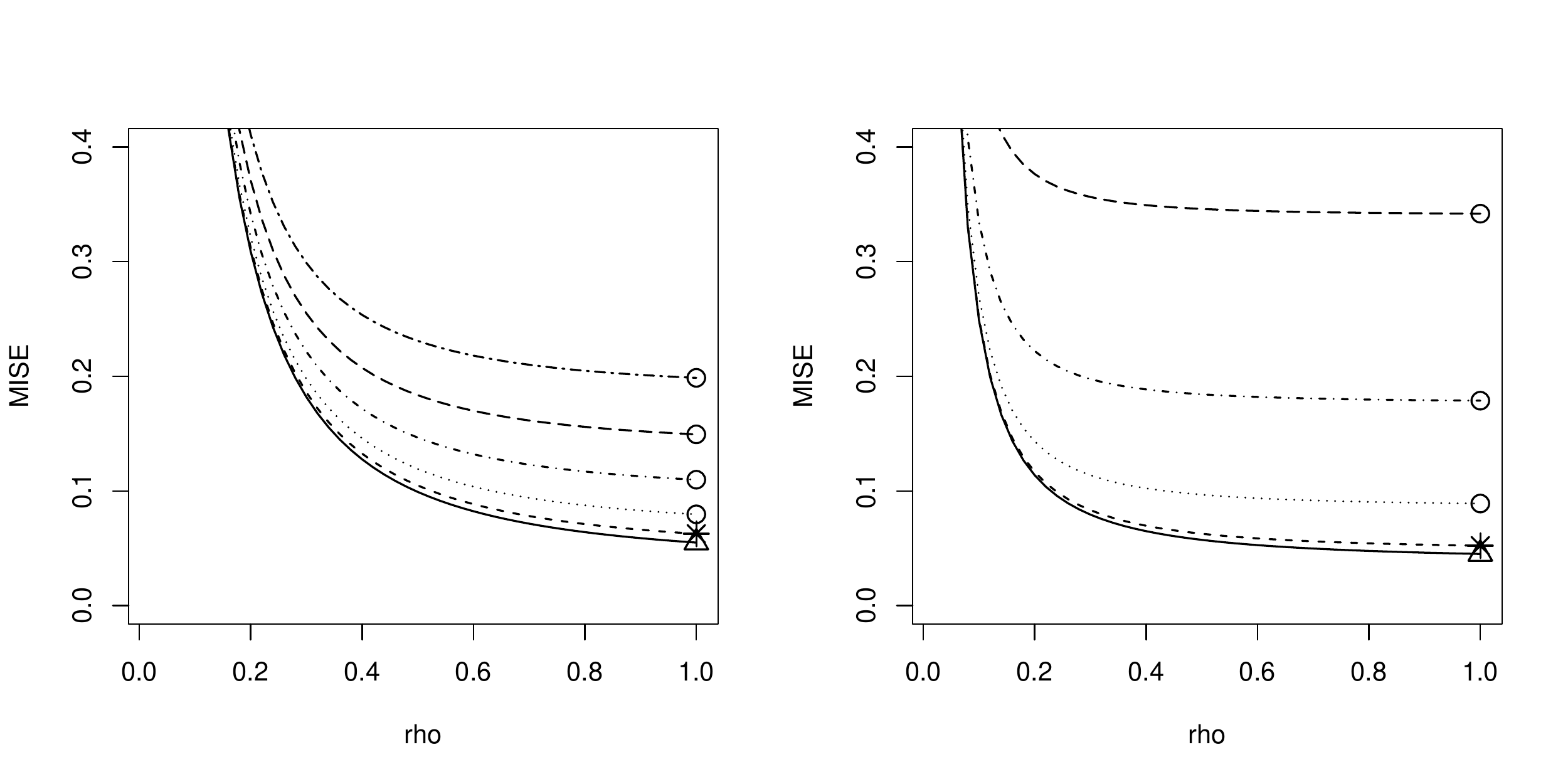}
}
\end{center}
\caption{Plots of the MSEs of the Tikhonov estimator (solid curves) and the hybrid regularisation estimator for $r = 1$ (dashed curves), $r = 2$ (dotted curves), $r = 3$ (dot-dashed curves), $r = 4$ (long-dashed curves) and $r = 5$ (two-dashed curves). The plots in the left and the right columns correspond to $\alpha = 1.1$ and $2$, respectively. The plots in the top, middle and bottom rows correspond to $\beta_{1}$, $\beta_{2}$ and $\beta_{3}$, respectively. }
\label{Fig-1}
\end{figure}
Figure \ref{Fig-1} gives the plots of the MSEs of the two estimators for different choices of $\rho$. In each plot, we have considered the mean squared errors of the hybrid estimator for every $r = 1,2 \ldots,5$. For the Tikhonov estimator, the smallest value of the mean squared is designated by a triangle. For the hybrid  estimator, the smallest value of the mean squared error for each choice of $r$ is marked by a circle. We also point out the smallest mean squared error across the different choices of $r$ by a star. In all of the above cases, the optimal values of $\rho$ and $r$ can be read from the plot and we do not mark them to avoid clutter.

The top four plots in Figure \ref{Fig-1} show that the optimal value of the mean squared error is markedly smaller for the hybrid estimator than for the Tikhonov estimator with the ratio between the two mean squared errors being about $2$ and $1.4$ for $\alpha = 1.1$ and $2$, respectively for both $\beta = \beta_{1}$ and $\beta_{2}$ (see Table \ref{Tab-1}). It can also be remarked that the optimal mean squared error corresponding to the hybrid estimator can also improve upon the optimal Tikhonov mean squared error even for some values of $r$ that are suboptimal. In fact, the difference in each case is statistically significant in the following sense -- the two mean squared errors, which are averages of independent Monte Carlo iterations, are significantly different, when a large sample test of difference of two means is applied. These observations lend support to Theorem \ref{thm-data2}. In the plots in the last row in Figure \ref{Fig-1}, where $\beta = \beta_{3}$, the minimum mean squared errors of the two estimators are not significantly different. This does not contradict Theorem \ref{thm-data2} since $\beta_{3}$ does not satisfy the assumption in that theorem for any $r = 1,2,\ldots,5$.

The first two choices of $\beta$ are at least partially expressed by the eigenfunctions associated with the five largest eigenvalues of $X$. Note that these eigenvalues explain only about $56\%$ of the total variation of $X$ if $\alpha = 1.1$, while this percentage is about $90\%$ if $\alpha = 2$. Thus, the performance of the hybrid estimator does not seem to depend much on whether or not the eigenvalues associated with the eigenfunctions expressing $\beta$ explain a large amount of the total variation of $X$. It is also observed that if we had chosen $r$ by the ``$85\%$-rule'', then one would end up choosing more principal components compared to the optimal value of $r$ found in the simulation studies for each value of $\alpha$ and $\beta = \beta_{1}$ or $\beta_{2}$. Further, the optimal number of principal components is $1$ for both values of $\alpha$ when $\beta = \beta_{3}$. These findings indicate that one should generally \emph{not} use the ``$85\%$-rule'' for choosing $r$ in the construction of the hybrid  estimator. Further, the simulation studies also confirm that when the eigenvalues of $X$ decay slowly (well-conditioned regime), it is better to choose a higher value of $r$. By doing so, we make substantial gains if $\beta$ is at least partially expressed by those $r$ eigenfunctions, and we will only perhaps lose out slightly otherwise. On the other hand, if the eigenvalues of $X$ decay fast (ill-conditioned regime), then a more conservative choice should be used (see Figure \ref{Fig-1}). This is consistent with the choice of $r$ using condition numbers that was recommended in subsection \ref{2.2.1}.

We also observe that the mean squared error of the hybrid estimator for appropriately chosen $r$ is significantly smaller than that of the Tikhonov estimator for all values of $\rho$ greater than the optimal one for the latter estimator, which is small except when $\beta = \beta_{3}$. In that case, for all $\rho > 0.2$, the mean squared errors of the two estimators are almost coincident for an appropriately chosen $r$.
From the simulation studies, it seems that the hybrid estimator acts as a safeguard against over-estimation. This is in contrast to the Tikhonov estimator which is found to be much more sensitive to choice of large values of $\rho$ when $\beta = \beta_{1}$ or $\beta_{2}$ (see Figure \ref{Fig-1}).

We next compare the MSEs of the hybrid regularisation estimator with that of the Tikhonov regularisation estimator as well as the spectral truncation estimator when the regularisation parameters in the hybrid estimator are chosen using the fully automated double cross-validation technique discussed in subsection \ref{2.2.1} and the regularisation parameter involved in each of the other two estimators is also chosen using cross-validation. Table \ref{Tab-1} gives the MSEs (averaged over $1000$ Monte-Carlo iterations) of the three estimators for the simulated models considered earlier as well as the results for the following choice of the $\gamma_{j}$'s in those models: $\gamma_{1} = 1$, $\gamma_{j} = 0.2(-1)^{j+1}(1-0.0001j)$ if $2 \leq j \leq 4$ and $\gamma_{5j+k} = 0.2(-1)^{5j+k+1}\{(5j)^{-\alpha/2}-0.0001k\}$ for $j\geq 1$ and $0 \leq k \leq 4$. As in section \ref{3}, we have chosen $\alpha = 1.1$ or $2$. This new set of $\gamma_{j}$'s generate ``closely-spaced'' eigenvalues and was also considered by \citet{HH07} and \citet{YC10}. The choice of the $\gamma_{j}$'s considered towards the beginning of this section leads to``well-spaced'' eigenvalues. It is known that the spectral truncation estimator has better (worse) performance compared to the Tikhonov regularisation estimator in the ``well-spaced'' (``closely-spaced'') scenario (see \citet{HH07}). The Monte-Carlo standard deviations of the MSEs are mostly of the order of $10^{-3}$ with some exceptions, but even these do not exceed $0.016$. All the significance statements made later take these standard deviations into account. We have also compared the hybrid estimator under the above models with the other two estimators when the hybrid estimator was computed using the algorithm given in subsection \ref{2.2.1}. The overall results obtained were quite similar to those presented in Table \ref{Tab-1}. {We also compared the performance of the estimators when the functional covariate is observed with error, and the results are reported in the Appendix}.

\begin{table}[t]
\caption{MSEs of the hybrid regularisation, the Tikhonov regularisation and the spectral truncation estimators.}{
\begin{center}
\scalebox{1}{
\begin{tabular}{cccccccc}
 \hline \\
\multicolumn{8}{c}{well-spaced} \\ \\ \hline \\ 
$\beta$ & $\alpha$ & $MSE_{ST}^{GCV}$ & $MSE_{ST}^{true}$ & $MSE_{TR}^{GCV}$ & $MSE_{TR}^{true}$ & $MSE_{HR}^{GCV}$ & $MSE_{HR}^{true}$ \\ \\ \hline 
$\beta_{1}$ & $1.1$ & $0.285$ & $0.272$ & $0.773$ & $0.516$ & $0.346$ & $0.26$ \\ 
 & $2$ & $0.296$ & $0.286$ & $0.608$ & $0.445$ & $0.311$ & $0.274$ \\ \\
$\beta_{2}$ & $1.1$ & $0.271$ & $0.247$ & $0.763$ & $0.494$ & $0.357$ & $0.24$ \\
 & $2$ & $0.252$ & $0.241$ & $0.689$ & $0.409$ & $0.284$ & $0.234$ \\ \\
$\beta_{3}$ & $1.1$ & $0.052$ & $0.05$ & $0.057$ & $0.055$ & $0.066$ & $0.063$ \\
 & $2$ & $0.05$ & $0.049$ & $0.046$ & $0.045$ & $0.052$ & $0.051$ \\
 \hline \\ 
 \multicolumn{8}{c}{closely-spaced} \\ \\ \hline \\ 
 $\beta$ & $\alpha$ & $MSE_{ST}^{GCV}$ & $MSE_{ST}^{true}$ & $MSE_{TR}^{GCV}$ & $MSE_{TR}^{true}$ & $MSE_{HR}^{GCV}$ & $MSE_{HR}^{true}$ \\ \\ \hline 
$\beta_{1}$ & $1.1$ & $1.09$ & $0.857$ & $1.054$ & $0.888$ & $0.935$ & $0.851$ \\ 
 & $2$ & $0.949$ & $0.854$ & $0.821$ & $0.694$ & $0.725$ & $0.697$ \\ \\
$\beta_{2}$ & $1.1$ & $1.055$ & $0.822$ & $1.015$ & $0.835$ & $0.887$ & $0.808$ \\
 & $2$ & $0.896$ & $0.813$ & $0.763$ & $0.647$ & $0.681$ & $0.647$ \\ \\
$\beta_{3}$ & $1.1$ & $0.051$ & $0.05$ & $0.045$ & $0.043$ & $0.058$ & $0.051$ \\
 & $2$ & $0.049$ & $0.048$ & $0.046$ & $0.043$ & $0.051$ & $0.05$ \\
 \hline
\end{tabular}}
\end{center}}
\label{Tab-1}
\end{table}

It is observed from Table \ref{Tab-1} that under the well-spaced scenario, the MSEs (true as well as cross-validated) of the hybrid estimator are significantly smaller than those of the Tikhonov estimator for $\beta = \beta_{1}$ and $\beta_{2}$. Somewhat surprisingly, the true MSEs of the hybrid estimator and the spectral estimator are not dissimilar. Although the cross-validation MSE of the spectral estimator for $\beta_{1}$ as well as $\beta_{2}$ is significantly smaller than that of the hybrid estimator for $\alpha= 1.1$, these MSEs are quite close when $\alpha = 2$. In the closely-spaced case, the cross-validation MSEs of the hybrid estimator are significantly smaller than those of the spectral and the Tikhonov estimators for all choices of $\alpha$ under $\beta_{1}$ and $\beta_{2}$. For these $\beta$'s, the true MSEs of the spectral estimator and hybrid estimator are comparable for $\alpha = 1.1$, but the former become significantly larger when $\alpha = 2$. {For $\beta = \beta_{3}$, it is found that the MSEs (true as well as cross-validated) of the three estimators are not significantly different from one another when $\alpha = 2$. In case $\alpha = 1.1$, the cross-validated MSE of the hybrid estimator is marginally larger than those of other two estimators. Further, the true MSE of the hybrid estimator is marginally larger than that of the spectral estimator in the well-spaced scenario. As mentioned in the earlier simulation study, these findings do not contradict the domination result in Theorem \ref{thm-data2}.} It seems that in both the well-spaced and the closely spaced situations, the cross-validation method for $\rho$ is slightly unstable when the eigenvalues decay slowly. This may be attributed to the fact that the cross-validation estimate is based on prediction error, whose difficulty reduces as the eigenvalues decay faster (see the discussion in p. 3428 in \citet{YC10}).

\section{Proofs of Formal Statements}
\label{5}

We first prove a Lemma that will allow us to connect the Fourier coefficient decay of $\beta$, the eigenvalue decay of $\mathscr{K}$, and the ridge parameter $\rho$.
 
\begin{lemma}  \label{lemma-2}
Suppose that $\lambda_{1} > \lambda_{2} > \ldots > 0$ is a sequence of reals and $\rho > 0$. Assume that $\lambda_{j} = O(j^{-\alpha})$ for some $\alpha > 1$ and for all sufficiently large $j \geq 1$. Let $\{b_{j}\}_{j \geq 1}$ be another sequence of reals such that $|b_{j}| \leq j^{-\eta}$ for some $\eta > 0$ and for all sufficiently large $j$. Then, for any $b \geq a \geq 0$ and any $c \geq 0$ with $2c\eta + a\alpha > 1$, we have 
$$\sum_{j=1}^{\infty} b_{j}^{2c}\lambda_{j}^{a}/(\lambda_{j} + \rho)^{b} \leq \mathrm{const.}\rho^{\frac{2c\eta}{\alpha}-b+a-\frac{1}{\alpha}}$$
if $2c\eta < \alpha(b-a) + 1$. Further, if $2c\eta > \alpha(b-a) + 1$, then $\sup_{\rho > 0} \sum_{j=1}^{\infty} b_{j}^{2c}\lambda_{j}^{a}/(\lambda_{j} + \rho)^{b} = \sum_{j=1}^{\infty} b_{j}^{2c}\lambda_{j}^{a}/\lambda_{j}^{b} < \infty$.
\end{lemma}

\begin{proof}
Consider the case when $2c\eta < \alpha(b-a) + 1$, and fix $J = \rho^{-1/\alpha}$. Note that 
\begin{eqnarray*}
\sum_{j > J} \frac{b_{j}^{2c}\lambda_{j}^{a}}{(\lambda_{j} + \rho)^{b}} &\leq& \mathrm{const.}\rho^{-b} \sum_{j > J} b_{j}^{2c}\lambda_{j}^{a} \ \leq \ \mathrm{const.}\rho^{-b} \sum_{j > J} j^{-2c\eta-a\alpha} \\
&\leq& \mathrm{const.}\rho^{-b}\int_{J}^{\infty} x^{-2c\eta-a\alpha} dx \ \leq \ \mathrm{const.}\rho^{\frac{2c\eta}{\alpha}-b+a-\frac{1}{\alpha}}.
\end{eqnarray*}
Also, 
\begin{eqnarray*}
\sum_{j \leq J} \frac{b_{j}^{2c}\lambda_{j}^{a}}{(\lambda_{j} + \rho)^{b}} &\leq& \mathrm{const.} \sum_{j \leq J} b_{j}^{2c}\lambda_{j}^{-b+a} \ \leq \ \mathrm{const.} \sum_{j \leq J} j^{-2c\eta+\alpha(b-a)} \\
&\leq& \mathrm{const.}\int_{0}^{J} x^{-2c\eta+\alpha(b-a)} dx \ \leq \ \mathrm{const.}\rho^{\frac{2c\eta}{\alpha}-b+a-\frac{1}{\alpha}}.
\end{eqnarray*}
This completes the proof of the first part of the lemma. 

 Next consider the case when $2c\eta > \alpha(b-a) + 1$. Note that $\sum_{j =1}^{\infty} b_{j}^{2c}\lambda_{j}^{a}/(\lambda_{j} + \rho)^{b} \leq \sum_{j =1}^{\infty} b_{j}^{2c}\lambda_{j}^{-b+a}$ for all $\rho > 0$. Further, 
$$\sum_{j =1}^{\infty} b_{j}^{2c}\lambda_{j}^{-b+a} \leq \mathrm{const.} \sum_{j=1}^{\infty} j^{-2c\eta+\alpha(b-a)} < \infty.$$
This proves the second part of the lemma.
\end{proof}

\begin{proof}[Proof of Theorem \ref{thm-orcl1}]
(a) Note that $\mathbb{E}(\hat{\beta}_{TR}) = (\mathscr{K} + \rho\mathscr{I})^{-1} \mathbb{E}(yX) = (\mathscr{K} + \rho\mathscr{I})^{-1}\mathscr{K}\beta = \sum_{j=1}^{\infty} \lambda_{j}(\lambda_{j} + \rho)^{-1} \langle\beta,\phi_{j}\rangle \phi_{j}$. Thus, $\mathrm{Bias}(\hat{\beta}_{TR}) = \rho\sum_{j=1}^{\infty} (\lambda_{j} + \rho)^{-1}\langle\beta,\phi_{j}\rangle \phi_{j}$. 

 For simplicity of the calculations of the variances of $\widetilde{\beta}_{TR}$ and $\widetilde{\beta}_{HR}$, we will take $\widetilde{C} = n^{-1} \sum_{i=1}^{n} y_{i}X_{i}$, $\widetilde{C}_{1} = n^{-1} \sum_{i=1}^{n} y_{i}X_{i}$ and $\widetilde{C}_{2} = n^{-1} \sum_{i=1}^{n} y_{i}Z_{i}$. This substitution will not effect the orders of the MSEs since the means $\overline{X}$, $\overline{Y}$ and $\overline{Z}$ are all of smaller order than the original and the substituted estimators of the cross-covariances.
\begin{eqnarray*}
\mathrm{Var}(\widetilde{\beta}_{TR}) &=& \mathrm{Var}\left(n^{-1} \sum_{i=1}^{n} y_{i} \{\sum_{j=1}^{\infty} (\lambda_{j}+\rho)^{-1} \langle X_{i}, \phi_{j} \rangle \phi_{j}\}\right) \\
&=& n^{-1} \mathrm{Var}\left(y\sum_{j=1}^{\infty} (\lambda_{j}+\rho)^{-1} \langle X, \phi_{j} \rangle \phi_{j}\right) \\
&=& n^{-1} \sum_{j=1}^{\infty} \sum_{j' =1}^{\infty} \frac{\langle \mathrm{Cov}(\langle yX, \phi_{j} \rangle, \langle yX, \phi_{j'} \rangle}{(\lambda_{j}+\rho)(\lambda_{j'}+\rho)} (\phi_{j} \otimes \phi_{j'}).
\end{eqnarray*}
Note that  $\mathrm{Cov}(\langle yX, \phi_{j} \rangle, \langle yX, \phi_{j'} \rangle) = \langle \mathrm{Var}(yX)\phi_{j}, \phi_{j'}\rangle$. Now, 
\begin{eqnarray*}
\mathrm{Var}(yX) &=& \mathbb{E}(y^{2}X \otimes X) - \mathbb{E}(yX) \otimes \mathbb{E}(yX) \\
&=& \mathbb{E}\{[\langle X,\beta\rangle^{2} + \epsilon^{2} + 2\epsilon\langle X,\beta\rangle](X \otimes X)\} - \mathscr{K}\beta \otimes \mathscr{K}\beta \\
&=& \mathbb{E}\{\langle X,\beta\rangle^{2} X \otimes X\} + \sigma^{2}\mathscr{K} - \mathscr{K}\beta \otimes \mathscr{K}\beta.
\end{eqnarray*}
It follows that $\langle \mathrm{Var}(yX)\phi_{j}, \phi_{j'}\rangle = \mathbb{E}\{\langle X,\beta\rangle^{2}\langle X,\phi_{j}\rangle\langle X,\phi_{j'}\rangle\} + \sigma^{2}\delta_{jj'}\lambda_{j} - \lambda_{j}\lambda_{j'}\langle\beta,\phi_{j}\rangle\langle\beta,\phi_{j'}\rangle$. Note that 
\begin{eqnarray*}
&& \mathbb{E}\{\langle X,\beta\rangle^{2}\langle X,\phi_{j}\rangle\langle X,\phi_{j'}\rangle\} \\
&& = \sum_{l=1}^{\infty} \sum_{l'=1}^{\infty} \langle\beta,\phi_{l}\rangle\langle\beta,\phi_{l'}\rangle \mathbb{E}\{\langle X,\phi_{l}\rangle\langle X,\phi_{l'}\rangle\langle X,\phi_{j}\rangle\langle X,\phi_{j'}\rangle\},
\end{eqnarray*}
and we denote the above infinite sum by $a_{jj'}$. If $j = j'$, then all terms vanish except those for which $l = l'$. This is due to the fact that the $\langle X,\phi_{k}\rangle$'s are independent with zero mean. So, in this case, 
\begin{eqnarray*}
a_{jj} &=& \sum_{l=1}^{\infty} \langle\beta,\phi_{l}\rangle^{2} \mathbb{E}\{\langle X,\phi_{l}\rangle^{2}\langle X,\phi_{j}\rangle^{2}\} \\
&=& \langle\beta,\phi_{j}\rangle^{2}\mathbb{E}\{\langle X,\phi_{j}\rangle^{4}\} + \lambda_{j}\sum_{l \neq j} \langle\beta,\phi_{l}\rangle^{2} \lambda_{l} \\
&=& \langle\beta,\phi_{j}\rangle^{2}[\mathrm{Var}(\langle X,\phi_{j}\rangle^{2}) + \lambda_{j}^{2}] + \lambda_{j}\sum_{l \neq j} \langle\beta,\phi_{l}\rangle^{2} \lambda_{l}  \\
&=& \langle\beta,\phi_{j}\rangle^{2}\mathrm{Var}(\langle X,\phi_{j}\rangle^{2}) + \lambda_{j}\sum_{l =1}^{\infty} \langle\beta,\phi_{l}\rangle^{2} \lambda_{l} \\
&=& \langle\beta,\phi_{j}\rangle^{2}\mathrm{Var}(\langle X,\phi_{j}\rangle^{2}) + \lambda_{j}\langle \mathscr{K}\beta,\beta\rangle.
\end{eqnarray*}
On the other hand if $j \neq j'$, then 
\begin{eqnarray*}
a_{jj'} = 2\langle\beta,\phi_{j}\rangle\langle\beta,\phi_{j'}\rangle \mathbb{E}\{\langle X,\phi_{j}\rangle^{2}\}\mathbb{E}\{\langle X,\phi_{j'}\rangle^{2}\} = 2\langle\beta,\phi_{j}\rangle\langle\beta,\phi_{j'}\rangle\lambda_{j}\lambda_{j'}.
\end{eqnarray*}
Consequently, 
\begin{eqnarray*}
&& \mathrm{Var}(\widetilde{\beta}_{TR}) \\
&=& n^{-1} \sum_{j=1}^{\infty} \sum_{j'=1}^{\infty} (\lambda_{j} + \rho)^{-1}(\lambda_{j'} + \rho)^{-1}\langle \mathrm{Var}(yX)\phi_{j}, \phi_{j'} \rangle (\phi_{j} \otimes \phi_{j'}) \\
&=& n^{-1} \sum_{j=1}^{\infty} (\lambda_{j} + \rho)^{-2} \left[\langle\beta,\phi_{j}\rangle^{2}\{\mathrm{Var}(\langle X,\phi_{j}\rangle^{2}) - \lambda_{j}^{2}\} \ + \right.\\
&& \hspace{5cm} \left.\lambda_{j}\{\langle \mathscr{K}\beta,\beta\rangle + \sigma^{2}\}\right] (\phi_{j} \otimes \phi_{j}) \\
&& + \ n^{-1} \sum_{1 \leq j \neq j' < \infty} \frac{\lambda_{j} \lambda_{j'}}{(\lambda_{j} + \rho)(\lambda_{j'} + \rho)} \langle\beta,\phi_{j}\rangle \langle\beta,\phi_{j'}\rangle (\phi_{j} \otimes \phi_{j'}).
\end{eqnarray*}
Since $\tr(\phi_{j} \otimes \phi_{j'}) = \langle\phi_{j},\phi_{j'}\rangle = \delta_{jj'}$, where $\delta$ is the Kronecker delta function,  it follows that
\begin{eqnarray}
&& \mathrm{MSE}(\widetilde{\beta}_{TR}) = \tr\{\mathbb{E}\{(\widetilde{\beta}_{TR} - \beta) \otimes (\widetilde{\beta}_{TR} - \beta)\}\} \label{mse-HH-1} \\
&=& \tr\{\mathrm{Var}(\widetilde{\beta}_{TR})\} + \tr\{\mathrm{Bias}(\widetilde{\beta}_{TR}) \otimes \mathrm{Bias}(\widetilde{\beta}_{TR})\} \nonumber \\
&=& n^{-1} \sum_{j=1}^{\infty} (\lambda_{j} + \rho)^{-2}[\langle\beta,\phi_{j}\rangle^{2}\{\mathrm{Var}(\langle X,\phi_{j}\rangle^{2}) - \lambda_{j}^{2}\} + \lambda_{j}\{\langle \mathscr{K}\beta,\beta\rangle + \sigma^{2}\}] \nonumber \\
&& + \ \rho^{2}\sum_{j=1}^{\infty} (\lambda_{j} + \rho)^{-2}\langle\beta,\phi_{j}\rangle^{2} \nonumber
\end{eqnarray}
Recall now that $\mathscr{K} = \mathscr{K}_{1} + \mathscr{K}_{2}$, and the eigenspaces of $\mathscr{K}_{1}$ and $\mathscr{K}_{2}$ are orthogonal. As a result, the eigenvalues and the eigenfunctions of $\mathscr{K}$ are the union of the sets of eigenvalues and the eigenfunctions, respectively, of $\mathscr{K}$ and $\mathscr{K}_{2}$. Furthermore, without loss of generality, we can assume that the eigenvalues of $\mathscr{K}_{1}$ are the $r$ largest eigenvalues of $\mathscr{K}$ to alleviate notation. With this convention,
\begin{eqnarray}
&& \mathrm{MSE}(\widetilde{\beta}_{TR}) \label{mse-HH-2} \\
&=& n^{-1} \sum_{j=1}^{r} (\lambda_{j1} + \rho)^{-2}[\langle\beta,\phi_{j1}\rangle^{2}\{\mathrm{Var}(\langle Y,\phi_{j1}\rangle^{2}) - \lambda_{j1}^{2}\} + \lambda_{j1}\{\langle \mathscr{K}\beta,\beta\rangle + \sigma^{2}\}] \nonumber \\
&& + \ n^{-1} \sum_{j=1}^{\infty} (\lambda_{j2} + \rho)^{-2}[\langle\beta,\phi_{j2}\rangle^{2}\{\mathrm{Var}(\langle Z,\phi_{j2}\rangle^{2}) - \lambda_{j2}^{2}\} + \lambda_{j2}\{\langle \mathscr{K}\beta,\beta\rangle + \sigma^{2}\}]  \nonumber \\
&& + \ \rho^{2}\left[\sum_{j=1}^{r} (\lambda_{j1} + \rho)^{-2}\langle\beta,\phi_{j1}\rangle^{2} + \sum_{j=1}^{\infty} (\lambda_{j2} + \rho)^{-2}\langle\beta,\phi_{j2}\rangle^{2}\right].  \nonumber
\end{eqnarray}

We next compute $\mathrm{MSE}(\widetilde{\beta}_{HR}) = \tr\{\mathbb{E}\{(\widetilde{\beta}_{HR} - \beta) \otimes (\widetilde{\beta}_{HR} - \beta)\}\}$. Write $\widetilde{\beta}_{1} = \mathscr{K}_{1}^{-}\widetilde{C}_{1}$ and $\widetilde{\beta}_{2} = \mathscr{K}_{\rho,2}^{-}\widetilde{C}_{2}$. Thus $\widetilde{\beta} = \widetilde{\beta}_{1} + \widetilde{\beta}_{2}$. Using the fact that $Y$ and $Z$ are uncorrelated, it can be straightforwardly shown that $\mathbb{E}(\widetilde{\beta}_{1}) = \sum_{j=1}^{r} \langle\beta,\phi_{j1}\rangle\phi_{j1}$. Similarly, $\mathbb{E}(\widetilde{\beta}_{2}) = \sum_{j=1}^{\infty} (\lambda_{j2} + \rho)^{-1} \lambda_{j2}\langle\beta,\phi_{j2}\rangle\phi_{j2}$. Since we have assumed that $\beta = \mathscr{P}_{1}\beta + \mathscr{P}_{2}\beta$ for identifiability, we have $\mathrm{Bias}(\widetilde{\beta}_{HR}) = \rho\sum_{j=1}^{\infty} (\lambda_{j2} + \rho)^{-1} \langle\beta,\phi_{j2}\rangle\phi_{j2}$. 

 Next, note that $\mathrm{Var}(\widetilde{\beta}_{HR}) = \mathrm{Var}(\widetilde{\beta}_{1}) + \mathrm{Var}(\widetilde{\beta}_{2}) + \mathrm{Cov}(\widetilde{\beta}_{1},\widetilde{\beta}_{2}) + \mathrm{Cov}(\widetilde{\beta}_{2},\widetilde{\beta}_{1}) = T_{1} + T_{2} + T_{3} + T_{4}$, say. Clearly, $T_{4} = T_{3}^{*}$. Also note that 
 $$T_{1} = n^{-1} \sum_{j=1}^{r} \sum_{j'=1}^{r} \lambda_{j1}^{-1}\lambda_{j'1}^{-1} \langle \mathrm{Var}(yY)\phi_{j1},\phi_{j'1}\rangle \phi_{j1} \otimes \phi_{j'1}.$$ 
 Now, 
\begin{eqnarray*}
\mathrm{Var}(yY) &=& \mathbb{E}(y^{2}Y \otimes Y) - \mathbb{E}(yY) \otimes \mathbb{E}(yY) \\
&=& \mathbb{E}\{[\langle Y,\beta\rangle^{2} + \langle Z,\beta\rangle^{2} + \epsilon^{2} + 2\langle Y,\beta\rangle\langle Z,\beta\rangle +\\
&& \hspace{1cm} 2\epsilon\langle Y,\beta\rangle + 2\epsilon\langle Z,\beta\rangle](Y \otimes Y)\} - \mathscr{K}_{1}\beta \otimes \mathscr{K}_{1}\beta \\
&=& \mathbb{E}\{\langle Y,\beta\rangle^{2}(Y \otimes Y)\} + (\langle\mathscr{K}_{2}\beta,\beta\rangle + \sigma^{2})\mathscr{K}_{1} - \mathscr{K}_{1}\beta \otimes \mathscr{K}_{1}\beta.
\end{eqnarray*}
Also, $\langle \mathrm{Var}(yY)\phi_{j1}, \phi_{j'1}\rangle = \mathbb{E}\{\langle Y,\beta\rangle^{2}\langle Y,\phi_{j1}\rangle\langle Y,\phi_{j'2}\rangle\} + (\langle\mathscr{K}_{2}\beta,\beta\rangle + \sigma^{2})\delta_{jj'}\lambda_{j1} - \lambda_{j1}\lambda_{j'1}\langle\beta,\phi_{j1}\rangle\langle\beta,\phi_{j'1}\rangle$. Calculations similar to those used earlier to derive the term $\mathrm{Var}(\widetilde{\beta}_{TR})$ now yield
\begin{eqnarray*}
\mathbb{E}\{\langle Y,\beta\rangle^{2}\langle Y,\phi_{j1}\rangle\langle Y,\phi_{j'2}\rangle\} = 
\begin{cases}
\langle\beta,\phi_{j1}\rangle^{2}\mathrm{Var}(\langle Y,\phi_{j1}\rangle^{2}) + \lambda_{j1}\langle\mathscr{K}_{1}\beta,\beta\rangle &\text{if $j = j'$} \\
2\langle\beta,\phi_{j1}\rangle\langle\beta,\phi_{j'1}\rangle\lambda_{j1}\lambda_{j'1} &\text{if $j \neq j'$}.
\end{cases}
\end{eqnarray*}
Thus, we obtain
\begin{eqnarray*}
T_{1} &=&  n^{-1} \sum_{j=1}^{r} \lambda_{j1}^{-2}\left[\langle\beta,\phi_{j1}\rangle^{2}\{\mathrm{Var}(\langle Y,\phi_{j1}\rangle^{2}) - \lambda_{j1}^{2}\} + \right. \\
&& \hspace{4cm} \left. \lambda_{j1}\{\langle \mathscr{K}\beta,\beta\rangle + \sigma^{2}\}\right] (\phi_{jY} \otimes \phi_{jY}) \\
&& + \ n^{-1} \sum_{1 \leq j \neq j' \leq r} \langle\beta,\phi_{jY}\rangle\langle\beta,\phi_{j'Y}\rangle (\phi_{jY} \otimes \phi_{j'Y}).
\end{eqnarray*}
Similar calculations yield
\begin{eqnarray*}
T_{2} &=&  n^{-1} \sum_{j=1}^{\infty} (\lambda_{j2} + \rho)^{-2}\left[\langle\beta,\phi_{j2}\rangle^{2}\{\mathrm{Var}(\langle Z,\phi_{j2}\rangle^{2}) - \lambda_{j2}^{2}\} + \right. \\
&& \hspace{4cm} \left. \lambda_{j2}\{\langle \mathscr{K}\beta,\beta\rangle + \sigma^{2}\}\right] (\phi_{j2} \otimes \phi_{j2}) \\
&& + \ n^{-1} \sum_{1 \leq j \neq j' < \infty} \frac{\lambda_{j2}\lambda_{j'2}}{(\lambda_{j2} + \rho)(\lambda_{j'2} + \rho)}\langle\beta_{2},\phi_{j2}\rangle\langle\beta,\phi_{j'2}\rangle (\phi_{j2} \otimes \phi_{j'2}).
\end{eqnarray*}
Moreover, $T_{3} = n^{-1} \sum_{j=1}^{r} \sum_{j'=1}^{\infty} (\lambda_{j'2} + \rho)^{-1}\lambda_{j'2} \langle\beta,\phi_{j1}\rangle\langle\beta,\phi_{j'2}\rangle (\phi_{j1} \otimes \phi_{j'2})$. So,
\begin{eqnarray}
&& \mathrm{MSE}(\widetilde{\beta}_{HR}) \label{mse-HR-1} \\
&=& n^{-1} \sum_{j=1}^{r} \lambda_{j1}^{-2}[\langle\beta,\phi_{j1}\rangle^{2}\{\mathrm{Var}(\langle Y,\phi_{j1}\rangle^{2}) - \lambda_{j1}^{2}\} + \lambda_{j1}\{\langle \mathscr{K}\beta,\beta\rangle + \sigma^{2}\}] \nonumber \\
&& + \ n^{-1} \sum_{j=1}^{\infty} (\lambda_{j2} + \rho)^{-2}[\langle\beta,\phi_{j2}\rangle^{2}\{\mathrm{Var}(\langle Z,\phi_{j2}\rangle^{2}) - \lambda_{j2}^{2}\} + \lambda_{j2}\{\langle \mathscr{K}\beta,\beta\rangle + \sigma^{2}\}] \nonumber \\
&& + \ \rho^{2}\sum_{j=1}^{\infty} (\lambda_{j2} + \rho)^{-2} \langle\beta,\phi_{j2}\rangle^{2}. \nonumber 
\end{eqnarray}
Hence, it follows from \eqref{mse-HH-2} and \eqref{mse-HR-1} that
\begin{eqnarray}
&& \mathrm{MSE}(\widetilde{\beta}_{TR}) - \mathrm{MSE}(\widetilde{\beta}_{HR}) \label{mse-compare}\\
&=& n^{-1} \sum_{j=1}^{r} \{(\lambda_{j1}+\rho)^{-2} - \lambda_{j1}^{-2}\}\left[\langle\beta,\phi_{j1}\rangle^{2}\{\mathrm{Var}(\langle Y,\phi_{j1}\rangle^{2}) - \lambda_{j1}^{2}\} + \right. \nonumber \\
&& \hspace{2cm} \left. \lambda_{j1}\{\langle \mathscr{K}\beta,\beta\rangle + \sigma^{2}\}\right] \nonumber \\
&& + \ \rho^{2}\sum_{j=1}^{r} (\lambda_{j1} + \rho)^{-2}\langle\beta,\phi_{j1}\rangle^{2} \nonumber \\
&=& n^{-1}A_{1}(\rho) + \rho^{2}A_{2}(\rho), \nonumber 
\end{eqnarray}
where $A_{1}(\rho)$ and $A_{2}(\rho)$ are the terms associated with $n^{-1}$ and $\rho^{2}$, respectively, the expression of $\mathrm{MSE}(\widetilde{\beta}_{TR}) - \mathrm{MSE}(\widetilde{\beta}_{HR})$ given above. It follows from the above equality that for each fixed $\rho > 0$, $\mathrm{MSE}(\widetilde{\beta}_{TR})- \mathrm{MSE}(\widetilde{\beta}_{HR}) > 0$ for all sufficiently large $n$. This proves part (a) of the theorem. 

To prove part (b), let $\rho_{n} = cn^{-\gamma}$. Then, it straightforwardly follows that for $\rho = \rho_{n}$, we have $n^{-1} = O(\rho_{n}^{2})$, equivalently $O(n^{-2\gamma})$. Further, $A_{1}(\rho_{n}) = o(1)$ as $n \rightarrow \infty$. Define $B(n) = A_{2}(\rho_{n})$. So, $B_{n}$ converges to $\sum_{j=1}^{r} \langle\beta,\phi_{j1}\rangle^{2}/\lambda_{j1}^{2}$ as $n \rightarrow \infty$. This limit is positive if and only if at least one of the $\langle\beta,\phi_{j1}\rangle$'s is non-zero. Thus, 
$$n^{2\gamma}\{\mathrm{MSE}(\widetilde{\beta}_{TR}) - \mathrm{MSE}(\widetilde{\beta}_{HR})\} > B(n) + o(1)$$
as $n \rightarrow \infty$.
\end{proof}

\begin{proof}[Proof of Theorem \ref{thm-data1}]

As in the oracle case, define $Y_{i} = \sum_{j=1}^{r} \langle X_{i},\phi_{j}\rangle \phi_{j}$ and $Z_{i} = X_{i} - Y_{i}$ for all $i=1,2,\ldots,n$ (these random variables are not observed in practice). By choice of the $Y_{i}$'s and the $Z_{i}$'s, their population covariance operators are $\mathscr{K}_{1} = \sum_{j=1}^{r} \lambda_{j} \phi_{j} \otimes \phi_{j}$ and $\mathscr{K}_{2} = \sum_{j=r+1}^{r} \lambda_{j} \phi_{j} \otimes \phi_{j}$, respectively. So, the corresponding eigenspaces are orthogonal. Also, define $\mathscr{K}_{\rho,2} = \mathscr{K}_{2} + \rho\mathscr{P}_{2}$. 

 Now, observe that
\begin{eqnarray*}
\hat{\mathscr{K}}_{\rho,2}^{-} - \mathscr{K}_{\rho,2}^{-} = \hat{\mathscr{K}}_{\rho}^{-1} - \mathscr{K}_{\rho}^{-1} + \sum_{j=1}^{r} (\lambda_{j}+\rho)^{-1} \phi_{j} \otimes \phi_{j} - \sum_{j=1}^{r} (\hat{\lambda}_{j}+\rho)^{-1} \hat{\phi}_{j} \otimes \hat{\phi}_{j}.
\end{eqnarray*}
Define $\hat{\mathscr{F}}_{a} = \sum_{j=1}^{r} (\lambda_{j}+a)^{-1} \phi_{j} \otimes \phi_{j} - \sum_{j=1}^{r} (\hat{\lambda}_{j}+a)^{-1} \hat{\phi}_{j} \otimes \hat{\phi}_{j}$ for any $a \geq 0$. In this notation, 
$\hat{\mathscr{K}}_{1}^{-} - \mathscr{K}_{1}^{-} = -\hat{\mathscr{F}}_{0}$. Also, $\hat{C}_{1} - \widetilde{C}_{1} = (\hat{\mathscr{P}}_{1} - \mathscr{P}_{1})\hat{C}$ and $\hat{C}_{2} - \widetilde{C}_{2} = (\hat{\mathscr{P}}_{2} - \mathscr{P}_{2})\hat{C} = (\mathscr{P}_{1} - \hat{\mathscr{P}}_{1})\hat{C}$. 

 Note that 
\begin{eqnarray*}
&& \hat{\beta}_{HR} = \widetilde{\beta}_{HR} + \sum_{l=1}^{8} U_{l}, \\
\mbox{where} && U_{1} = (\hat{\mathscr{K}}_{1}^{-} - \mathscr{K}_{1}^{-})(\hat{C}_{1} - \widetilde{C}_{1})  = -\hat{\mathscr{F}}_{0}(\hat{\mathscr{P}}_{1} - \mathscr{P}_{1})\hat{C} \\
&& U_{2} = (\hat{\mathscr{K}}_{\rho,2}^{-} - \mathscr{K}_{\rho,2}^{-})(\hat{C}_{2} - \widetilde{C}_{2}) = (\hat{\mathscr{K}}_{\rho}^{-1} - \mathscr{K}_{\rho}^{-1})(\mathscr{P}_{1} - \hat{\mathscr{P}}_{1})\hat{C} + \hat{\mathscr{F}}_{\rho}(\mathscr{P}_{1} - \hat{\mathscr{P}}_{1})\hat{C} \\
&& U_{3} = \mathscr{K}_{1}^{-}(\hat{C}_{1} - \widetilde{C}_{1}) = \mathscr{K}_{1}^{-}(\hat{\mathscr{P}}_{1} - \mathscr{P}_{1})\hat{C} \\
&& U_{4} = \mathscr{K}_{\rho,2}^{-}(\hat{C}_{2} - \widetilde{C}_{2}) = \mathscr{K}_{\rho,2}^{-}(\mathscr{P}_{1} - \hat{\mathscr{P}}_{1})\hat{C}   \\
&& U_{5} = (\hat{\mathscr{K}}_{1}^{-} - \mathscr{K}_{1}^{-})(\widetilde{C}_{1} - C_{1}) = -\hat{\mathscr{F}}_{0}(\widetilde{C}_{1} - C_{1})  \\
&& U_{6} = (\hat{\mathscr{K}}_{\rho,2}^{-} - \mathscr{K}_{\rho,2}^{-})(\widetilde{C}_{2} - C_{2}) = (\hat{\mathscr{K}}_{\rho}^{-1} - \mathscr{K}_{\rho}^{-1})(\widetilde{C}_{2} - C_{2}) + \hat{\mathscr{F}}_{\rho}(\widetilde{C}_{2} - C_{2})  \\
&& U_{7} = (\hat{\mathscr{K}}_{1}^{-} - \mathscr{K}_{1}^{-})C_{1} = -\hat{\mathscr{F}}_{0}C_{1} \\
&& U_{8} = (\hat{\mathscr{K}}_{\rho,2}^{-} - \mathscr{K}_{\rho,2}^{-})C_{2} = (\hat{\mathscr{K}}_{\rho}^{-1} - \mathscr{K}_{\rho}^{-1})C_{2} + \hat{\mathscr{F}}_{\rho}C_{2}.
\end{eqnarray*}
Putting the pieces together, we get 
\begin{eqnarray*}
&& \mathbb{E}\{(\hat{\beta}_{HR} - \beta) \otimes (\hat{\beta}_{HR} - \beta)\} \\
&=& \mathbb{E}\{(\widetilde{\beta}_{HR} - \beta) \otimes (\widetilde{\beta}_{HR} - \beta)\} + \mathbb{E}\left\{\left(\sum_{l=1}^{8} U_{l}\right) \otimes \left(\sum_{l=1}^{8} U_{l}\right)\right\} \\
&& + \ \mathbb{E}\left\{(\widetilde{\beta}_{HR} - \beta) \otimes \left(\sum_{l=1}^{8} U_{l}\right)\right\} + \mathbb{E}\left\{\left(\sum_{l=1}^{8} U_{l}\right) \otimes (\widetilde{\beta}_{HR} - \beta)\right\}. \\
\end{eqnarray*}
So,
\begin{eqnarray*}
\mathrm{MSE}(\hat{\beta}_{HR}) = \mathrm{MSE}(\widetilde{\beta}_{HR}) + \mathbb{E}\left(\left\|\sum_{l=1}^{8} U_{l}\right\|^{2}\right) +  2\mathbb{E}\left(\langle\sum_{l=1}^{8} U_{l},\widetilde{\beta}_{HR} - \beta\rangle\right).
\end{eqnarray*}
Using the Cauchy-Schwarz inequality, one has
\begin{eqnarray}
&& |\mathrm{MSE}(\hat{\beta}_{HR}) - \mathrm{MSE}(\widetilde{\beta}_{HR})| \label{estm-HR-3} \\
&& \leq O(1)\left[\sum_{l=1}^{8} \mathbb{E}(\|U_{l}\|^{2}) + \mathbb{E}^{1/2}\{\|\widetilde{\beta}_{HR} - \beta\|^{2}\}\left\{\sum_{l=1}^{8} \mathbb{E}(\|U_{l}\|^{2})\right\}^{1/2} \right]. \nonumber
\end{eqnarray}

Note that $\mathbb{E}\{\|\widetilde{\beta}_{HR} - \beta\|^{2}\} = \mathrm{MSE}(\widetilde{\beta}_{HR})$ so that it can be obtained from the general expression in equation (\ref{mse-HR-1}). The second term in the right hand side of equation (\ref{mse-HR-1}) equals (by the definition of $Z_{1}$ and the assumptions in the theorem)
\begin{eqnarray}
&& n^{-1} \sum_{j=r+1}^{\infty} (\lambda_{j} + \rho)^{-2} \left[\langle\beta,\phi_{j}\rangle^{2}\lambda_{j}^{2} + \lambda_{j}\{\langle\mathscr{K}\beta,\beta\rangle + \sigma^{2}\} \right] \label{estm-HR-5} 
= O(n^{-1}) \left[ 1 + \sum_{j=r+1}^{\infty} \frac{\lambda_{j}}{\lambda_{j}+\rho^{2}} \right] = \frac{O(1)}{n\rho^{1+\frac{1}{\alpha}}}. 
\end{eqnarray}
Here, the last equality follows from Lemma \ref{lemma-2} by taking $a = 1, b = 2$ and $c = 0$ in the statement of that lemma. It also follows from Lemma \ref{lemma-2} by taking $a = 0, b=2$ and $c = 1$ that 
\begin{eqnarray}  \label{estm-HR-55}
\sum_{j=r+1}^{\infty} (\lambda_{j} + \rho)^{-2} \langle\beta,\phi_{j}\rangle^{2} = O(\rho^{L}),
\end{eqnarray}
where $L = (2\eta-1)/\alpha - 2$ or $L = 0$ according as $2\eta < 2\alpha + 1$ or $2\eta > 2\alpha + 1$. Put $m = L + 2$. So, the third term in the right hand side of equation (\ref{mse-HR-1}) is $O(\rho^{m})$, where $m = (2\eta-1)/\alpha$ or $m = 2$ according as $\alpha > \eta - 1/2$ or $\alpha < \eta - 1/2$. Combining this bound with \eqref{estm-HR-5} and the fact that first term in the right hand side of equation (\ref{mse-HR-1}) is $O(n^{-1})$, we obtain
\begin{eqnarray}
\mathbb{E}\{\|\widetilde{\beta}_{HR} - \beta\|^{2}\} = O(1)\left\{\frac{1}{n\rho^{1+\frac{1}{\alpha}}} + \rho^{m}\right\},  \label{estm-HR-6}
\end{eqnarray}
where $m = (2\eta-1)/\alpha$ or $m = 2$ depending on whether $\alpha > \eta - 1/2$ or $\alpha < \eta - 1/2$.

 We will now consider bounds for $\mathbb{E}(\|U_{l}\|^{2})$ for $l = 1,2,\ldots, 8$. First note that for any $a \geq 0$, we have
\begin{eqnarray*}
\hat{\mathscr{F}}_{a} &=& -\sum_{j=1}^{r} \{(\hat{\lambda}_{j}+a)^{-1} - (\lambda_{j}+a)^{-1}\} (\hat{\phi}_{j} \otimes \hat{\phi}_{j}) - \\
&& \ \sum_{j=1}^{r} (\lambda_{j}+a)^{-1} \{\hat{\phi}_{j} \otimes (\hat{\phi}_{j} - \phi_{j}) + (\hat{\phi}_{j} - \phi_{j}) \otimes \phi_{j}\} \\ 
\Rightarrow \ \ \hs\hat{\mathscr{F}}_{a}\hs &\leq& \sum_{j=1}^{r} |(\hat{\lambda}_{j}+a)^{-1} - (\lambda_{j}+a)^{-1}| + 2\sum_{j=1}^{r} (\lambda_{j}+a)^{-1} \|\hat{\phi}_{j} - \phi_{j}\|.
\end{eqnarray*}
Some straightforward but tedious moment calculations yield $\mathbb{E}\{\hs\hat{\mathscr{K}} - \mathscr{K}\hs^{8}\} = O(n^{-4})$ so that $\mathbb{E}\{\hs\hat{\mathscr{K}} - \mathscr{K}\hs^{4}\} = O(n^{-2})$. Thus, using Lemma 2.2 and 2.3 in \citet{HK12}, we have that for any $a \geq 0$
\begin{eqnarray}  \label{eq1}
\mathbb{E}\{\hs\hat{\mathscr{F}}_{a}\hs^{4}\} = O(n^{-2})
\end{eqnarray}
as $n \rightarrow \infty$. We will use this fact often in the proof. We will also use the fact that 
\begin{eqnarray}
&& \mathbb{E}\{\hs\hat{\mathscr{P}}_{1} - \mathscr{P}_{1}\hs^{8}\} \label{eq2} \\
&\leq& \mathbb{E}\{\hs\sum_{j=1}^{r} \{\hat{\phi}_{j} \otimes (\hat{\phi}_{j} - \phi_{j}) + (\hat{\phi}_{j} - \phi_{j}) \otimes \phi_{j}\}\hs^{8}\} \nonumber \\
&\leq& O(1) \sum_{j=1}^{r} \mathbb{E}\{\|\hat{\phi}_{j} - \phi_{j}\|^{8}\} \ \leq \  O(1)\mathbb{E}\{\hs\hat{\mathscr{K}} - \mathscr{K}\hs^{8}\} = O(n^{-4})  \nonumber
\end{eqnarray}
as $n \rightarrow \infty$. The third inequality follows from Lemma 2.3 in \citet{HK12}.

 Note that $\mathbb{E}(\|U_{1}\|^{2}) \leq O(1)E^{1/2}\{\hs\hat{\mathscr{F}}_{0}\hs_{\infty}^{4}\}E^{1/4}\{\hs\hat{\mathscr{P}}_{1} - \mathscr{P}_{1}\hs_{\infty}^{8}\}E^{1/4}\{\|\hat{C}\|^{8}\}$. It directly follows that $\mathbb{E}\{\|\widetilde{C}\|^{8}\} = O(1)$ as $n \rightarrow \infty$. Thus using \eqref{eq1} and \eqref{eq2} along with the fact that the operator norm is bounded above by the Hilbert-Schmidt norm, we have
\begin{eqnarray}
\mathbb{E}(\|U_{1}\|^{2}) = O(n^{-2})  \label{estm-HR-7}
\end{eqnarray}
as $n \rightarrow \infty$. 

 Next note that $\mathbb{E}(\|U_{3}\|^{2}) \leq \hs\mathscr{K}_{1}^{-}\hs_{\infty}^{2}E^{1/2}\{\hs\hat{\mathscr{P}}_{1} - \mathscr{P}_{1}\hs^{4}\}E^{1/2}\{\|\hat{C}\|^{4}\}$. Using the fact that $\hs\mathscr{K}_{1}^{-}\hs_{\infty} = \lambda_{r}^{-1}$, we get that
\begin{eqnarray}
\mathbb{E}(\|U_{3}\|^{2}) = O(n^{-1})  \label{estm-HR-9}
\end{eqnarray}
as $n \rightarrow \infty$.

 Next note that $\mathbb{E}(\|U_{5}\|^{2}) \leq E^{1/2}\{\hs\hat{\mathscr{F}}_{0}\hs^{4}\}E^{1/2}\{\|\widetilde{C}_{1} - C_{1}\|^{4}\}$. It is easy to show that $\mathbb{E}\{\|\widetilde{C}_{1} - C_{1}\|^{4}\} = O(n^{-2})$ as $n \rightarrow \infty$. So, it follows from \eqref{eq1} that 
\begin{eqnarray}
\mathbb{E}(\|U_{5}\|^{2}) = O(n^{-2}). \label{estm-HR-10}
\end{eqnarray}
Similar calculations also show that
\begin{eqnarray}
\mathbb{E}(\|U_{7}\|^{2}) = O(n^{-1})  \label{estm-HR-11}
\end{eqnarray}
as $n \rightarrow \infty$. Next, observe that 
\begin{eqnarray} 
&& \mathbb{E}(\|U_{6}\|^{2}) \leq 2E^{1/2}\{\hs\hat{\mathscr{F}}_{\rho}\hs^{4}\}E^{1/2}\{\|
  \widetilde{C}_{2} - C_{2}\|^{4}\} \label{eqn-modf1} \\
&& \hspace{2cm}  + 2E\{\|(\hat{\mathscr{K}}_{\rho}^{-1} - \mathscr{K}_{\rho}^{-1})
  (\widetilde{C}_{2} - C_{2})\|^{2}\}.  \nonumber
\end{eqnarray}  
From the fact that $\mathbb{E}\{\|\widetilde{C}_{2} - C_{2}\|^{4}\} = O(n^{-2})$ as $n \rightarrow \infty$ and using \eqref{eq1}, it follows that the first term on the right hand side of \eqref{eqn-modf1} is $O(n^{-2})$ as $n \rightarrow \infty$. Further, 
\begin{eqnarray}
&& E\{\|(\hat{\mathscr{K}}_{\rho}^{-1} - \mathscr{K}_{\rho}^{-1})(\widetilde{C}_{2} - C_{2}\|^{2}\} \label{eqn-modf2} \\
&\leq& E\{\hs{\mathscr{K}}_{\rho}^{-1}\hs_{\infty}\|(\hat{\mathscr{K}} - \mathscr{K})\mathscr{K}_{\rho}^{-1}(\widetilde{C}_{2} - C_{2}\|^{2}\} \nonumber \\
&\leq& \rho^{-2}E^{1/2}\{\hs\hat{\mathscr{K}} - \mathscr{K}\hs^{4}\}E^{1/2}\{\|\mathscr{K}_{\rho}^{-1}(\widetilde{C}_{2} - C_{2}\|^{4}\} \nonumber \\
&\leq& O(n^{-1}\rho^{-2}) \left[ E^{1/2}\{\|\mathscr{K}_{\rho,1}^{-}(\widetilde{C}_{2} - C_{2}\|^{4}\} + E^{1/2}\{\|\mathscr{K}_{\rho,2}^{-}(\widetilde{C}_{2} - C_{2}\|^{4}\} \right], \nonumber 
\end{eqnarray} 
where $\mathscr{K}_{\rho,1}^{-} = \sum_{j=1}^{r} (\lambda_{j} + \rho)^{-1} (\phi_{j} \otimes \phi_{j})$. The last inequality also uses the bound for $\mathbb{E}\{\hs\hat{\mathscr{K}} - \mathscr{K}\hs_{\infty}^{4}\}$ obtained for deriving \eqref{estm-HR-8}. 

Now, using the fact that for any $j = r+1, r+2, \ldots$, we have $E\{\langle\widetilde{C}_{2} - C_{2},\phi_{j}\rangle^{4}\} = O(n^{-2})E^{2}\{\langle y_{1}Z_{1} - \mathscr{K}_{2}\beta,\phi_{j}\rangle^{2}\} = O(n^{-2})\{\langle\beta,\phi_{j}\rangle^{2}\lambda_{j}^{2} + \lambda_{j}(\sigma^{2} + \langle \mathscr{K}\beta,\beta\rangle)\}^{2}$, we get that
\begin{eqnarray*}
&& E\{\|\mathscr{K}_{\rho,2}^{-}(\widetilde{C}_{2} - C_{2}\|^{4}\} \\
&=& E\left[ \left\{ \sum_{j=r+1}^{\infty} (\lambda_{j}+\rho)^{-2} \langle\widetilde{C}_{2} - C_{2},\phi_{j}\rangle^{2} \right\}^{2} \right] \\
&=& E\left\{ \sum_{j_{1},j_{2}=r+1}^{\infty} (\lambda_{j_{1}}+\rho)^{-2}(\lambda_{j_{2}}+\rho)^{-2}\langle\widetilde{C}_{2} - C_{2},\phi_{j_{1}}\rangle^{2}\langle\widetilde{C}_{2} - C_{2},\phi_{j_{2}}\rangle^{2} \right\} \\
&\leq& \sum_{j_{1},j_{2}=r+1}^{\infty} (\lambda_{j_{1}}+\rho)^{-2}(\lambda_{j_{2}}+\rho)^{-2} E^{1/2}\{\langle\widetilde{C}_{2} - C_{2},\phi_{j_{1}}\rangle^{4}\}E^{1/2}\{\langle\widetilde{C}_{2} - C_{2},\phi_{j_{2}}\rangle^{4}\} \\
&\leq& O(n^{-2})\left[\sum_{j=r+1}^{\infty} (\lambda_{j}+\rho)^{-2} \{\langle\beta,\phi_{j}\rangle^{2}\lambda_{j}^{2} + \lambda_{j}(\sigma^{2} + \langle \mathscr{K}\beta,\beta\rangle)\}\right]^{2} \\
&\leq& O(n^{-2}) \left[ 1 + \left\{\sum_{j=r+1}^{\infty} (\lambda_{j} + \rho)^{-2}\lambda_{j}\right\}^{2}\right] = O(n^{-2}\rho^{-2-2/\alpha}),
\end{eqnarray*}
by an application of Lemma \ref{lemma-2}. Now, using \eqref{eqn-modf1} and \eqref{eqn-modf2}, we get that
\begin{eqnarray}
\mathbb{E}(\|U_{6}\|^{2}) = o(n^{-1}\rho^{-1-1/\alpha})  \label{estm-HR-10-1}
\end{eqnarray}
as $n \rightarrow \infty$.

 Next note that $\mathbb{E}(\|U_{8}\|^{2}) \leq \mathbb{E}\{\|(\hat{\mathscr{K}}_{\rho}^{-1} - \mathscr{K}_{\rho}^{-1})C_{2}\|^{2}\} + \mathbb{E}\{\|\hat{\mathscr{F}}_{\rho}C_{2}\|^{2}\}$. From earlier calculations and using \eqref{eq1}, it follows that $\mathbb{E}\{\|\hat{\mathscr{F}}_{\rho}C_{2}\|^{2}\} = O(n^{-1})$ as $n \rightarrow \infty$. Next, note that $\mathbb{E}\{\|(\hat{\mathscr{K}}_{\rho}^{-1} - \mathscr{K}_{\rho}^{-1})C_{2}\|^{2}\} = \mathbb{E}\{\|\hat{\mathscr{K}}_{\rho}^{-1}(\hat{\mathscr{K}} - \mathscr{K})\mathscr{K}_{\rho}^{-1}C_{2}\|^{2}\} \leq E^{1/2}\{\hs\hat{\mathscr{K}}_{\rho}^{-1}\mathscr{K}_{\rho}\hs_{\infty}^{4}\}E^{1/2}\{\|\mathscr{K}_{\rho}^{-1}(\hat{\mathscr{K}} - \mathscr{K})\mathscr{K}_{\rho}^{-1}C_{2}\|^{4}\}$. Observe that $\hs\hat{\mathscr{K}}_{\rho}^{-1}\mathscr{K}_{\rho}\hs_{\infty} = \hs\hat{\mathscr{K}}_{\rho}^{-1}(\mathscr{K}_{\rho} - \hat{\mathscr{K}}_{\rho}) + \mathscr{I}\hs_{\infty} \leq \rho^{-1}\hs\mathscr{K} - \hat{\mathscr{K}}\hs_{\infty} + 1 \leq \rho^{-1}\hs\hat{\mathscr{K}} - \mathscr{K}\hs + 1$. So, we have $\mathbb{E}\{\hs\hat{\mathscr{K}}_{\rho}^{-1}\mathscr{K}_{\rho}\hs_{\infty}^{4}\} \leq 1 + \rho^{-4}\mathbb{E}\{\hs\hat{\mathscr{K}} - \mathscr{K}\hs^{4}\} = 1 + O(n^{-2}\rho^{-4}) = O(1)$ since $n\rho^{2} \rightarrow \infty$. We have
\begin{eqnarray*}
 \mathbb{E}\{\|\mathscr{K}_{\rho}^{-1}(\hat{\mathscr{K}} -  \mathscr{K})\mathscr{K}_{\rho}^{-1}C_{2}\|^{4}\} 
&=& \mathbb{E}\left[ \sum_{j=1}^{\infty} \langle\mathscr{K}_{\rho}^{-1}(\hat{\mathscr{K}} -  \mathscr{K})\mathscr{K}_{\rho}^{-1}C_{2},\phi_{j}\rangle^{2} \right]^{2} \\
&=& \mathbb{E}\left[\sum_{j=1}^{\infty} \langle(\hat{\mathscr{K}} - \mathscr{K})\mathscr{K}_{\rho}^{-1}C_{2},\mathscr{K}_{\rho}^{-1}\phi_{j}\rangle^{2} \right]^{2} \\
&=& \mathbb{E}\left[\sum_{j=1}^{\infty} \langle(\hat{\mathscr{K}} -  \mathscr{K})\mathscr{K}_{\rho}^{-1}\mathscr{K}_{2}\beta,(\lambda_{j} + \rho)^{-1}\phi_{j}\rangle^{2} \right]^{2}
\end{eqnarray*}
\noindent We denote the above expectation by $T$. Now,
\begin{eqnarray}
&& T = \sum_{j_{1},j_{2} = 1}^{\infty} \sum_{l_{1},l_{2},l_{3},l_{4} = r+1}^{\infty} \left[\frac{\prod_{u=1}^{4}(\langle\beta,\phi_{l_{u}}\rangle\lambda_{l_{u}})}{(\lambda_{j_{1}}+\rho)^{2}(\lambda_{j_{2}}+\rho)^{2}\prod_{u=1}^{4} (\lambda_{l_{u}}+\rho)} \times \right. \label{estm} \\
&& \hspace{2cm} \left. \mathbb{E}\left\{\prod_{i=1}^{2}\langle(\hat{\mathscr{K}} - \mathscr{K})\phi_{j_{1}},\phi_{l_{u}}\rangle\prod_{i=3}^{4}\langle(\hat{\mathscr{K}} - \mathscr{K})\phi_{j_{2}},\phi_{l_{u}}\rangle\right\}\right].  \nonumber
\end{eqnarray}
\noindent Direct calculation yields that if $j_{1} = j_{2}$ in the expression of $T$ above, then
\begin{eqnarray*}
&& \mathbb{E}\left\{\prod_{i=1}^{2}\langle(\hat{\mathscr{K}} - \mathscr{K})\phi_{j_{1}},\phi_{l_{u}}\rangle\prod_{i=3}^{4}\langle(\hat{\mathscr{K}} - \mathscr{K})\phi_{j_{2}},\phi_{l_{u}}\rangle\right\} \\
&\leq& O(n^{-2})\left\{[\lambda_{j_{1}}^{2}\mathbf{1}\{j_{1} = l_{1} = l_{2}\} + \lambda_{j_{1}}\lambda_{l_{1}}\mathbf{1}\{j_{1} \neq l_{1} = l_{2}\}] \times \right. \\
&& \hspace{2cm} \left.[\lambda_{j_{1}}^{2}\mathbf{1}\{j_{1} = l_{3} = l_{4}\} + \lambda_{j_{1}}\lambda_{l_{3}}\mathbf{1}\{j_{1} \neq l_{3} = l_{4}\}]\right\}.
\end{eqnarray*} 
\noindent On the other hand if $j_{1} \neq j_{2}$, then
\begin{eqnarray*}
&& \mathbb{E}\left\{\prod_{i=1}^{2}\langle(\hat{\mathscr{K}} - \mathscr{K})\phi_{j_{1}},\phi_{l_{u}}\rangle\prod_{i=3}^{4}\langle(\hat{\mathscr{K}} - \mathscr{K})\phi_{j_{2}},\phi_{l_{u}}\rangle\right\} \\
&\leq& O(n^{-2})\left\{[\lambda_{j_{1}}^{2}\mathbf{1}\{j_{1} = l_{1} = l_{2}\} + \lambda_{j_{1}}\lambda_{l_{1}}\mathbf{1}\{j_{1} \neq l_{1} = l_{2}\}]\times \right. \\
&& \hspace{2cm} \left. [\lambda_{j_{2}}^{2}\mathbf{1}\{j_{2} = l_{3} = l_{4}\} + \lambda_{j_{2}}\lambda_{l_{3}}\mathbf{1}\{j_{2} \neq l_{3} = l_{4}\}]  \right. \\
&& \hspace{2cm} \left. + \ \lambda_{j_{1}}^{2}\lambda_{j_{2}}^{2}\mathbf{1}\{j_{1} = l_{3} = l_{4}\}\mathbf{1}\{j_{2} = l_{1} = l_{2}\}\right\}
\end{eqnarray*}
So, we have
\begin{eqnarray*}
&& T = O(n^{-2})\left[ \left(\sum_{j_{1}=1}^{r}  \frac{\lambda_{j_{1}}}{(\lambda_{j_{1}}+\rho)^{2}}\right)^{2} \times \left(\sum_{l_{1},l_{2}=r+1}^{\infty} \frac{\langle\beta,\phi_{l_{1}}\rangle^{2}\langle\beta,\phi_{l_{2}}\rangle^{2}\lambda_{l_{1}}^{3}\lambda_{l_{3}}^{2}}{(\lambda_{l_{1}}+\rho)^{2}(\lambda_{l_{2}}+\rho)^{2}}\right) + \right. \\
&& \hspace{2cm} \left. 2\sum_{j_{1} = 1}^{r} \sum_{j_{2} = r+1}^{\infty} \sum_{l_{1}=r+1}^{\infty} \frac{\langle\beta,\phi_{l_{1}}\rangle^{2}\langle\beta,\phi_{j_{2}}\rangle^{2}\lambda_{l_{1}}^{3}\lambda_{j_{2}}^{4}\lambda_{j_{1}}}{(\lambda_{l_{1}}+\rho)^{2}(\lambda_{j_{1}}+\rho)^{2}(\lambda_{j_{2}}+\rho)^{4}} +   \right. \\
&& \hspace{2cm} \left. 4\sum_{j_{1} = 1}^{r} \sum_{j_{2} = r+1}^{\infty} \sum_{\substack{l_{1},l_{2}=r+1\\l_{2} \neq j_{2}}}^{\infty} \frac{\langle\beta,\phi_{l_{1}}\rangle^{2}\langle\beta,\phi_{l_{2}}\rangle^{2}\lambda_{l_{1}}^{3}\lambda_{l_{2}}^{3}\lambda_{j_{1}}\lambda_{j_{2}}}{(\lambda_{l_{1}}+\rho)^{2}(\lambda_{j_{1}}+\rho)^{2}(\lambda_{l_{2}}+\rho)^{2}(\lambda_{j_{2}}+\rho)^{2}} +  \right. \\
&& \hspace{2cm} \left. \sum_{j_{1}=r+1}^{\infty} \frac{\langle\beta,\phi_{j_{1}}\rangle^{4}\lambda_{j_{1}}^{4}}{(\lambda_{j_{1}} + \rho)^{8}} + 2\sum_{\substack{j_{1},l_{1}=r+1 \\ j_{1} \neq l_{1}}}^{\infty} \frac{\langle\beta,\phi_{j_{1}}\rangle^{2}\langle\beta,\phi_{l_{1}}\rangle^{2}\lambda_{j_{1}}^{5}\lambda_{l_{1}}^{3}}{(\lambda_{j_{1}}+\rho)^{6}(\lambda_{l_{1}}+\rho)^{2}} + \right. \\
&& \hspace{2cm} \left. \sum_{\substack{j_{1},l_{1},l_{2}=r+1 \\ j_{1} \neq l_{1}, j_{1} \neq l_{2}}}^{\infty} \frac{\langle\beta,\phi_{l_{1}}\rangle^{2}\langle\beta,\phi_{l_{2}}\rangle^{2}\lambda_{j_{1}}^{2}\lambda_{l_{1}}^{3}\lambda_{l_{2}}^{3}}{(\lambda_{j_{1}}+\rho)^{4}(\lambda_{l_{1}}+\rho)^{2}(\lambda_{l_{2}}+\rho)^{2}} + \right. \\
&& \hspace{2cm} \left. 2\sum_{\substack{j_{1},j_{2}=r+1 \\ j_{1} \neq j_{2}}}^{\infty} \frac{\langle\beta,\phi_{j_{1}}\rangle^{2}\langle\beta,\phi_{j_{2}}\rangle^{2}\lambda_{j_{1}}^{4}\lambda_{j_{2}}^{4}}{(\lambda_{j_{1}}+\rho)^{4}(\lambda_{j_{2}}+\rho)^{4}} + \right. \\
&& \hspace{2cm} \left. 2\sum_{\substack{j_{1},j_{2},l_{1}=r+1 \\ j_{1} \neq j_{2}, j_{2} \neq l_{1}}}^{\infty} \frac{\langle\beta,\phi_{j_{1}}\rangle^{2}\langle\beta,\phi_{l_{2}}\rangle^{2}\lambda_{j_{1}}^{4}\lambda_{j_{2}}\lambda_{l_{2}}^{3}}{(\lambda_{j_{1}}+\rho)^{4}(\lambda_{j_{2}}+\rho)^{2}(\lambda_{l_{2}}+\rho)^{2}} + \right. \\
&& \hspace{2cm} \left. 2\sum_{\substack{j_{1},j_{2},l_{1},l_{2}=r+1 \\ j_{1} \neq j_{2}, j_{1} \neq l_{1}, j_{2} \neq l_{2}}}^{\infty} \frac{\langle\beta,\phi_{l_{1}}\rangle^{2}\langle\beta,\phi_{l_{2}}\rangle^{2}\lambda_{j_{1}}\lambda_{j_{2}}\lambda_{l_{1}}^{3}\lambda_{l_{2}}^{3}}{(\lambda_{j_{1}}+\rho)^{2}(\lambda_{j_{2}}+\rho)^{2}(\lambda_{l_{1}}+\rho)^{2}(\lambda_{l_{2}}+\rho)^{2}} \right].
\end{eqnarray*}
Using the simple bound that $\langle \beta,\phi_{l}\rangle^{2} \leq \|\beta\|^{2}$ and applying Lemma \ref{lemma-2} to the above expression with $c = 0$ and appropriately chosen $a$ and $b$ for each infinite sum, we get that $T = O(n^{-2}\rho^{-2-2/\alpha})$ as $n \rightarrow \infty$. This together with the fact that $\mathbb{E}\{\hs\hat{\mathscr{K}}_{\rho}^{-1}\mathscr{K}_{\rho}\hs_{\infty}^{4}\} = O(1)$ as $n \rightarrow \infty$ implies that $\mathbb{E}\{\|(\hat{\mathscr{K}}_{\rho}^{-1} - \mathscr{K}_{\rho}^{-1})C_{2}\|^{2}\} = O(n^{-1}\rho^{-1-1/\alpha})$ as $n \rightarrow \infty$. So, we have
\begin{eqnarray}
\mathbb{E}(\|U_{8}\|^{2}) = O(n^{-1}\rho^{-1-1/\alpha})  \label{estm-HR-12}
\end{eqnarray}
 as $n \rightarrow \infty$.

 We now turn to controlling $\mathbb{E}(\|U_{4}\|^{2})$. First we decompose $U_{4}$ as 
$$U_{4} = \mathscr{K}_{\rho,2}^{-}(\mathscr{P}_{1} - \hat{\mathscr{P}}_{1})(\hat{C} - C) + \mathscr{K}_{\rho,2}^{-}(\mathscr{P}_{1} - \hat{\mathscr{P}}_{1})C,$$
and denote the first and the second terms by $U_{41}$ and $U_{42}$, respectively. Calculations similar to those carried out earlier yield $\mathbb{E}(\|U_{41}\|^{2}) = O(n^{-2}\rho^{-2})$ as $n \rightarrow \infty$.

To bound $\mathbb{E}(\|U_{42}\|^{2})$, set $M_{n}^{2} = An^{-1}\rho^{-2}$ for some $A > 0$, and define the set 
\begin{eqnarray*}
G_{n} = \left\{\max_{j=1,2,\ldots, r} |\hat{\lambda}_{j} - \lambda_{j}| \leq M_{n}\right\}.
\end{eqnarray*}
Since $\max_{j=1,2,\ldots,r} \mathbb{E}\{(\hat{\lambda}_{j} - \lambda_{j})^{2}\} = O(n^{-1})$ as $n \rightarrow \infty$, Markov's inequality yields$P(G_{n}^{c}) < \rho^{2}$ as $n \rightarrow \infty$ for an appropriate choice of $A$. Thus, $\mathbb{E}\{\|U_{42}\|^{2}\mathbf{1}(G_{n}^{c})\} \leq \rho^{-2}E^{1/2}\{\hs\hat{\mathscr{P}}_{1} - \mathscr{P}_{1}\hs^{4}\}\sqrt{P(G_{n}^{c})} \leq \rho^{-1}E^{1/2}\{\hs\hat{\mathscr{K}} - \mathscr{K}\hs^{4}\} = O(n^{-1}\rho^{-1}) = o(n^{-1}\rho^{-1-1/\alpha})$ as $n \rightarrow \infty$. Consequently, it suffices to bound $\mathbb{E}\{\|U_{42}\|^{2}\mathbf{1}(G_{n})\}$.

\noindent Using the resolvent formalism, we represent $\mathscr{P}_1$ as
\begin{eqnarray*}
\mathscr{P}_{1} = \frac{1}{2{\pi}i}\int_{\Gamma} (\mathscr{K} - z\mathscr{I})^{-1} dz,
\end{eqnarray*}
where $i^2=-1$ and $\Gamma$ is the boundary of a closed disk containing $\{\lambda_{j}:j=1,...,r\}$ and excluding $\{\lambda_{j}:j>r\}$ (see \citet{HE15}). Similarly, 
\begin{eqnarray*}
\hat{\mathscr{P}}_{1} = \frac{1}{2{\pi}i}\int_{\hat{\Gamma}} (\hat{\mathscr{K}} - z\mathscr{I})^{-1} dz,
\end{eqnarray*}
where $\hat{\Gamma}$ is the boundary of a closed disk containing $\{\hat\lambda_{j}:j=1,...,r\}$ and excluding $\{\hat\lambda_{j}:j>r\}$. Since $M_{n} \rightarrow 0$ as $n \rightarrow \infty$, so for all sufficiently large $n$, $M_{n} < (\lambda_{r} - \lambda_{r+1})/4$. Thus, for all sufficiently large $n$, $\hat{\Gamma}$ can be chosen to be $\Gamma$ for all sample points in the set $G_{n}$. Thus, for all sufficiently large $n$, we have 
\begin{eqnarray}
&& \mathbb{E}\{\|U_{42}\|^{2}\mathbf{1}(G_{n})\} \nonumber \\
&=& \mathbb{E}\left\{\left\| \frac{1}{2{\pi}i} \int_{\Gamma} \mathscr{K}_{\rho,2}^{-}[(\hat{\mathscr{K}} - z\mathscr{I})^{-1} - (\mathscr{K} - z\mathscr{I})^{-1}]C dz \right\|^{2}\mathbf{1}(G_{n})\right\} \nonumber \\
&\leq& \mathbb{E}\left\{\left\| \frac{1}{2{\pi}i} \int_{\Gamma} \mathscr{K}_{\rho,2}^{-}(\hat{\mathscr{K}} - z\mathscr{I})^{-1}(\hat{\mathscr{K}} - \mathscr{K})(\mathscr{K} - z\mathscr{I})^{-1}C dz \right\|^{2}\mathbf{1}(G_{n})\right\} \nonumber \\
&\leq& 2\mathbb{E}\left\{\left\| \frac{1}{2{\pi}i} \int_{\Gamma} \mathscr{K}_{\rho,2}^{-}[(\hat{\mathscr{K}} - z\mathscr{I})^{-1} - (\mathscr{K} - z\mathscr{I})^{-1}](\hat{\mathscr{K}} - \mathscr{K})(\mathscr{K} - z\mathscr{I})^{-1}C dz \right\|^{2}\right\} \nonumber \\
&& + \ 2\mathbb{E}\left\{\left\| \frac{1}{2{\pi}i} \int_{\Gamma} \mathscr{K}_{\rho,2}^{-}(\mathscr{K} - z\mathscr{I})^{-1}(\hat{\mathscr{K}} - \mathscr{K})(\mathscr{K} - z\mathscr{I})^{-1}C dz \right\|^{2}\right\} \nonumber \\
&=& 2\mathbb{E}\left\{\left\| \frac{1}{2{\pi}i} \int_{\Gamma} \mathscr{K}_{\rho,2}^{-}(\hat{\mathscr{K}} - z\mathscr{I})^{-1}(\hat{\mathscr{K}} - \mathscr{K})(\mathscr{K} - z\mathscr{I})^{-1}(\hat{\mathscr{K}} - \mathscr{K})(\mathscr{K} - z\mathscr{I})^{-1}C dz \right\|^{2}\right\} \label{eq3} \\
&& + \ 2\mathbb{E}\left\{\left\| \frac{1}{2{\pi}i} \int_{\Gamma} \mathscr{K}_{\rho,2}^{-}(\mathscr{K} - z\mathscr{I})^{-1}(\hat{\mathscr{K}} - \mathscr{K})(\mathscr{K} - z\mathscr{I})^{-1}C dz \right\|^{2}\right\} \nonumber
\end{eqnarray}
Now note that 
\begin{eqnarray*}
&& \left\| \frac{1}{2{\pi}i} \int_{\Gamma} \mathscr{K}_{\rho,2}^{-}(\hat{\mathscr{K}} - z\mathscr{I})^{-1}(\hat{\mathscr{K}} - \mathscr{K})(\mathscr{K} - z\mathscr{I})^{-1}(\hat{\mathscr{K}} - \mathscr{K})(\mathscr{K} - z\mathscr{I})^{-1}C dz \right\|^{2} \\
&\leq& \hs\mathscr{K}_{\rho,2}^{-}\hs_{\infty}^{2} \left\| \frac{1}{2{\pi}i} \int_{\Gamma} (\hat{\mathscr{K}} - z\mathscr{I})^{-1}(\hat{\mathscr{K}} - \mathscr{K})(\mathscr{K} - z\mathscr{I})^{-1}(\hat{\mathscr{K}} - \mathscr{K})(\mathscr{K} - z\mathscr{I})^{-1}C dz \right\|^{2} \\
&\leq& \rho^{-2} \frac{L^{2}}{4\pi^{2}} \sup_{\Gamma} \left|(\hat{\mathscr{K}}-z\mathscr{I})^{-1}(\hat{\mathscr{K}} - \mathscr{K})(\mathscr{K} - z\mathscr{I})^{-1}(\hat{\mathscr{K}} - \mathscr{K})(\mathscr{K} - z\mathscr{I})^{-1}C\right|^{2},
\end{eqnarray*}
where $L$ denotes the arc length of the contour $\Gamma$. This last inequality follows from properties of complex contour integrals (see \citet{Conw78}). Now let us note that
\begin{eqnarray*}
&& \sup_{\Gamma} \left|(\hat{\mathscr{K}}-z\mathscr{I})^{-1}(\hat{\mathscr{K}} - \mathscr{K})(\mathscr{K} - z\mathscr{I})^{-1}(\hat{\mathscr{K}} - \mathscr{K})(\mathscr{K} - z\mathscr{I})^{-1}C\right|^{2} \\
&\leq& \hs\hat{\mathscr{K}} - \mathscr{K}\hs^{4} \sup_{\Gamma} \hs(\hat{\mathscr{K}}-z\mathscr{I})^{-1}\hs_{\infty}\}\hs(\mathscr{K} - z\mathscr{I})^{-1}\hs_{\infty}^{2} \\
&=& \hs\hat{\mathscr{K}} - \mathscr{K}\hs^{4} \sup_{\Gamma} |z|^{-3} \leq \mathrm{const.} \hs\hat{\mathscr{K}} - \mathscr{K}\hs^{4},
\end{eqnarray*}
where the last inequality follows because $\Gamma$ only encompasses $\lambda_{1}, \lambda_{2}, \ldots, \lambda_{r}$ and all of them are bounded away from zero. Thus, from the above facts, it follows that the first expectation in the right hand side of \eqref{eq3} is bounded above by $O(1)\rho^{-2}\mathbb{E}\{\hs\hat{\mathscr{K}} - \mathscr{K}\hs^{4}\} = O(n^{-2}\rho^{-2})$ as $n \rightarrow \infty$.

We continue by noting that
\begin{eqnarray*}
&& \mathbb{E}\left|\frac{1}{2{\pi}i}\int_{\Gamma}\langle(\mathscr{K}-z\mathscr{I})^{-1}(\hat{\mathscr{K}} - \mathscr{K})(\mathscr{K} - z\mathscr{I})^{-1}C,\phi_{j}\rangle dz\right|^{2} \\
&=& \frac{1}{(2{\pi}i)^{2}} \int_{\Gamma} \int_{\Gamma} \mathbb{E}\left\{ \langle(\mathscr{K}-z_{1}\mathscr{I})^{-1}(\hat{\mathscr{K}} - \mathscr{K})(\mathscr{K} - z_{1}\mathscr{I})^{-1}\mathscr{K}\beta,\phi_{j}\rangle \times \right. \\
&& \hspace{1cm} \left. \langle(\mathscr{K}-z_{2}\mathscr{I})^{-1}(\hat{\mathscr{K}} - \mathscr{K})(\mathscr{K} - z_{2}\mathscr{I})^{-1}\mathscr{K}\beta,\phi_{j}\rangle dz_{1}dz_{2}\right\} \\
&=& \frac{1}{(2{\pi}i)^{2}} \int_{\Gamma} \int_{\Gamma} (\lambda_{j} - z_{1})^{-1}(\lambda_{j} - z_{2})^{-1} \sum_{l_{1}, l_{2} = 1}^{\infty} \frac{\langle\beta,\phi_{l_{1}}\rangle\langle\beta,\phi_{l_{2}}\rangle\lambda_{l_{1}}\lambda_{l_{2}}}{(\lambda_{l_{1}} - z_{1})(\lambda_{l_{2}} - z_{2})} \times \\
&& \hspace{3cm} \mathbb{E}\{\langle(\hat{\mathscr{K}} - \mathscr{K})\phi_{l_{1}},\phi_{j}\rangle\langle(\hat{\mathscr{K}} - \mathscr{K})\phi_{l_{2}},\phi_{j}\rangle\} dz_{1}dz_{2}  \\
&=& \frac{1}{(2{\pi}i)^{2}} \int_{\Gamma} \int_{\Gamma} (\lambda_{j} - z_{1})^{-1}(\lambda_{j} - z_{2})^{-1} \sum_{l_{1}, l_{2} = 1}^{\infty} \frac{\langle\beta,\phi_{l_{1}}\rangle\langle\beta,\phi_{l_{2}}\rangle\lambda_{l_{1}}\lambda_{l_{2}}}{(\lambda_{l_{1}} - z_{1})(\lambda_{l_{2}} - z_{2})} \times \\
&& \hspace{3cm} n^{-1}\{\lambda_{j}^{2}\mathbf{1}\{j = l_{1} = l_{2}\} + \lambda_{j}\lambda_{l}\mathbf{1}\{j \neq l_{1} = l_{2}\}\} dz_{1}dz_{2} \\
&=& \frac{n^{-1}}{(2{\pi}i)^{2}} \int_{\Gamma} \int_{\Gamma} (\lambda_{j} - z_{1})^{-1}(\lambda_{j} - z_{2})^{-1} \left\{ \sum_{l=1}^{r} \frac{\langle\beta,\phi_{l}\rangle^{2}\rangle\lambda_{l}^{3}\lambda_{j}}{(\lambda_{l} - z_{1})(\lambda_{l} - z_{2})} + \right. \\
&& \hspace{3cm} \left. \sum_{l=r+1}^{\infty} \frac{\langle\beta,\phi_{l}\rangle^{2}\rangle\lambda_{l}^{2}[\lambda_{j}^{2}\mathbf{1}\{j = l_{1} = l_{2}\} + \lambda_{j}\lambda_{l}\mathbf{1}\{j \neq l_{1} = l_{2}\}]}{(\lambda_{l} - z_{1})(\lambda_{l} - z_{2})} \right\} dz_{1}dz_{2}.
\end{eqnarray*}
Since $\Gamma$ does not contain $\lambda_{r+1}, \lambda_{r+2}, \ldots$, it follows by the Cauchy integral theorem (see \citet{Conw78}) that when $l$ and $j$ varies over $r+1,r+2,\ldots$, we have 
$$\int_{\Gamma} (\lambda_{l} - z)^{-1}(\lambda_{j} - z)^{-1} dz = 0.$$
Furthermore, for any $l = 1,2,\ldots,r$, we have 
\begin{eqnarray*}
&& \frac{1}{2{\pi}i} \int_{\Gamma} \frac{dz}{(\lambda_{j} - z)(\lambda_{l} - z)} \\
&=& \frac{1}{2{\pi}i} \int_{\Gamma} \left(\frac{1}{\lambda_{j} - z} - \frac{1}{\lambda_{l} - z}\right) dz \times \frac{1}{\lambda_{l} - \lambda_{j}} \\
&=& -\frac{1}{\lambda_{l} - \lambda_{j}} \times \frac{1}{2{\pi}i} \int_{\Gamma} \frac{dz}{\lambda_{l} - z} =  -\frac{1}{\lambda_{l} - \lambda_{j}}.
\end{eqnarray*}
The third inequality follows from Cauchy integral theorem along with the fact that $\Gamma$ does not contain $\lambda_{r+1}, \lambda_{r+2}, \ldots$ and the fact that $j$ varies over $r+1,r+2,\ldots$. The last equality follows from Cauchy formula (see \citet{Conw78}), stating that the integral is the winding number of $\Gamma$ around $\lambda_{l}$, which equals one. 

Combining all of the above facts, we finally deduce
\begin{eqnarray*}
\mathbb{E}\{\|U_{42}\|^{2}\mathbf{1}(G_{n})\} &\leq& \sum_{j=r+1}^{\infty} (\lambda_{j} + \rho)^{-2} \sum_{l=1}^{r} \frac{\langle\beta,\phi_{l}\rangle^{2}\rangle\lambda_{l}^{3}\lambda_{j}}{n(\lambda_{l} - \lambda_{j})^{2}} \\
&\leq& n^{-1} \sum_{j=r+1}^{\infty} \frac{\lambda_{j}}{(\lambda_{j} + \rho)^{2}} \left[ \sum_{l=1}^{r} \frac{\langle\beta,\phi_{l}\rangle^{2}\rangle\lambda_{l}^{3}}{n(\lambda_{l} - \lambda_{r+1})^{2}} \right] \\
&=& O(n^{-1}\rho^{-1-1/\alpha}) 
\end{eqnarray*}
as $n \rightarrow \infty$ by using \eqref{eq1}. Thus, we have 
\begin{eqnarray}
\mathbb{E}(\|U_{4}\|^{2}) = O(n^{-1}\rho^{-1-1/\alpha})  \label{estm-HR-13}
\end{eqnarray}
as $n \rightarrow \infty$.

Finally, we provide a bound for $\mathbb{E}(\|U_{2}\|^{2})$. Note that $\mathbb{E}(\|U_{2}\|^{2}) \leq 2\mathbb{E}\{\|(\hat{\mathscr{K}}_{\rho}^{-1} - \mathscr{K}_{\rho}^{-1})(\mathscr{P}_{1} - \hat{\mathscr{P}}_{1})\hat{C}\|^{2}\} + 2\mathbb{E}\{\|\hat{\mathscr{F}}_{\rho}(\hat{\mathscr{P}}_{1} - \mathscr{P}_{1})\hat{C}\|^{2}\}$.  

Similar arguments as above show that $\mathbb{E}\{\|\hat{\mathscr{F}}_{\rho}(\hat{\mathscr{P}}_{1} - \mathscr{P}_{1})\hat{C}\|^{2}\}$ is $O(n^{-2})$ as $n \rightarrow \infty$.  Further, $\mathbb{E}\{\|(\hat{\mathscr{K}}_{\rho}^{-1} - \mathscr{K}_{\rho}^{-1})(\mathscr{P}_{1} - \hat{\mathscr{P}}_{1})\hat{C}\|^{2}\} \leq 2\mathbb{E}\{\|(\hat{\mathscr{K}}_{\rho}^{-1} - \mathscr{K}_{\rho}^{-1})(\mathscr{P}_{1} - \hat{\mathscr{P}}_{1})(\hat{C} - C)\|^{2}\} + \mathbb{E}\{\|(\hat{\mathscr{K}}_{\rho}^{-1} - \mathscr{K}_{\rho}^{-1})(\mathscr{P}_{1} - \hat{\mathscr{P}}_{1})C\|^{2}\}$. Now,
\begin{eqnarray*}
&& \mathbb{E}\{\|(\hat{\mathscr{K}}_{\rho}^{-1} - \mathscr{K}_{\rho}^{-1})(\mathscr{P}_{1} - \hat{\mathscr{P}}_{1})(\hat{C} - C)\|^{2}\} \\
&\leq& E^{1/2}\{\hs\hat{\mathscr{K}}_{\rho}^{-1} - \mathscr{K}_{\rho}^{-1}\hs_{\infty}^{4}\}E^{1/4}\{\hs\hat{\mathscr{P}}_{1} - \mathscr{P}_{1}\hs_{\infty}^{8}\}E^{1/4}\{\|\hat{C} - C\|^{8}\} \\
&\leq& O(n^{-3}\rho^{-4}) = o(n^{-1}\rho^{-1-1/\alpha})
\end{eqnarray*}
as $n \rightarrow \infty$ by using \eqref{eq2}, the fact that $E\{\|\hat{C} - C\|^{8}\} = O(n^{-4})$ as $n \rightarrow \infty$ and arguments similar to those used earlier.

Next, 
\begin{eqnarray}
&& \mathbb{E}\{\|(\hat{\mathscr{K}}_{\rho}^{-1} - \mathscr{K}_{\rho}^{-1})(\mathscr{P}_{1} - \hat{\mathscr{P}}_{1})C\|^{2}\}  \label{eqn-modf3} \\
&\leq& E\{\|\hat{\mathscr{K}}_{\rho}^{-1}(\hat{\mathscr{K}} - \mathscr{K})\mathscr{K}_{\rho}^{-1}(\mathscr{P}_{1} - \hat{\mathscr{P}}_{1})C\|^{2}\}    \nonumber \\
&\leq& \rho^{-2}E^{1/2}\{\hs\mathscr{K} - \mathscr{K}\hs^{4}\}E^{1/2}\{\|\mathscr{K}_{\rho}^{-1}(\mathscr{P}_{1} - \hat{\mathscr{P}}_{1})C\|^{4}\}   \nonumber \\
&\leq& O(n^{-1}\rho^{-2})E^{1/2}\{\|\mathscr{K}_{\rho}^{-1}(\mathscr{P}_{1} - \hat{\mathscr{P}}_{1})C\|^{4}\}.   \nonumber 
\end{eqnarray}
Now, 
\begin{eqnarray*}
&& E^{1/2}\{\|\mathscr{K}_{\rho}^{-1}(\mathscr{P}_{1} - \hat{\mathscr{P}}_{1})C\|^{4}\}  \\
&\leq& O(1)[E^{1/2}\{\|\mathscr{K}_{\rho,1}^{-}(\mathscr{P}_{1} - \hat{\mathscr{P}}_{1})C\|^{4}\} + E^{1/2}\{\|\mathscr{K}_{\rho,2}^{-}(\mathscr{P}_{1} - \hat{\mathscr{P}}_{1})C\|^{4}\}].
\end{eqnarray*}
The first term on the right hand side of the above inequality is $O(n^{-1})$ as $n \rightarrow \infty$ and we need to bound the term $E\{\|\mathscr{K}_{\rho,2}^{-}(\mathscr{P}_{1} - \hat{\mathscr{P}}_{1})C\|^{4}\}$. To do this, we will follows the same arguments as those used to bound $E(||U_{42}||^{2})$ earlier.

Proceeding as in the case of bounding $E(||U_{42}||^{2})$, it is easy to see that to obtain a bound for $E\{\|\mathscr{K}_{\rho,2}^{-}(\mathscr{P}_{1} - \hat{\mathscr{P}}_{1})C\|^{4}\}$, it is enough to obtain a bound for $E\{(2{\pi}i)^{-1}\int_{\Gamma}\mathscr{K}_{\rho,2}^{-}(\mathscr{K}-z\mathscr{I})^{-1}(\hat{\mathscr{K}} - \mathscr{K})(\mathscr{K}-z\mathscr{I})^{-1}Cdz||^{4}\}$, where $\Gamma$ is the same contour as considered in the case of $E(||U_{42}||^{2})$. Now, expanding the latter term, we get that
\begin{eqnarray*}
&& E\left\{\|(2{\pi}i)^{-1}\int_{\Gamma}\mathscr{K}_{\rho,2}^{-}(\mathscr{K}-z\mathscr{I})^{-1}(\hat{\mathscr{K}} - \mathscr{K})(\mathscr{K}-z\mathscr{I})^{-1}Cdz||^{4}\right\} \\
&=& \frac{1}{(2{\pi}i)^{4}} \sum_{j_{1},j_{2}=r+1}^{\infty} \sum_{l_{1},l_{2},l_{3},l_{4}=1}^{\infty} \left\{ (\lambda_{j_{1}} + \rho)^{-2}(\lambda_{j_{2}} + \rho)^{-2} \times \right. \\ 
&& \ \ \ \ \left. \int_{\Gamma}\int_{\Gamma}\int_{\Gamma}\int_{\Gamma} \prod_{u=1}^{2} \frac{\langle\beta,\phi_{l_{u}}\rangle\lambda_{l_{u}}}{(\lambda_{l_{u}} - z_{u})(\lambda_{j_{1}} - z_{u})} \prod_{u=3}^{4} \frac{\langle\beta,\phi_{l_{u}}\rangle\lambda_{l_{u}}}{(\lambda_{l_{u}} - z_{u})(\lambda_{j_{2}} - z_{u})} S \prod_{u=1}^{4}dz_{l_{u}} \right\},
\end{eqnarray*}
where 
$$ S = E\left\{\prod_{u=1}^{2} \langle\hat{(\mathscr{K}} - \mathscr{K})\phi_{l_{u}},\phi_{j_{1}}\rangle \prod_{u=3}^{4} \langle\hat{(\mathscr{K}} - \mathscr{K})\phi_{l_{u}},\phi_{j_{2}}\rangle\right\}.$$
We obtained the expression of $S$ after \eqref{estm} while bounding $E(||U_{8}||^{2})$ earlier. Plugging-in those expressions and using the Cauchy integral theorem arguments used while bounding $E(||U_{42}||^{2})$ earlier, we get that
\begin{eqnarray*}
&& E\left\{\|(2{\pi}i)^{-1}\int_{\Gamma}\mathscr{K}_{\rho,2}^{-}(\mathscr{K}-z\mathscr{I})^{-1}(\hat{\mathscr{K}} - \mathscr{K})(\mathscr{K}-z\mathscr{I})^{-1}Cdz||^{4}\right\} \\
&\leq& O(n^{-2}) \sum_{j_{1},j_{2}=r+1}^{\infty} \sum_{l_{1},l_{2}=1}^{r} \frac{\langle\beta,\phi_{l_{1}}\rangle^{2}\langle\beta,\phi_{l_{2}}\rangle^{2}\lambda_{l_{1}}^{3}\lambda_{l_{2}}^{3}\{\lambda_{j_{1}}^{2}I(j_{1}=j_{2}) + \lambda_{j_{1}}\lambda_{j_{2}}I(j_{1} \neq j_{2})\}}{(\lambda_{j_{1}}+\rho)^{2}(\lambda_{j_{2}}+\rho)^{2}\prod_{u=1}^{2}(\lambda_{l_{u}}-\lambda_{j_{1}})\prod_{u=3}^{4}(\lambda_{l_{u}}-\lambda_{j_{2}})} \\
&\leq& O(n^{-2}) \left[\sum_{j=r+1}^{\infty} \frac{\lambda_{j}}{(\lambda_{j} + \rho)^{2}}\right]^{2} \ = \ O(n^{-2}\rho^{-2-2/\alpha})
\end{eqnarray*}
as $n \rightarrow \infty$ by Lemma \ref{lemma-2}.
Thus, it follows from \eqref{eqn-modf3} that $\mathbb{E}\{\|(\hat{\mathscr{K}}_{\rho}^{-1} - \mathscr{K}_{\rho}^{-1})(\mathscr{P}_{1} - \hat{\mathscr{P}}_{1})C\|^{2}\} = o(n^{-1}\rho^{-1-1/\alpha})$ and hence
\begin{eqnarray}
\mathbb{E}(\|U_{2}\|^{2}) = o(n^{-1}\rho^{-1-1/\alpha}) \label{estm-HR-8}
\end{eqnarray}
as $n \rightarrow \infty$.

The bound for $|\mathrm{MSE}(\hat{\beta}_{HR}) - \mathrm{MSE}(\widetilde{\beta}_{HR})|$ given in the statement of Theorem \ref{thm-data1} now follows from \eqref{estm-HR-3} and using the bounds \eqref{estm-HR-6}, \eqref{estm-HR-7}, \eqref{estm-HR-9}, \eqref{estm-HR-10},  \eqref{estm-HR-11}, \eqref{estm-HR-10-1}, \eqref{estm-HR-12}, \eqref{estm-HR-13} and \eqref{estm-HR-8}.

The bound for $\mathrm{MSE}(\hat{\beta}_{HR})$ is obtained by combining the above bound for $|\mathrm{MSE}(\hat{\beta}_{HR}) - \mathrm{MSE}(\widetilde{\beta}_{HR})|$ with the bound obtained for $\mathrm{MSE}(\widetilde{\beta}_{HR})$ in \eqref{estm-HR-6}.

The proofs of the results for $\hat{\beta}_{TR}$ are directly analogous to those for $\hat{\beta}_{HR}$ and are therefore omitted.
\end{proof}

\begin{proof}[Proof of Theorem \ref{thm-data2}]
It was obtained in the proof of part(b) of Theorem \ref{thm-orcl1} that $\mathrm{MSE}(\widetilde{\beta}_{TR}) - \mathrm{MSE}(\widetilde{\beta}_{HR}) = O(1)\rho^{2}$ as $n \rightarrow \infty$, where the $O(1)$ term is bounded below by a positive number for all sufficiently large $n$. Let $\kappa_{1}$ be a positive number which is less than this $O(1)$ term for all sufficiently large $n$. Then, 
\begin{eqnarray}
\mathrm{MSE}(\widetilde{\beta}_{TR}) - \mathrm{MSE}(\widetilde{\beta}_{HR}) > \kappa_{1}\rho^{2}  \label{mse-compare-1}
\end{eqnarray}
as $n \rightarrow \infty$. Note that this bound is irrespective of whether $\alpha > \eta - 1/2$ or $\alpha < \eta - 1/2$.

Now, if $\rho \sim cn^{-\varepsilon}$ for any $\varepsilon < \alpha/(5\alpha-2\eta+2)$ when $\alpha > \eta - 1/2$ or for any $\varepsilon < \alpha/(3\alpha+1)$ when $\alpha < \eta - 1/2$, it can be checked from \eqref{HR-data2} that $|\mathrm{MSE}(\hat{\beta}_{HR}) - \mathrm{MSE}(\widetilde{\beta}_{HR})| = o(\rho^{2})$ as $n \rightarrow \infty$. So, by Theorem \ref{thm-data1}, it follows that $|\mathrm{MSE}(\hat{\beta}_{TR}) - \mathrm{MSE}(\widetilde{\beta}_{TR})| = o(\rho^{2})$ as $n \rightarrow \infty$. Fix $\kappa_{0}$ to be any positive number less than $\kappa_{1}$. Thus, using the inequality
\begin{eqnarray*}
&& \mathrm{MSE}(\hat{\beta}_{TR}) - \mathrm{MSE}(\hat{\beta}_{HR}) 
> \{\mathrm{MSE}(\widetilde{\beta}_{TR}) - \mathrm{MSE}(\widetilde{\beta}_{HR})\} - |\mathrm{MSE}(\hat{\beta}_{TR}) - \mathrm{MSE}(\widetilde{\beta}_{TR})| \\
&& \qquad\qquad\qquad\qquad \qquad\qquad\qquad\qquad \qquad\qquad\qquad\qquad-|\mathrm{MSE}(\hat{\beta}_{HR}) - \mathrm{MSE}(\widetilde{\beta}_{HR})|
\end{eqnarray*}
along with \eqref{mse-compare-1} and the rates of convergences of $|\mathrm{M|SE}(\hat{\beta}_{TR}) - \mathrm{M|SE}(\widetilde{\beta}_{TR})|$ and $|\mathrm{MSE}(\hat{\beta}_{HR}) - \mathrm{MSE}(\widetilde{\beta}_{HR})|$ obtained above, it follows that 
$$\mathrm{MSE}(\hat{\beta}_{TR}) - \mathrm{MSE}(\hat{\beta}_{HR}) > \kappa_{0}n^{-2\varepsilon}$$
for all sufficiently large $n$ and for the above choices of $\varepsilon$.
\end{proof}

\begin{proof}[Proof of Theorem \ref{discrete-thm-data1}]
Note that $m^{-1}MSE(\hat{\beta}_{HR}^{(m)})$ is equal to the MSE of $\hat{\beta}_{HR}^{(m)}$ as an estimator of $\Phi_{m}(\beta)$ when we compute it based on $$\Phi_{m}(X_{1}), \Phi_{m}(X_{2}), \ldots, \Phi_{m}(X_{n}),$$ where $\Phi_{m}(X) = X^{(m)}/\sqrt{m}$ and $\Phi_{m}(\beta) = \beta^{(m)}/\sqrt{m}$. Since Theorem \ref{thm-data1} applies to any separable Hilbert space, we will follow the proof of this theorem for the above-mentioned random variables and parameter.

First observe that when deriving bounds for $E(||U_{l}||^{2})$ in the proof of Theorem \ref{thm-data1}, we required bounds for $\sum_{j \geq 1} \lambda_{j}^{a}/(\lambda_{j} + \rho)^{b}$ for $b \geq a > 0$. So, in the discrete case, we need bounds for $\sum_{j = 1}^{m} (\lambda_{j}^{(m)})^{a}/(\lambda_{j}^{(m)} + \rho)^{b}$ for $b \geq a > 0$. But by assumption (A2') and using the arguments in the proof of Lemma \ref{lemma-2}, it follows that
\begin{eqnarray*}
\sum_{j = 1}^{m} \frac{(\lambda_{j}^{(m)})^{a}}{(\lambda_{j}^{(m)} + \rho)^{b}} \leq O(1)\rho^{-b + a - \frac{1}{\alpha}}
\end{eqnarray*}
as $m \rightarrow \infty$, where the $O(1)$ term is uniform over $m$.

Next, we bound $\sum_{j=1}^{m} \langle\beta^{(m)},\phi_{j}^{(m)}\rangle^{2}/(\lambda_{j}^{(m)} + \rho)^{2}$, which is the discrete version of \eqref{estm-HR-55} used in the proof of Theorem \ref{thm-data1}. First, consider the case when $\alpha > \eta' - 1/2$. Observe that since $m > \rho^{-2}$ and $\alpha > 1$, in this case, we have $\rho^{-1/\alpha} < m^{1/\eta'}$. So, defining $J = [\rho^{-1/\alpha}]$ as in the proof of Lemma \ref{lemma-2} and by using assumption (A3'), we have 
\begin{eqnarray*}
\sum_{j = 1}^{J} \frac{\langle\beta^{(m)},\phi_{j}^{(m)}\rangle^{2}}{(\lambda_{j}^{(m)} + \rho)^{2}} &\leq& O(1)\rho^{\frac{2\eta'}{\alpha} - 2 - \frac{1}{\alpha}},
\end{eqnarray*}
where the $O(1)$ term is uniform over $m$. Further,
\begin{eqnarray*}
\sum_{j = J+1}^{[m^{1/\eta'}]} \frac{\langle\beta^{(m)},\phi_{j}^{(m)}\rangle^{2}}{(\lambda_{j}^{(m)} + \rho)^{2}} &\leq& O(1)\rho^{-2} \sum_{j > J} j^{-2\eta'} \ \leq \ O(1)\rho^{\frac{2\eta'}{\alpha} - 2 - \frac{1}{\alpha}}, \\
\sum_{j > [m^{1/\eta'}]}^{m} \frac{\langle\beta^{(m)},\phi_{j}^{(m)}\rangle^{2}}{(\lambda_{j}^{(m)} + \rho)^{2}} &\leq& O(1)\rho^{-2} \sum_{j > [m^{1/\eta'}]} m^{-2} \ \leq \ O(1)\rho^{-2}/m,
\end{eqnarray*}
where all the $O(1)$ terms above are uniform in $m$. Combining all the above inequalities, and using the facts that $m > \rho^{-2}$ and $(2\eta' - 1)/\alpha < 2$, we have
\begin{eqnarray*}
\sum_{j = 1}^{m} \frac{\langle\beta^{(m)},\phi_{j}^{(m)}\rangle^{2}}{(\lambda_{j}^{(m)} + \rho)^{2}} &\leq& O(1)\rho^{\frac{2\eta' - 1}{\alpha} - 2},
\end{eqnarray*}
where the $O(1)$ term is uniform over $m$. We then consider the case $\alpha < \eta' - 1/2$. In this case, we may either have $m > \rho^{-\eta'/\alpha}$ or $m \leq \rho^{-\eta'/\alpha}$. In the first scenario, as in the proof of Lemma \ref{lemma-2}, we have
\begin{eqnarray*}
\sup_{\rho > 0} \sum_{j = 1}^{[m^{1/\eta'}]} \frac{\langle\beta^{(m)},\phi_{j}^{(m)}\rangle^{2}}{(\lambda_{j}^{(m)} + \rho)^{2}} &\leq& O(1) \sum_{j = 1}^{[m^{1/\eta'}]} j^{-2\eta'+2\alpha} \ \leq \ O(1), \\
\sum_{j > [m^{1/\eta'}]}^{m} \frac{\langle\beta^{(m)},\phi_{j}^{(m)}\rangle^{2}}{(\lambda_{j}^{(m)} + \rho)^{2}} &\leq& O(1)\rho^{-2} \sum_{j > [m^{1/\eta'}]} m^{-2} \ \leq \ O(1)\rho^{-2}/m,
\end{eqnarray*}
where all the $O(1)$ terms are uniform in $m$. Since $m > \rho^{-2}$, we have
\begin{eqnarray*}
\sum_{j = 1}^{m} \frac{\langle\beta^{(m)},\phi_{j}^{(m)}\rangle^{2}}{(\lambda_{j}^{(m)} + \rho)^{2}} &\leq& O(1),
\end{eqnarray*}
with the $O(1)$ term being uniform in $m$. In the other scenario, when $m \leq \rho^{-\eta'/\alpha}$, we have
\begin{eqnarray*}
\sup_{\rho > 0} \sum_{j = 1}^{[m^{1/\eta'}]} \frac{\langle\beta^{(m)},\phi_{j}^{(m)}\rangle^{2}}{(\lambda_{j}^{(m)} + \rho)^{2}} &\leq& O(1) \sum_{j = 1}^{[m^{1/\eta'}]} j^{-2\eta'+2\alpha} \ \leq \ O(1), \\
\sum_{j > [m^{1/\eta'}]}^{[\rho^{-1/\alpha}]} \frac{\langle\beta^{(m)},\phi_{j}^{(m)}\rangle^{2}}{(\lambda_{j}^{(m)} + \rho)^{2}} &\leq& O(1) \sum_{j \leq [\rho^{-1/\alpha}]} m^{-2} j^{2\alpha} \ \leq \ O(1)\rho^{-2-1/\alpha}/m^{2}, \\
\sum_{j > [\rho^{-1/\alpha}]}^{m} \frac{\langle\beta^{(m)},\phi_{j}^{(m)}\rangle^{2}}{(\lambda_{j}^{(m)} + \rho)^{2}} &\leq& O(1)\rho^{-2} \sum_{j > [\rho^{-1/\alpha}]}^{m} m^{-2} \ \leq \ O(1)\rho^{-2}/m,
\end{eqnarray*}
where all the $O(1)$ terms are uniform in $m$. Since $\alpha > 1$ and $m > \rho^{-2}$, we again have
\begin{eqnarray*}
\sum_{j = 1}^{m} \frac{\langle\beta^{(m)},\phi_{j}^{(m)}\rangle^{2}}{(\lambda_{j}^{(m)} + \rho)^{2}} &\leq& O(1),
\end{eqnarray*}
with the $O(1)$ term being uniform in $m$. Thus, analogous to the bound in \eqref{estm-HR-55} in the proof of Theorem \ref{thm-data1}, we have
\begin{eqnarray*}
\rho^{2} \sum_{j = 1}^{m} \frac{\langle\beta^{(m)},\phi_{j}^{(m)}\rangle^{2}}{(\lambda_{j}^{(m)} + \rho)^{2}} &\leq& O(\rho^{M}),
\end{eqnarray*}
where $M = (2\eta' - 1)/\alpha$ or $M = 2$ according as $\alpha > \eta' - 1/2$ or $\alpha < \eta' - 1/2$.

The proof of the present theorem is now complete by using arguments similar to those used in the proof of Theorem \ref{thm-data2}, using the above bounds and noting that the other $O(1)$ terms in the proof of Theorem \ref{thm-data1}, namely, those involved with the $E(||U_{l}||^{2})$'s and that of $m^{-1}E\{||\widetilde{\beta}_{HR}^{(m)} - \beta^{(m)}||^{2}\}$ are uniform over $m$.
\end{proof}

\section{Appendix}

This appendix contains the results of a simulation study when the functional covariate is observed with error and a comparison of the three regularisation methods considered on a real data set. It also contains the proof of the fact that standard Brownian motion verifies Assumption (A3'), as claimed earlier in the paper.

\subsection*{Performance on Simulated Data}

\indent We consider the same models as those in the later part of Section 6 in the paper, the only difference being that now the functional covariate $X$ is observed subject to measurement error. That is, instead of $X_{i}(t)$, we observe $W_{i}(t) = X_{i}(t) + \xi_{i,t}$, $i = 1,2,\ldots,n$, where the $\{\xi_{i,t} : t \in [0,1]\}$'s comprise i.i.d. standard Gaussian errors independent of the $X_{i}$'s. Note that the above equation is only formal, and $W_{i}$ is not a tight Gaussian random element in $L_{2}[0,1]$ since its covariance is not trace class. The equality is only understood in the weak (cylinder) sense. Thus, the theory in the main paper cannot accommodate such a setting. The following table gives the empirical MSE's of the spectral truncation, the Tikhonov, and the hybrid estimators under the models considered.

\indent It is observed from Table \ref{Tab-1} that the Tikhonov estimator has significantly larger cross-validated MSE than the other two estimators in both the well-spaced and the closely-spaced settings for $\beta_{1}$ and $\beta_{2}$. The hybrid estimator and the spectral truncation estimator have quite similar performance for the above settings. While the cross-validation estimates of their MSEs are close to the true MSEs in the well-spaced case, these estimates are significantly different from the true MSEs in the closely-spaced case. Hence, the tentative patterns in the error-in-covariate models are different those found in Section 6 of the main paper. One possible reason for this is that the error in the covariate automatically induces a regularization, so the spectral truncation estimator in this case behaves almost like a partial (considering only the first $m$ eigenfunctions) Tikhonov estimator with the $\rho$-parameter equal to the variance of the error component. On the other hand, the usual Tikhonov estimator adds an additional regularisation (the $\rho$ in its definition) over and above the variance of the error component, which induces additional bias and may explain the deterioration in its performance.

\begin{table}[t]
\caption{MSEs of the hybrid regularisation, the Tikhonov regularisation and the spectral truncation estimators for simulated models}{
\begin{center}
\scalebox{1}{
\begin{tabular}{cccccccc}
 \hline \\
\multicolumn{8}{c}{well-spaced} \\ \\ \hline \\ 
$\beta$ & $\alpha$ & $MSE_{ST}^{GCV}$ & $MSE_{ST}^{opt}$ & $MSE_{TR}^{GCV}$ & $MSE_{TR}^{opt}$ & $MSE_{HR}^{GCV}$ & $MSE_{HR}^{opt}$ \\ \\ \hline 
$\beta_{1}$ & $1.1$ & $0.361$ & $0.324$ & $0.845$ & $0.622$ & $0.323$ & $0.31$ \\ 
 & $2$ & $0.371$ & $0.361$ & $0.904$ & $0.742$ & $0.362$ & $0.356$ \\ \\
$\beta_{2}$ & $1.1$ & $0.324$ & $0.296$ & $0.829$ & $0.595$ & $0.295$ & $0.288$ \\
 & $2$ & $0.328$ & $0.312$ & $0.891$ & $0.692$ & $0.337$ & $0.31$ \\ \\
$\beta_{3}$ & $1.1$ & $0.053$ & $0.049$ & $0.065$ & $0.063$ & $0.072$ & $0.07$ \\
 & $2$ & $0.051$ & $0.049$ & $0.057$ & $0.054$ & $0.064$ & $0.061$ \\
 \hline \\ 
 \multicolumn{8}{c}{closely-spaced} \\ \\ \hline \\ 
 $\beta$ & $\alpha$ & $MSE_{ST}^{GCV}$ & $MSE_{ST}^{opt}$ & $MSE_{TR}^{GCV}$ & $MSE_{TR}^{opt}$ & $MSE_{HR}^{GCV}$ & $MSE_{HR}^{opt}$ \\ \\ \hline 
$\beta_{1}$ & $1.1$ & $1.256$ & $1.073$ & $1.571$ & $1.309$ & $1.258$ & $1.064$ \\ 
 & $2$ & $1.233$ & $1.024$ & $1.54$ & $1.301$ & $1.237$ & $1.019$ \\ \\
$\beta_{2}$ & $1.1$ & $1.163$ & $1.017$ & $1.507$ & $1.236$ & $1.174$ & $1.009$ \\
 & $2$ & $1.157$ & $0.965$ & $1.472$ & $1.225$ & $1.159$ & $0.961$ \\ \\
$\beta_{3}$ & $1.1$ & $0.053$ & $0.05$ & $0.056$ & $0.053$ & $0.067$ & $0.061$ \\
 & $2$ & $0.052$ & $0.049$ & $0.057$ & $0.052$ & $0.061$ & $0.059$ \\
\hline
\end{tabular}}
\end{center}}
\label{Tab-1}
\end{table}

\subsection*{Performance on Real Data}

\indent We consider the well-known Canadian weather data set to test the relative performance of the hybrid, the Tikhonov and the spectral trunction estimator. The data set can be obtained from
\begin{center}
 \url{http://www.psych.mcgill.ca/misc/fda/downloads/FDAfuns/R/data/} 
\end{center}

It contains the monthly precipitations at $n = 35$ weather stations in Canada. The data set also contains the daily temperatures for those weather stations over a year. We fit a functional linear model with response variable being the logarithm of the annual precipitation and the functional predictor being the daily temperatures. We do not aim to elicit any novel instights in this very intensely studied data set. Rather, we wish to use it as a benchmark, as it is publicly available and indeed often used as such.  

To quantify the effectiveness of the estimators, we use the prediction error in terms of MSE. Admittedly, the problem of prediction problem is quite different from the problem of estimation problem in functional regression, and our theoretical results to not address the case of prediction. We nevertheless follow this approach, as it is essentially the only check we can practically carry out: since we do not know the true slope function $\beta$, we cannot estimate the estimation MSE. We estimate the prediction error of each estimator by randomly splitting the data set in two parts, one part for estimating the slope function (including choosing the tuning parameters) and the other part for prediction. This splitting protects against over-fitting, which would confound the true performance of the estimators. We then calculate the prediction error, and the estimate its mean is obtained by taking an average over $1000$ randomly generated splits. 

The prediction errors for the spectral truncation and the Tikhonov estimators are thus estimated to be $3.57$ and $3.27$, respectively, while that of the hybrid estimator is estimated as $2.87$, representing a considerable improvement over either of the two standard estimators.

\subsection*{Validity of Assumption (A3') in Standard Brownian Motion}

We will use the following facts about standard Brownian motion: 
\begin{itemize}
\item[(a)] $\lambda_{j} = \{(j-0.5)\pi\}^{-1}$ for all $j \geq 1$, \\
\item[(b)] $\phi_{j}(t) = \sqrt{2}\sin\{(j-0.5){\pi}t\}$ for $t \in [0,1]$ and for all $j \geq 1$, \\
\item[(c)] The covariance kernel $K(s,t) = \min(s,t)$ of standard Brownian motion is continuously differentiable almost everywhere on $[0,1]^{2}$.  \\
\end{itemize}
Note that it follows from (b) above that $\phi_{j}'(t) = \sqrt{2}(j-0.5){\pi}\cos\{(j-0.5){\pi}t\}$. Thus, we also have
\begin{itemize}
\item[(d)] $\sup_{t \in [0,1]} |\phi_{j}'(t)| \leq \sqrt{2}\lambda_{j}^{-1/2}$ for all $j \geq 1$.
\end{itemize}
Suppose that $\beta$ lies in the RKHS of standard Brownian motion and thus $$\sum_{j=1}^{\infty} \lambda_{j}^{-1}\langle\beta,\phi_{j}\rangle^{2} < \infty.$$  By the Cauchy-Schwarz inequality, we have
\begin{eqnarray} 
\sum_{j=1}^{\infty} |\langle\beta,\phi_{j}\rangle| < \infty.   \label{sbm-1}
\end{eqnarray}
 
First observe that the dispersion matrix $K^{(m)} = ((k_{ij}))$ of $\Phi_{m}(X)$ satisfies $k_{ij} = K(t_{i},t_{j})/m$ for $1\leq i,j \leq m$ with $K$ denoting the covariance kernel of $X$. Consider the kernel
\begin{eqnarray*}
K_{app}(u,v) = \sum_{i,j=1}^{m} K(t_{i},t_{j}) \mathbb{I}\{(u,v) \in I_{i,m} \times I_{j,m}\}
\end{eqnarray*}
for $u ,v \in [0,1]$, where $I_{i,m} = [(i-1)/m,i/m)$ for $1 \leq i \leq m$. Note that $u \in I_{[mu]+1,m}$ for every $u \in [0,1]$. If $(\lambda,\psi)$ is an eigenvalue/eigenfunction pair of the operator associated with the kernel $K_{app}$, then
\begin{eqnarray*}
&& \int_{0}^{1} K_{app}(u,v)\psi(v)dv = \lambda\psi(u) \ \ \ \forall u \\
&& \Rightarrow \int_{0}^{1} K(t_{[mu]+1},t_{[mv]+1})\psi(v)dv = \lambda\psi(u) \ \ \ \forall u \\
&& \Rightarrow \sum_{i=1}^{m} \int_{(i-1)/m}^{i/m} K(t_{[mu]+1,m},t_{i})\psi(v)dv = \lambda\psi(u) \ \ \ \forall u \\
&& \Rightarrow \sum_{i=1}^{m} K(t_{[mu]+1,m},t_{i}) \{\Psi(i/m) - \Psi((i-1)/m)\} = \lambda\psi(u) \ \ \ \forall u,
\end{eqnarray*}
where $\Psi(s) = \int_{0}^{s} \psi(t)dt$. Upon integrating the last equation over $u \in I_{j,m}$ for every $1 \leq j \leq m$, we obtain
\begin{eqnarray*}
&& \sum_{i=1}^{m} \left[\int_{(j-1)/m}^{j/m} K(t_{[mu]+1,m},t_{i})du \right] \{\Psi(i/m) - \Psi((i-1)/m)\} = \lambda \int_{(j-1)/m}^{j/m} \psi(u) \ \ \ \forall j, \\
&& \Rightarrow \sum_{i=1}^{m} m^{-1}K(t_{j},t_{i}) \{\Psi(i/m) - \Psi((i-1)/m)\} = \lambda \{\Psi(j/m) - \Psi((j-1)/m)\} \ \ \ \forall j. \\
\end{eqnarray*} 
Defining $d_{j} = \{\Psi(j/m) - \Psi((j-1)/m)\}$ for $1 \leq j \leq m$ and $d = (d_{1},d_{2},\ldots,d_{m})$, the above equation reduces to $K^{(m)}d = \lambda d$. So, $\lambda$ is also an eigenvalue of $K^{(m)}$ and the corresponding eigenvector is $\phi_{j}^{(m)} = d/||d||$. In other words, $K^{(m)}$ and $K_{app}$ share the same non-zero eigenvalues, which are at most $m$ in number. 

Let $\widetilde{\phi}_{j}^{(m)}$ denote the eigenfunction of $\mathscr{K}_{app}$ associated with the eigenvalue $\lambda_{j}^{(m)}$. It can be shown that the $\widetilde{\phi}_{j}^{(m)}$'s are uniformly bounded for all $j$ and $m$, and it also follows that
\begin{eqnarray}
\phi_{j}^{(m)} = O(1)\,m^{1/2}\left(\int_{(l-1)/m}^{l/m} \widetilde{\phi}_{j}^{(m)}(v) dv : l = 1,2,\ldots,m \right),  \label{sbm-2}
\end{eqnarray}
where the $O(1)$ term is positive and is uniform over $j$ and $m$.

Next, by a simple Taylor series expansion and using fact (c) stated earlier, it follows that 
\begin{eqnarray}
\hs\mathscr{K}_{app} - \mathscr{K}\hs^{2} = O(m^{-2})   \label{sbm-3}
\end{eqnarray}
as $m \rightarrow \infty$. Now for any $k \geq 1$ and any $j = 1,2,\ldots,m$,
\begin{eqnarray}
\langle\phi_{k},\widetilde{\phi}_{j}^{(m)}\rangle &=& \int_{0}^{1} \phi_{k}(v)\widetilde{\phi}_{j}^{(m)}(v)dv \nonumber \\
&=& \sum_{l=1}^{m} \int_{(l-1)/m}^{l/m} \phi_{k}(v)\widetilde{\phi}_{j}^{(m)}(v)dv \nonumber \\
&=& \sum_{l=1}^{m} \int_{(l-1)/m}^{l/m} \left[\phi_{k}(t_{l}) + (v - t_{l})\phi_{k}'(v_{1})\right]\widetilde{\phi}_{j}^{(m)}(v)dv \nonumber \\
&=& \sum_{l=1}^{m} \phi_{k}(t_{l}) \int_{(l-1)/m}^{l/m} \widetilde{\phi}_{j}^{(m)}(v)dv + \sum_{l=1}^{m} \int_{(l-1)/m}^{l/m} (v - t_{l})\phi_{k}'(v_{1})\widetilde{\phi}_{j}^{(m)}(v)dv \nonumber \\
&=& O(1)m^{-1/2} \sum_{l=1}^{m} \phi_{k}(t_{l})\phi_{j,l}^{(m)} + \sum_{l=1}^{m} \int_{(l-1)/m}^{l/m} (v - t_{l})\phi_{k}'(v_{1})\widetilde{\phi}_{j}^{(m)}(v)dv,  \label{sbm-4}
\end{eqnarray}
where we denote $\phi_{j}^{(m)} = (\phi_{j,l}^{(m)} : l=1,2,\ldots,m)'$ and the $O(1)$ term is non-zero and uniform over $j$ and $m$. Now note that 
\begin{eqnarray*}
\lambda_{k}\langle\phi_{k},\widetilde{\phi}_{j}^{(m)}\rangle &=& \langle\mathscr{K}\phi_{k},\widetilde{\phi}_{j}^{(m)}\rangle \ = \ \langle\phi_{k},\mathscr{K}\widetilde{\phi}_{j}^{(m)}\rangle \nonumber 
= \langle\phi_{k},(\mathscr{K} - \mathscr{K}_{app})\widetilde{\phi}_{j}^{(m)}\rangle + \lambda_{j}^{(m)}\langle\phi_{k},\widetilde{\phi}_{j}^{(m)}\rangle. \nonumber 
\end{eqnarray*}
Thus, using \eqref{sbm-3} and the fact that assumption (A2') holds for standard Brownian motion, we have
\begin{eqnarray}
\lambda_{k}|\langle\phi_{k},\widetilde{\phi}_{j}^{(m)}\rangle| \leq \hs\mathscr{K} - \mathscr{K}_{app}\hs + \lambda_{j}^{(m)} = O(1)\{m^{-1} + j^{-\alpha}\}, \label{sbm-5}
\end{eqnarray}
where the $O(1)$ term is uniform over $j$ and all large $m$.

Since $\beta$ lies in the RKHS of standard Brownian motion, it can be shown that
\begin{eqnarray*}
m^{-1/2}\langle\beta^{(m)},\phi_{j}^{(m)}\rangle &=& m^{-1/2}\sum_{k=1}^{\infty} \langle\beta,\phi_{k}\rangle\lambda_{k}\sum_{l=1}^{m} \phi_{k}(t_{l})\phi_{j,l}^{(m)} \\
\Rightarrow |m^{-1/2}\langle\beta^{(m)},\phi_{j}^{(m)}\rangle - \langle\beta,\phi_{j}\rangle| &\leq& \left|\langle\beta,\phi_{j}\rangle\left\{1 - \lambda_{j}m^{-1/2}\sum_{l=1}^{m} \phi_{j}(t_{l})\phi_{j,l}^{(m)}\right\}\right| \\
&& + \ \left|\sum_{k \neq j} \langle\beta,\phi_{k}\rangle \lambda_{k} m^{-1/2}\sum_{l=1}^{m} \phi_{k}(t_{l})\phi_{j,l}^{(m)}\right| \\
&\leq& |\langle\beta,\phi_{j}\rangle|\left(1 + \lambda_{j}\left\{m^{-1}\sum_{l=1}^{m} \phi_{j}^{2}(t_{l})\right\}^{1/2}\right) \\
&& + \ \sum_{k \neq j} |\langle\beta,\phi_{k}\rangle| \lambda_{k} m^{-1/2}\left|\sum_{l=1}^{m} \phi_{k}(t_{l})\phi_{j,l}^{(m)}\right|
\end{eqnarray*}
The first term on the right hand side of the above inequality is bounded above by $O(1)|\langle\beta,\phi_{j}\rangle|$, where the $O(1)$ term is uniform in $j$ and $m$. This is due to the fact that the $|\phi_{j}(t)| \leq \sqrt{2}$ for all $j$. For bounding the second term, first observe that from \eqref{sbm-4}, it follows that
\begin{eqnarray*}
&& \lambda_{k} m^{-1/2}\left|\sum_{l=1}^{m} \phi_{k}(t_{l})\phi_{j,l}^{(m)}\right| \\
&\leq& O(1)\left[ \lambda_{k}|\langle\phi_{k},\widetilde{\phi}_{j}^{(m)}| + \lambda_{k} \sum_{l=1}^{m} \left|\int_{(l-1)/m}^{l/m} (v - t_{l})\phi_{k}'(v_{1})\widetilde{\phi}_{j}^{(m)}(v)dv \right| \right] \\
&\leq& O(1)\left[m^{-1} + j^{-\alpha} + \lambda_{k}\sum_{l=1}^{m} \left( \int_{(l-1)/m}^{l/m} (v-t_{l})^{2}[\phi_{k}'(v_{1})]^{2}dv\right)^{1/2} \left( \int_{(l-1)/m}^{l/m} [\widetilde{\phi}_{j}^{(m)}(v)]^{2}dv \right)^{1/2} \right] \\
&\leq& O(1)\left[m^{-1} + j^{-\alpha} + \lambda_{k}\sup_{v \in [0,1]}|\phi_{k}'(v)| \left( \sum_{l=1}^{m} \int_{(l-1)/m}^{l/m} (v-t_{l})^{2}dv \right)^{1/2} \left( \sum_{l=1}^{m} \int_{(l-1)/m}^{l/m} [\widetilde{\phi}_{j}^{(m)}(v)]^{2}dv \right)^{1/2} \right] \\
&\leq& O(1)\left[m^{-1} + j^{-\alpha} + \lambda_{k}^{1/2}m^{-1} \right],
\end{eqnarray*}
where the $O(1)$ term is uniform over $j$ and all large $m$. The second inequality above follows from \eqref{sbm-5} and fact (d) stated earlier. Thus, using \eqref{sbm-1}, we have
\begin{eqnarray*}
&& \sum_{k \neq j} |\langle\beta,\phi_{k}\rangle| \lambda_{k} m^{-1/2}\left|\sum_{l=1}^{m} \phi_{k}(t_{l})\phi_{j,l}^{(m)}\right| \\
&\leq& O(1)\left[(m^{-1} + j^{-\alpha}) \sum_{k=1}^{\infty} |\langle\beta,\phi_{k}\rangle| + m^{-1} \sum_{k=1}^{\infty} \lambda_{k}^{1/2}|\langle\beta,\phi_{k}\rangle| \right] \\
&=& O(1)\left[m^{-1} + j^{-\alpha}\right],
\end{eqnarray*}
with the $O(1)$ term being uniform over $j$ and all large $m$. Combining the inequalities obtained above, we get that
\begin{eqnarray*}
|m^{-1/2}\langle\beta^{(m)},\phi_{j}^{(m)}\rangle| \leq O(1)\left[j^{-\eta} + j^{-\alpha} + m^{-1}\right]
\end{eqnarray*}
if assumption (A3) holds, where the $O(1)$ term is uniform over $j$ and all large $m$. This completes the proof of assumption (A3') for standard Brownian motion. So, we have $\eta' = \alpha$ or $\eta' = \eta$ in assumption (A3') according as $\eta \geq \alpha$ or $\eta < \alpha < 2\eta -1$.

\bibliographystyle{imsart-nameyear}
\bibliography{biblio1}

\begin{thebibliography}{29}

\bibitem[\protect\citeauthoryear{Amini and Wainwright}{2012}]{AW12}
\begin{barticle}[author]
\bauthor{\bsnm{Amini},~\bfnm{Arash~A.}\binits{A.~A.}} \AND
  \bauthor{\bsnm{Wainwright},~\bfnm{Martin~J.}\binits{M.~J.}}
(\byear{2012}).
\btitle{Sampled forms of functional {PCA} in reproducing kernel {H}ilbert
  spaces}.
\bjournal{Ann. Statist.}
\bvolume{40}
\bpages{2483--2510}.
\bdoi{10.1214/12-AOS1033}
\bmrnumber{3097610}
\end{barticle}
\endbibitem

\bibitem[\protect\citeauthoryear{Bai and Saranadasa}{1996}]{BS96}
\begin{barticle}[author]
\bauthor{\bsnm{Bai},~\bfnm{Zhidong}\binits{Z.}} \AND
  \bauthor{\bsnm{Saranadasa},~\bfnm{Hewa}\binits{H.}}
(\byear{1996}).
\btitle{Effect of high dimension: by an example of a two sample problem}.
\bjournal{Statist. Sinica}
\bvolume{6}
\bpages{311--329}.
\bmrnumber{1399305 (97i:62062)}
\end{barticle}
\endbibitem

\bibitem[\protect\citeauthoryear{Cai and Yuan}{2011}]{CY11}
\begin{barticle}[author]
\bauthor{\bsnm{Cai},~\bfnm{T.~Tony}\binits{T.~T.}} \AND
  \bauthor{\bsnm{Yuan},~\bfnm{Ming}\binits{M.}}
(\byear{2011}).
\btitle{Optimal estimation of the mean function based on discretely sampled
  functional data: phase transition}.
\bjournal{Ann. Statist.}
\bvolume{39}
\bpages{2330--2355}.
\bdoi{10.1214/11-AOS898}
\bmrnumber{2906870}
\end{barticle}
\endbibitem

\bibitem[\protect\citeauthoryear{Cardot, Ferraty and Sarda}{2003}]{CFS03}
\begin{barticle}[author]
\bauthor{\bsnm{Cardot},~\bfnm{Herv{\'e}}\binits{H.}},
  \bauthor{\bsnm{Ferraty},~\bfnm{Fr{\'e}d{\'e}ric}\binits{F.}} \AND
  \bauthor{\bsnm{Sarda},~\bfnm{Pascal}\binits{P.}}
(\byear{2003}).
\btitle{Spline estimators for the functional linear model}.
\bjournal{Statist. Sinica}
\bvolume{13}
\bpages{571--591}.
\bmrnumber{1997162 (2004e:62072)}
\end{barticle}
\endbibitem

\bibitem[\protect\citeauthoryear{Cardot and Johannes}{2010}]{CJ10}
\begin{barticle}[author]
\bauthor{\bsnm{Cardot},~\bfnm{Herv{\'e}}\binits{H.}} \AND
  \bauthor{\bsnm{Johannes},~\bfnm{Jan}\binits{J.}}
(\byear{2010}).
\btitle{Thresholding projection estimators in functional linear models}.
\bjournal{J. Multivariate Anal.}
\bvolume{101}
\bpages{395--408}.
\bdoi{10.1016/j.jmva.2009.03.001}
\bmrnumber{2564349}
\end{barticle}
\endbibitem

\bibitem[\protect\citeauthoryear{Cardot, Mas and Sarda}{2007}]{CMS07}
\begin{barticle}[author]
\bauthor{\bsnm{Cardot},~\bfnm{Herv{\'e}}\binits{H.}},
  \bauthor{\bsnm{Mas},~\bfnm{Andr{\'e}}\binits{A.}} \AND
  \bauthor{\bsnm{Sarda},~\bfnm{Pascal}\binits{P.}}
(\byear{2007}).
\btitle{C{LT} in functional linear regression models}.
\bjournal{Probab. Theory Related Fields}
\bvolume{138}
\bpages{325--361}.
\bdoi{10.1007/s00440-006-0025-2}
\bmrnumber{2299711 (2007m:60055)}
\end{barticle}
\endbibitem

\bibitem[\protect\citeauthoryear{Cardot and Sarda}{2006}]{HP06}
\begin{bincollection}[author]
\bauthor{\bsnm{Cardot},~\bfnm{Herv\'{e}}\binits{H.}} \AND
  \bauthor{\bsnm{Sarda},~\bfnm{Pascal}\binits{P.}}
(\byear{2006}).
\btitle{Linear Regression Models for Functional Data}.
In \bbooktitle{The Art of Semiparametrics}.
\bseries{Contributions to Statistics}
\bpages{49-66}.
\bpublisher{Physica-Verlag HD}.
\bdoi{10.1007/3-7908-1701-5_4}
\end{bincollection}
\endbibitem

\bibitem[\protect\citeauthoryear{Chakraborty and Panaretos}{2016}]{CP16-Supp}
\begin{bunpublished}[author]
\bauthor{\bsnm{Chakraborty},~\bfnm{Anirvan}\binits{A.}} \AND
  \bauthor{\bsnm{Panaretos},~\bfnm{Victor~M.}\binits{V.~M.}}
(\byear{2016}).
\btitle{Supplement to ``Hybrid regularisation of functional linear models''}.
\end{bunpublished}
\endbibitem

\bibitem[\protect\citeauthoryear{Chen and Qin}{2010}]{CQ10}
\begin{barticle}[author]
\bauthor{\bsnm{Chen},~\bfnm{Song~Xi}\binits{S.~X.}} \AND
  \bauthor{\bsnm{Qin},~\bfnm{Ying-Li}\binits{Y.-L.}}
(\byear{2010}).
\btitle{A two-sample test for high-dimensional data with applications to
  gene-set testing}.
\bjournal{Ann. Statist.}
\bvolume{38}
\bpages{808--835}.
\bdoi{10.1214/09-AOS716}
\bmrnumber{2604697 (2011c:62187)}
\end{barticle}
\endbibitem

\bibitem[\protect\citeauthoryear{Comte and Johannes}{2012}]{CJ12}
\begin{barticle}[author]
\bauthor{\bsnm{Comte},~\bfnm{Fabienne}\binits{F.}} \AND
  \bauthor{\bsnm{Johannes},~\bfnm{Jan}\binits{J.}}
(\byear{2012}).
\btitle{Adaptive functional linear regression}.
\bjournal{Ann. Statist.}
\bvolume{40}
\bpages{2765--2797}.
\bdoi{10.1214/12-AOS1050}
\bmrnumber{3097959}
\end{barticle}
\endbibitem

\bibitem[\protect\citeauthoryear{Conway}{1978}]{Conw78}
\begin{bbook}[author]
\bauthor{\bsnm{Conway},~\bfnm{John~B.}\binits{J.~B.}}
(\byear{1978}).
\btitle{Functions of one complex variable},
\bedition{second} ed.
\bseries{Graduate Texts in Mathematics}
\bvolume{11}.
\bpublisher{Springer-Verlag, New York-Berlin}.
\bmrnumber{503901 (80c:30003)}
\end{bbook}
\endbibitem

\bibitem[\protect\citeauthoryear{Crambes, Kneip and Sarda}{2009}]{CKS09}
\begin{barticle}[author]
\bauthor{\bsnm{Crambes},~\bfnm{Christophe}\binits{C.}},
  \bauthor{\bsnm{Kneip},~\bfnm{Alois}\binits{A.}} \AND
  \bauthor{\bsnm{Sarda},~\bfnm{Pascal}\binits{P.}}
(\byear{2009}).
\btitle{Smoothing splines estimators for functional linear regression}.
\bjournal{Ann. Statist.}
\bvolume{37}
\bpages{35--72}.
\bdoi{10.1214/07-AOS563}
\bmrnumber{2488344 (2010i:62089)}
\end{barticle}
\endbibitem

\bibitem[\protect\citeauthoryear{Cuevas, Febrero and Fraiman}{2002}]{CFF02}
\begin{barticle}[author]
\bauthor{\bsnm{Cuevas},~\bfnm{Antonio}\binits{A.}},
  \bauthor{\bsnm{Febrero},~\bfnm{Manuel}\binits{M.}} \AND
  \bauthor{\bsnm{Fraiman},~\bfnm{Ricardo}\binits{R.}}
(\byear{2002}).
\btitle{Linear functional regression: The case of fixed design and functional
  response}.
\bjournal{Canadian Journal of Statistics}
\bvolume{30}
\bpages{285--300}.
\bdoi{10.2307/3315952}
\end{barticle}
\endbibitem

\bibitem[\protect\citeauthoryear{Ferraty and Vieu}{2000}]{FerratyReg}
\begin{barticle}[author]
\bauthor{\bsnm{Ferraty},~\bfnm{Fr\'{e}d\'{e}ric}\binits{F.}} \AND
  \bauthor{\bsnm{Vieu},~\bfnm{Philippe}\binits{P.}}
(\byear{2000}).
\btitle{Dimension fractale et estimation de la r\'{e}gression dans des espaces
  vectoriels semi-norm\'{e}s}.
\bjournal{Comptes Rendus de l'Acad\'{e}mie des Sciences - Series I -
  Mathematics}
\bvolume{330}
\bpages{139 - 142}.
\bdoi{http://dx.doi.org/10.1016/S0764-4442(00)00140-3}
\end{barticle}
\endbibitem

\bibitem[\protect\citeauthoryear{Grenander}{1981}]{grenander1981}
\begin{bbook}[author]
\bauthor{\bsnm{Grenander},~\bfnm{Ulf}\binits{U.}}
(\byear{1981}).
\btitle{Abstract inference}.
\bpublisher{Wiley New York}.
\end{bbook}
\endbibitem

\bibitem[\protect\citeauthoryear{Hall and Horowitz}{2007}]{HH07}
\begin{barticle}[author]
\bauthor{\bsnm{Hall},~\bfnm{Peter}\binits{P.}} \AND
  \bauthor{\bsnm{Horowitz},~\bfnm{Joel~L.}\binits{J.~L.}}
(\byear{2007}).
\btitle{Methodology and convergence rates for functional linear regression}.
\bjournal{Ann. Statist.}
\bvolume{35}
\bpages{70--91}.
\bdoi{10.1214/009053606000000957}
\bmrnumber{2332269 (2008k:62134)}
\end{barticle}
\endbibitem

\bibitem[\protect\citeauthoryear{Hall and Hosseini-Nasab}{2006}]{HHN06}
\begin{barticle}[author]
\bauthor{\bsnm{Hall},~\bfnm{Peter}\binits{P.}} \AND
  \bauthor{\bsnm{Hosseini-Nasab},~\bfnm{Mohammad}\binits{M.}}
(\byear{2006}).
\btitle{On properties of functional principal components analysis}.
\bjournal{J. R. Stat. Soc. Ser. B Stat. Methodol.}
\bvolume{68}
\bpages{109--126}.
\bdoi{10.1111/j.1467-9868.2005.00535.x}
\bmrnumber{2212577}
\end{barticle}
\endbibitem

\bibitem[\protect\citeauthoryear{Hocking}{2003}]{Hock03}
\begin{bbook}[author]
\bauthor{\bsnm{Hocking},~\bfnm{Ronald~R.}\binits{R.~R.}}
(\byear{2003}).
\btitle{Methods and applications of linear models},
\bedition{second} ed.
\bseries{Wiley Series in Probability and Statistics}.
\bpublisher{Wiley-Interscience [John Wiley \& Sons], Hoboken, NJ}.
\bdoi{10.1002/0471434159}
\bmrnumber{1963885 (2004b:62002)}
\end{bbook}
\endbibitem

\bibitem[\protect\citeauthoryear{Hoerl and Kennard}{1970}]{HK70}
\begin{barticle}[author]
\bauthor{\bsnm{Hoerl},~\bfnm{A.~E.}\binits{A.~E.}} \AND
  \bauthor{\bsnm{Kennard},~\bfnm{R.~W.}\binits{R.~W.}}
(\byear{1970}).
\btitle{Ridge regression: Biased estimation for nonorthogonal problems}.
\bjournal{Technometrics}
\bvolume{12}
\bpages{55--67}.
\bdoi{10.2307/1267351}
\end{barticle}
\endbibitem

\bibitem[\protect\citeauthoryear{Horv{\'a}th and Kokoszka}{2012}]{HK12}
\begin{bbook}[author]
\bauthor{\bsnm{Horv{\'a}th},~\bfnm{Lajos}\binits{L.}} \AND
  \bauthor{\bsnm{Kokoszka},~\bfnm{Piotr}\binits{P.}}
(\byear{2012}).
\btitle{Inference for functional data with applications}.
\bseries{Springer Series in Statistics}.
\bpublisher{Springer}, \baddress{New York}.
\bdoi{10.1007/978-1-4614-3655-3}
\bmrnumber{2920735}
\end{bbook}
\endbibitem

\bibitem[\protect\citeauthoryear{Hsing and Eubank}{2015}]{HE15}
\begin{bbook}[author]
\bauthor{\bsnm{Hsing},~\bfnm{Tailen}\binits{T.}} \AND
  \bauthor{\bsnm{Eubank},~\bfnm{Randall}\binits{R.}}
(\byear{2015}).
\btitle{Theoretical foundations of functional data analysis, with an
  introduction to linear operators}.
\bseries{Wiley Series in Probability and Statistics}.
\bpublisher{John Wiley \& Sons, Ltd., Chichester}.
\bdoi{10.1002/9781118762547}
\bmrnumber{3379106}
\end{bbook}
\endbibitem

\bibitem[\protect\citeauthoryear{Jolliffe}{2002}]{Joll02}
\begin{bbook}[author]
\bauthor{\bsnm{Jolliffe},~\bfnm{I.~T.}\binits{I.~T.}}
(\byear{2002}).
\btitle{Principal component analysis},
\bedition{second} ed.
\bseries{Springer Series in Statistics}.
\bpublisher{Springer-Verlag, New York}.
\bmrnumber{2036084 (2004k:62010)}
\end{bbook}
\endbibitem

\bibitem[\protect\citeauthoryear{Li and Hsing}{2007}]{LH07}
\begin{barticle}[author]
\bauthor{\bsnm{Li},~\bfnm{Yehua}\binits{Y.}} \AND
  \bauthor{\bsnm{Hsing},~\bfnm{Tailen}\binits{T.}}
(\byear{2007}).
\btitle{On rates of convergence in functional linear regression}.
\bjournal{J. Multivariate Anal.}
\bvolume{98}
\bpages{1782--1804}.
\bdoi{10.1016/j.jmva.2006.10.004}
\bmrnumber{2392433 (2009e:62174)}
\end{barticle}
\endbibitem

\bibitem[\protect\citeauthoryear{Marx and Eilers}{1999}]{marx1999}
\begin{barticle}[author]
\bauthor{\bsnm{Marx},~\bfnm{Brian~D}\binits{B.~D.}} \AND
  \bauthor{\bsnm{Eilers},~\bfnm{Paul~HC}\binits{P.~H.}}
(\byear{1999}).
\btitle{Generalized linear regression on sampled signals and curves: a P-spline
  approach}.
\bjournal{Technometrics}
\bvolume{41}
\bpages{1--13}.
\end{barticle}
\endbibitem

\bibitem[\protect\citeauthoryear{Ramsay and Dalzell}{1991}]{ramsay1991}
\begin{barticle}[author]
\bauthor{\bsnm{Ramsay},~\bfnm{James~O}\binits{J.~O.}} \AND
  \bauthor{\bsnm{Dalzell},~\bfnm{CJ}\binits{C.}}
(\byear{1991}).
\btitle{Some tools for functional data analysis}.
\bjournal{Journal of the Royal Statistical Society. Series B (Methodological)}
\bpages{539--572}.
\end{barticle}
\endbibitem

\bibitem[\protect\citeauthoryear{Ramsay and Silverman}{2005}]{RS05}
\begin{bbook}[author]
\bauthor{\bsnm{Ramsay},~\bfnm{Jim~O.}\binits{J.~O.}} \AND
  \bauthor{\bsnm{Silverman},~\bfnm{Bernard~W.}\binits{B.~W.}}
(\byear{2005}).
\btitle{Functional data analysis},
\bedition{second} ed.
\bseries{Springer Series in Statistics}.
\bpublisher{Springer}, \baddress{New York}.
\bmrnumber{2168993}
\end{bbook}
\endbibitem

\bibitem[\protect\citeauthoryear{Tikhonov and Arsenin}{}]{tikhonov1977}
\begin{bbook}[author]
\bauthor{\bsnm{Tikhonov},~\bfnm{Andrei}\binits{A.}} \AND
  \bauthor{\bsnm{Arsenin},~\bfnm{Vasilii}\binits{V.}}
\btitle{Solutions of ill-posed problems}.
\end{bbook}
\endbibitem

\bibitem[\protect\citeauthoryear{Yao, Muller and Wang}{2005}]{yao2005}
\begin{barticle}[author]
\bauthor{\bsnm{Yao},~\bfnm{Fang}\binits{F.}},
  \bauthor{\bsnm{Muller},~\bfnm{Hans-Georg}\binits{H.-G.}} \AND
  \bauthor{\bsnm{Wang},~\bfnm{Jane-Ling}\binits{J.-L.}}
(\byear{2005}).
\btitle{Functional linear regression analysis for longitudinal data}.
\bjournal{Ann. Statist.}
\bvolume{33}
\bpages{2873--2903}.
\bdoi{10.1214/009053605000000660}
\end{barticle}
\endbibitem

\bibitem[\protect\citeauthoryear{Yuan and Cai}{2010}]{YC10}
\begin{barticle}[author]
\bauthor{\bsnm{Yuan},~\bfnm{Ming}\binits{M.}} \AND
  \bauthor{\bsnm{Cai},~\bfnm{T.~Tony}\binits{T.~T.}}
(\byear{2010}).
\btitle{A reproducing kernel {H}ilbert space approach to functional linear
  regression}.
\bjournal{Ann. Statist.}
\bvolume{38}
\bpages{3412--3444}.
\bdoi{10.1214/09-AOS772}
\bmrnumber{2766857 (2012b:62237)}
\end{barticle}
\endbibitem

\end{thebibliography}

\end{document}